\documentclass[12pt]{amsart}
\usepackage{amssymb}
\usepackage{graphicx}
\usepackage{xcolor} 
\usepackage{tensor}

\usepackage{fullpage} 

\definecolor{green}{rgb}{0,0.8,0} 

\newtheorem{theorem}{Theorem}[section]
\newtheorem{corollary}[theorem]{Corollary}
\newtheorem{lemma}[theorem]{Lemma}
\newtheorem{proposition}[theorem]{Proposition}
\theoremstyle{definition}
\newtheorem{definition}[theorem]{Definition}

\theoremstyle{remark}
\newtheorem{remark}[theorem]{Remark}
\numberwithin{equation}{section}
\newcommand{\nrm}[1]{\Vert#1\Vert}
\newcommand{\abs}[1]{\vert#1\vert}
\newcommand{\brk}[1]{\langle#1\rangle}
\newcommand{\set}[1]{\{#1\}}

\newcommand{\les}{\lesssim}

\newcommand{\aleq}{\lesssim}
\newcommand{\ageq}{\gtrsim}

\newcommand{\lap}{\triangle}

\newcommand{\ud}{\mathrm{d}}
\newcommand{\rd}{\partial}
\newcommand{\nb}{\nabla}

\newcommand{\bb}{\Big}

\newcommand{\alp}{\alpha}
\newcommand{\bt}{\beta}
\newcommand{\gmm}{\gamma}
\newcommand{\Gmm}{\Gamma}
\newcommand{\dlt}{\delta}

\newcommand{\eps}{\epsilon}
\newcommand{\ep}{\epsilon}

\newcommand{\lmb}{\lambda}

\newcommand{\sgm}{\sigma}
\newcommand{\Sgm}{\Sigma}
\newcommand{\tht}{\theta}

\newcommand{\omg}{\omega}
\newcommand{\Omg}{\Omega}

\newcommand{\bfC}{{\bf C}}

\newcommand{\bfJ}{{\bf J}}

\newcommand{\bbR}{\mathbb R}
\newcommand{\bbS}{\mathbb S}

\newcommand{\bbZ}{\mathbb Z}

\newcommand{\calD}{\mathcal D}

\newcommand{\calM}{\mathcal M}

\newcommand{\calQ}{\mathcal Q}

\newcommand{\calS}{\mathcal S}

\newcommand{\calU}{\mathcal U}

\newcommand{\PD}{\calQ}

\newcommand{\uC}{\underline{C}}

\newcommand{\dur}{\nu}
\newcommand{\dvr}{\lmb}

\newcommand{\pfstep}[1]{\vspace{.5em} \noindent {\it #1.}}

\newcommand{\supp}{\mathrm{supp}\,}
\newcommand{\rst}{\upharpoonright}
\newcommand{\f}{\frac}

\newcommand{\dc}{\gmm}

\newcommand{\esssup}{\mathop{\mathrm{ess\,sup}}}

\begin{document}

\title[]{Solutions to the Einstein-scalar-field system in spherical symmetry with large bounded variation norms}
\author{Jonathan Luk}
\address{DPMMS, Cambridge University, Cambridge, UK}
\email{jluk@dpmms.cam.ac.uk}

\author{Sung-Jin Oh}
\address{Department of Mathematics, UC Berkeley, Berkeley, CA 94720, USA}
\email{sjoh@math.berkeley.edu}

\author{Shiwu Yang}
\address{DPMMS, Cambridge University, Cambridge, UK}
\email{S.Yang@dpmms.cam.ac.uk}


\begin{abstract}
It is well-known that small, regular, spherically symmetric characteristic initial data to the Einstein-scalar-field system which are decaying towards (future null) infinity give rise to solutions which are foward-in-time global (in the sense of future causal geodesic completeness). We construct a class of spherically symmetric solutions which are global but the initial norms are consistent with initial data \underline{not} decaying towards infinity. This gives the following consequences:
\begin{enumerate}
\item We prove that there exist foward-in-time global solutions with arbitrarily large (and in fact infinite) initial bounded variation (BV) norms and initial Bondi masses. 
\item While general solutions with non-decaying data do not approach Minkowski spacetime, we show using the results of \cite{LO1} that if a sufficiently strong asymptotic flatness condition is imposed on the initial data, then the solutions we construct (with large BV norms) approach Minkowski spacetime with a sharp inverse polynomial rate. 
\item Our construction can be easily extended so that data are posed at past null infinity and we obtain solutions with large BV norms which are causally geodesically complete both to the past and to the future. 
\end{enumerate}
Finally, we discuss applications of our method to construct global solutions for other nonlinear wave equations with infinite critical norms.
\end{abstract}
\maketitle

\section{Introduction}

We study the Einstein-scalar-field system for a Lorentzian manifold $(\mathcal M, g)$ and a real-valued function $\phi:\mathcal M\to \mathbb R$:
\begin{equation}\label{ESS}
\begin{cases}
Ric_{\mu\nu}-\f 12 R g_{\mu\nu}=2T_{\mu\nu},\\
T_{\mu\nu}=\rd_\mu\phi\rd_\nu\phi-\f 12 g_{\mu\nu}(g^{-1})^{\alp\bt}\rd_\alp\phi\rd_\bt\phi,\\
\Box_g \phi=0
\end{cases}
\end{equation}
in $(3+1)$ dimensions with \emph{spherically symmetric data}. It is known \cite{Christodoulou:1986, Christodoulou:1993bt} that small, regular and sufficiently decaying initial data give rise to foward-in-time global solutions in the sense that they are future causally geodesically complete. In this paper, we show that the decay condition can be removed and be replaced by the requirement that the growth of the integral of the data is suitably mild at infinity (see Theorem~\ref{thm:main.finite}). As a particular consequence, we construct global solutions with arbitrarily large (and in fact infinite) BV norms and Bondi masses\footnote{See definition in \eqref{BV.def} and \eqref{Bondi.def}.
}




To further discuss our results, we recall the reduction of \eqref{ESS} in spherical symmetry. It is well-known that in spherical symmetry we can introduce null coordinates $(u,v)$ such that the metric $g$ takes the form
$$g=-\Omg^2 \ud u \cdot \ud v+r^2 \ud \sigma_{\mathbb S^2},$$
where $\ud \sigma_{\mathbb S^2}$ is the standard metric on the unit round sphere and $r$ is the area-radius of the orbit of the symmetry group $SO(3)$ (see Section \ref{setup}). We normalize the coordinates so that $u = v$ on the axis of symmetry $\Gmm = \set{r = 0}$. Defining the Hawking mass $m$ by the relation
\begin{equation} \label{eq:hawking-m}
\Omg^{2} = - \frac{4 \rd_{u} r \rd_{v} r}{1 - \frac{2m}{r}},
\end{equation}
the Einstein-scalar field system reduces to the following system of equations for $(r,\phi, m)$ in $(1+1)$ dimensions, which we will also call the \emph{spherically symmetric Einstein-scalar-field} (SSESF) system:
\begin{equation} \label{eq:SSESF} \tag{SSESF}
\left\{
\begin{aligned}
	\rd_{u} \rd_{v} r = & \frac{2m \rd_{u} r \rd_{v}r }{(1-\frac{2m}{r}) r^{2}}, \\
	\rd_{u} \rd_{v} (r \phi) = & \frac{2m \rd_{u} r \rd_{v}r }{(1-\frac{2m}{r}) r^{2}} \phi, \\
	\rd_{u} r \rd_{u} m = & \frac{1}{2} (1-\f{2m}{r}) r^{2} (\rd_{u} \phi)^{2} , \\
	\rd_{v} r \rd_{v} m = & \frac{1}{2} (1-\f{2m}{r}) r^{2} (\rd_{v} \phi)^{2} .
\end{aligned}
\right.
\end{equation}
We will consider solutions to \eqref{eq:SSESF} via studying the characteristic initial value problem for which initial data are posed on a constant $u$ curve $C_{u_0}:=\{(u,v):u=u_0, \, v\geq u_0\}$. On $C_{u_0}$, after imposing the gauge condition $(\rd_v r)(u_0,v)=\f 12$ and the boundary conditions $r(u_0,u_0)=m(u_0,u_0)=0$, the initial value for $\Phi(v):=2(\rd_v(r\phi))(u_0,v)$ can be freely prescribed. It is easy to show that if $\Phi(v)$ is $C^1$, then there exists a unique local solution to \eqref{eq:SSESF}. We refer the readers to Section \ref{setup} for further discussions on the characteristic initial value problem.

As mentioned earlier, \eqref{eq:SSESF} is known to have global solutions for small, regular and decaying initial data. More precisely, Christodoulou \cite{Christodoulou:1986} showed that there exists a universal constant $\delta_0>0$ such that if
\begin{equation}\label{small.data.cond.1}
\sup_v\left((1+r)^3|\Phi|(v)+(1+r)^4|\rd_v\Phi|(u_0,v)\right)\leq \delta_0,
\end{equation}
then the solution is foward-in-time global. An analogous small data result in fact holds without assuming spherical symmetry as long as the higher derivatives of the scalar field and appropriate geometric quantities are also small and decaying. In the vacuum case, this was first proved by Christodoulou-Klainerman \cite{CK}. An alternative proof was later given by Lindblad-Rodnianski \cite{LR}, who also treated the case of the Einstein-scalar-field system.

Returning to the special case of spherical symmetry, in fact a much stronger result is known: Christodoulou showed in \cite{Christodoulou:1993bt} that only the bounded variation (BV) norm of the initial data $\Phi$ is required to be small\footnote{In \cite{Christodoulou:1993bt}, the initial data for $\Phi$ are in fact allowed to be in BV. In this paper, however, while we will use the BV norm as a measure of the size of the initial data, we will only consider initial data such that $\Phi$ is at least a $C^1$ function, in which case the BV norm is equivalent to the norm in \eqref{small.data.cond.2}.}, i.e., there exists a universal constant $\delta_1>0$ such that if\footnote{Notice that for $\delta_0$ sufficiently small, initial data satisfying \eqref{small.data.cond.1} obviously also obey \eqref{small.data.cond.2}.}
\begin{equation}\label{small.data.cond.2}
\int_{u_0}^\infty |\rd_v\Phi|(v')\, dv'\leq \delta_1,
\end{equation}
then the solution is global toward the future.

On the other hand, in the large data regime, Christodoulou showed in \cite{Christodoulou:1991} that not all initial data give rise to future causally geodesically complete solutions. In particular, for some class of initial data, the future Cauchy development contains a black hole region and is future causally geodesically incomplete.

The purpose of this paper is to construct a class of solutions which on one hand are global (in the sense of future causal geodesic completeness), but on the other hand their initial data are non-decaying and therefore large when measured using an integrated norm\footnote{On the other hand, we emphasize that the initial data that we allow in the main theorem are in fact small in a pointwise sense.}. One way to measure the size of the initial data is by the BV norm
\begin{equation}\label{BV.def}
\int_{u_0}^\infty |\rd_v\Phi|(v')\, dv',
\end{equation}
which is a scaling-invariant quantity and as mentioned above, the smallness of the BV norm guarantees that the solution is global. We will also quantify the largeness of the initial data by the initial Bondi mass, which is defined as the limit of the Hawking mass as $v\to \infty$ on the initial curve, i.e.,
\begin{equation}\label{Bondi.def}
M_{u_0}:=\lim_{v\to \infty} m(u_0,v).
\end{equation}
In fact, our construction allows both the initial BV norm and Bondi mass to be infinite. More precisely, the following is the main result of this paper:
\begin{theorem} \label{thm:main.finite}
Consider the characteristic initial value problem from an outgoing curve $C_{u_0}$ with $v\geq u_0$, $\rd_v r \rst_{C_{u_0}} = \frac{1}{2}$ and $r(u_0,u_0)=m(u_0,u_0)=0$.
Suppose the data on the initial curve $C_{u_0}$ is given by
\[
2 \rd_{v}(r \phi)(u_0, v)=\Phi(v),
\]
where $\Phi:[u_0,\infty)\to\mathbb R$ is a smooth function satisfying the following conditions for some $\dc>0$:
\begin{equation}
 \label{eq:IDcond}
\int_{u}^{v}|\Phi(v')|\, \ud v'\leq \ep(v-u)^{1-\dc}, \quad |\Phi(v)|+|\Phi'(v)|\leq \ep, \quad \forall v\geq u\geq u_0
\end{equation}
Then there exists $\eps>0$ depending only on $\dc$ such that the unique solution to \eqref{eq:SSESF} arising from the given data is future causally geodesically complete. Moreover, the solution satisfies the following uniform a priori estimates:
\begin{equation} \label{eq:main:geo}
	\rd_v r>\f 13,\quad -\f 16>\rd_u r >-\f 23,\quad \f{2m}{r}<\f 12, \\
\end{equation}
and
\begin{equation} \label{eq:main:apriori}
\begin{gathered}
	\abs{\phi} \leq C \eps \min \set{1, r^{-\gamma}}, \quad
	\abs{\rd_{v} (r \phi)} \leq C ( \abs{\Phi(v)} + \eps \min \set{1, r^{-\gamma}}), \quad
	\abs{\rd_{u} (r \phi)} \leq C \eps, \\
\abs{\rd_{v}^{2} (r \phi)} + \abs{\rd_{v}^{2} r} + \abs{\rd_{u}^{2} (r \phi)} + \abs{\rd_{u}^{2} r} \leq C \eps
\end{gathered}
\end{equation}
for some constant $C>0$ depending only on $\dc$.
\end{theorem}
\begin{remark}
We note explicitly that the constants $\ep$ and $C$ in the above theorem are independent of $u_0$.
\end{remark}
\begin{remark}
In addition, if the second derivative of the data $\Phi$ is bounded (i.e., $\abs{\Phi''(v)} \leq C$), we will show that the solution obeys corresponding higher regularity bounds which are uniform with respect to $u_{0}$; see Proposition~\ref{prop:high-d}.
\end{remark}
The proof of this theorem will occupy most of this paper. Global existence of a unique solution in an appropriate coordinate system will be established in Section \ref{sec.proof}.
As a brief comment on the proof,
we note that even for spherically symmetric solutions to the \emph{linear wave equation} $\Box_{\mathbb R^{1+3}}\phi=0$, if the initial data on an outgoing null cone $C_{u_0}$ are only required to satisfy \eqref{eq:IDcond}, then the solutions $\phi$ and its first derivatives in general \emph{do not decay in time}. In fact, $\phi$ only satisfies decay estimates in $r$. In our setting, since we can use the method of characteristics in spherical symmetry, we will only need to integrate the error terms along null curves and the $r$ decay is therefore already sufficient to close a nonlinear problem.

Since our solution does not decay in time, even after establishing global existence of the solution in an appropriate coordinate system, it does not immediately follow that the solution is future causally geodesically complete. For this we need an additional geometric argument, which will be carried out in Section \ref{sec:cgc}.


An immediate corollary of Theorem \ref{thm:main.finite} is the following result, which follows simply from the observation that there exists $\Phi$ satisfying the assumptions of Theorem \ref{thm:main.finite} such that \eqref{BV.def} and \eqref{Bondi.def} are both infinite. The proof of this corollary will be carried out in Section \ref{sec.infinite.BV.mass}.
\begin{corollary}\label{cor.infinite.BV.mass}
There exist solutions to \eqref{ESS} with spherically symmetric data such that the data have arbitrarily large (in fact infinite) BV norm and initial Bondi mass, while the development is future causally geodesically complete.
\end{corollary}

We now briefly describe some generalizations and consequences of Theorem \ref{thm:main.finite}, but we will refer the readers to Sections \ref{sec.quantitative} and \ref{sec.past} for more details. First, while Theorem \ref{thm:main.finite} itself does not show that the solution ``decays'' while approaching timelike infinity\footnote{Since the solution to the \emph{linear} wave equation with data satisfying the assumptions of Theorem \ref{thm:main.finite} does not decay, general solutions constructed in Theorem \ref{thm:main.finite} in fact may not decay in time.}, if we assume in addition a sufficiently strong asymptotic flatness condition, then we can apply the results in \cite{LO1} to show that solution satisfies a pointwise inverse polynomial decay rate. In fact, Theorem \ref{thm:main.finite} also gives the first examples of solutions with large initial BV norms which satisfy the assumptions of the \emph{conditional} decay result in \cite{LO1}. As a consequence of the pointwise decay, (a subclass of) these solutions are also stable with respect to small but not necessarily spherically symmetric perturbations \cite{LO2}. See further discussion in Section \ref{sec.quantitative}.

Second, a consequence of the proof of Theorem \ref{thm:main.finite} is the construction of a large data spacetime solution which is causally geodescally complete \emph{both to the future and the past}. This is achieved by making use of the uniformity of the estimates of Theorem \ref{thm:main.finite} in $u_0$ and taking the limit $u_0\to -\infty$. We refer the readers to Theorem \ref{thm:main} in Section \ref{sec.past} for a precise statement.

Before turning to further discussions of our results, we note that the problem of constructing large data global solutions to supercritical nonlinear wave equations has attracted much recent attention. We refer the readers to \cite{BeSo, KrSc, WaYu, Yang} and the references therein for some recent results. The ideas in the work \cite{WaYu} in particular is inspired by the monumental work of Christodoulou \cite{Chr} in general relativity on the formation of trapped surfaces, which is itself a semi-global\footnote{in the sense that the large data solution constructed in \cite{Chr} is global towards past null infinity.} large data result. Our result appears to be the first in which the large data solutions are global \emph{both to the future and to the past}. As an example of the robustness of our methods, we also consider the much simpler supercritical semilinear equation
\begin{equation} \label{eq:NLW}
	\Box_{\bbR^{1+3}} \phi = \pm \phi^{7}.
\end{equation}
We show that \eqref{eq:NLW} admits solutions with infinite $\dot{H}^{\f76}\times \dot{H}^{\f 16}$ norms for all time and are global both to the future and to the past (see Theorem \ref{thm:nlw}), thus extending\footnote{We emphasize however that the solutions we construct are in different regimes compared to \cite{KrSc} and \cite{BeSo}. See Remark~\ref{rem:nlw} for a more detailed comparison.} the results of Krieger-Schlag \cite{KrSc} and Beceanu-Soffer \cite{BeSo}. 

\subsection{Quantitative decay rates and nonlinear stability}\label{sec.quantitative}

In general, the solutions that are constructed in Theorem \ref{thm:main.finite} may not exhibit uniform decay in $v$. Nevertheless, in this section we show that if one imposes the following strong asymptotic flatness condition on the $C^1$ initial data:
\begin{equation}\label{AF}
\limsup_{v\to \infty}\left( (1+v)^{\omg'} \abs{\Phi}(v) + (1+v)^{\omg'+1} \abs{\rd_{v}\Phi}(v)\right) \leq A_0 < \infty
\end{equation}
for $\omg'>1$, then in fact the solution decays\footnote{Notice that if \eqref{AF} holds, the \emph{linear} solution obviously decays.} in $v$.

To see this, we apply the result in \cite{LO1} by the first two authors. In \cite{LO1}, the long time asymptotics of spherically symmetric, causally geodesically complete solutions to the Einstein-scalar field system was studied. It was shown that (see Theorems~3.1 and 3.2 and Remark 3.9 in \cite{LO1}) sharp pointwise inverse polynomial decaying bounds hold for the solution \emph{even for large initial data}, as long as the solution is \emph{assumed} to satisfy the bound\footnote{In \cite{LO1}, it was shown that alternatively,
$$\sup_{u}\int_{C_u\cap\{r\leq R\}} \left(|\rd_v^2(r\phi)|(u,v)+(\rd_v r)^{-1}|\rd_v^2 r|(u,v)\right) \ud v\to 0$$
is sufficient to guarantee that the inverse polynomial decaying bounds hold. We cite \eqref{AF} instead as it is more convenient to apply in the proof of Theorem \ref{thm:decay}.}
\begin{equation}\label{LO1.assumption}
\sup_{u}\int_{C_u\cap\{r\leq R\}} \left(|\rd_v^2(r\phi)|^p(u,v)+(\rd_v r)^{-p}|\rd_v^2 r|^p(u,v)\right) \ud v\leq C
\end{equation}
for some $p>1$, $R>0$ and $C>0$. Here, and below, we use the convention that $C_u$ denote a constant $u$ curve.

In \cite{LO1}, it was further proved using the work of Christodoulou \cite{Christodoulou:1993bt} that the pointwise decay results holds if the initial data obey \eqref{AF} and have small BV norm. On the other hand, the work \cite{LO1} leaves open the question whether there exist \emph{any} solutions with large BV norm that satisfy both \eqref{AF} and \eqref{LO1.assumption}. Our present work provides a construction of such spacetimes. More precisely, we have
\begin{theorem}\label{thm:decay}
Assume, in addition to the assumptions of Theorem \ref{thm:main.finite}, that $\Phi$ obeys the following bounds for some $A_0>0$ and $\omg'>1$:
$$\sup_{v\in [u_0,\infty]}\left( (1+v)^{\omg'} \abs{\Phi}(v) + (1+v)^{\omg'+1} \abs{\rd_{v}\Phi}(v)\right) \leq A_0 < \infty.$$
Then the following decay estimates hold for $\omg:=\min \set{\omg', 3}$ and for some $A_1>0$:
\begin{align}
	\abs{\phi} \leq & A_{1} \min \set{u^{-\omg}, r^{-1} u^{-(\omg-1)}}, \label{eq:decay1:1} \\
	\abs{\rd_{v}(r \phi)} \leq & A_{1} \min \set{u^{-\omg}, r^{-\omg}}, \label{eq:decay1:2} \\
	\abs{\rd_{u} (r \phi)} \leq & A_{1} u^{-\omg}, \label{eq:decay1:3}\\
	\abs{\rd_{v}^{2} (r \phi)} \leq & A_{1} \min \set{u^{-(\omg+1)}, r^{-(\omg+1)}}, \label{eq:decay2:1} \\
	\abs{\rd_{u}^{2} (r \phi)} \leq & A_{1} u^{-(\omg+1)}, \label{eq:decay2:2} \\
	\abs{\rd_{v}^2 r} \leq & A_{1} \min \set{u^{-3}, r^{-3}}, \label{eq:decay2:3} \\
	\abs{\rd_{u}^2 r} \leq & A_{1} u^{-3}. \label{eq:decay2:4}
\end{align}
\end{theorem}
\begin{proof}
By Theorems 3.1 and 3.2 and Remark 3.9 in \cite{LO1}, the desired decay rates hold if for some $p>1$, $R>0$ and $C>0$, we have
$$\sup_{u}\int_{C_u\cap\{r\leq R\}} \left(|\rd_v^2(r\phi)|^p(u,v)+(\rd_v r)^{-p}|\rd_v^2 r|^p(u,v)\right) \ud v \leq C.$$
This latter bound indeed holds (for any $p>1$ and any $R>0$) in view of the estimates \eqref{eq:main:geo} and \eqref{eq:main:apriori} in Theorem \ref{thm:main.finite}.
\end{proof}

\begin{remark}
As already noted in \cite{LO1}, the decay rates that are obtained in Theorem \ref{thm:decay} are sharp.
\end{remark}

Given these decay rates, it seems natural to ask whether the solutions constructed in Theorem \ref{thm:decay} are stable with respect to small perturbations \emph{even outside spherical symmetry}. This question will be addressed in a forthcoming paper \cite{LO2} in which we answer this question in the affirmative\footnote{To be precise, this holds for a subclass of the solutions constructed in Theorem \ref{thm:decay}. In particular, some (not necessarily small) higher derivative bounds \emph{for the initial data} are also required.}, thus extending the proof of the nonlinear stability of Minkowski spacetime to a more general class of dispersive spacetimes.

\subsection{Large data solutions which are both future and past complete}\label{sec.past}

Theorem \ref{thm:main.finite} constructs future causally geodesically complete solutions to the future of the hypersurface $C_{u_0}$. On the other hand, a priori estimates in Theorem \ref{thm:main.finite} are independent of $u_0$. One can therefore\footnote{To be precise, in order to justify this procedure, one need an additional assumption on the second derivative of $\Phi$ (see \eqref{thm:main.add.assumption} below) as well as potentially taking $\ep>0$ to be smaller. We explicitly note, however, that the smaller of $\ep$ is \emph{independent} of the second derivative bounds of $\Phi$.} take $u_0\to -\infty$ and obtain solutions to \eqref{eq:SSESF} for $(u,v)\in \set{(u, v) : -\infty < u < \infty, \ u \leq v < \infty}$. The solutions constructed in this manner are moreover causally geodesically complete both towards the future and the past. More precisely, we have the following theorem:
\begin{theorem}\label{thm:main}
Let $\Phi:\mathbb R\to \mathbb R$ be a smooth function such that \eqref{eq:IDcond} holds for some $\gamma>0$, i.e.,
\begin{equation*}
\int_{u}^{v}|\Phi(v')|\, \ud v'\leq \ep(v-u)^{1-\dc}, \quad |\Phi(v)|+|\Phi'(v)|\leq \eps, \quad \forall -\infty<u\leq v<\infty
\end{equation*}
and
\begin{equation}\label{thm:main.add.assumption}
\sup_v|\Phi''(v)|<\infty.
\end{equation}
Then for every $\dc>0$, there exists $\ep>0$ such that there exists a solution to \eqref{eq:SSESF} which is both future and past casually geodesically complete and obeys
$$\lim_{u\to -\infty} 2\rd_v(r\phi)(u,v)=\Phi(v),\quad \lim_{u\to -\infty} \rd_{v} r(u,v)=\f 12.$$
\end{theorem}

The proof of global existence in a suitable double null coordinate system will be carried out in Section~\ref{sec.proof.inf}. Future and past causal geodesic completeness of the solution will be proved in Section~\ref{sec:cgc}.

We contrast Theorem \ref{thm:main} with the works \cite{WaYu, Yang} on large solutions to nonlinear wave equations. In \cite{WaYu, Yang}, the key idea, inspired by \cite{Chr}, is to construct solutions which are large but ``sufficiently outgoing''. For instance, on an initial Cauchy hypersurface $\{t=0\}$, this means that $\rd_v\phi$ is appropriately small, while $\rd_u\phi$ is allowed to be large. This approach, while useful to obtain a global solution to the future, does not seem applicable to construct solutions which are global both to the future and to the past as in Theorem \ref{thm:main}. On the other hand, we also note that the work \cite{WaYu} allows the initial data to be large in a compact region in space, whereas in Theorem \ref{thm:main}, the ``largeness'' of the initial data is only achieved by the lack of decay at infinity.

\subsection{Outline of the paper}
We end the introduction with an outline of the remainder of the paper. In Section \ref{sec.prelim}, we will discuss some preliminaries, including the geometric setup and some identities that we will repeatedly use. The main theorem (Theorem \ref{thm:main.finite}) will then be proved in Section \ref{sec.proof}, modulo the assertion that the resulting spacetime is future causally geodesically complete. Then using the estimates obtained in Section \ref{sec.proof}, we will prove Theorem \ref{thm:main} in Section \ref{sec.proof.inf}, again modulo causal geodesic completeness. In Section \ref{sec.infinite.BV.mass}, we will then return to the proof of Corollary \ref{cor.infinite.BV.mass}. In Section~\ref{sec:cgc}, we finally complete the proof of Theorems~\ref{thm:main.finite} and \ref{thm:main} by establishing the causal geodesic completeness statements. Lastly, in Appendix \ref{appendix}, we will apply the methods in this paper to study the equation $\Box_{\mathbb R^{1+3}}\phi=\pm\phi^7$, and show that there exists solutions global to the future and the past which have infinite critical $\dot H^{\f 76} \times \dot H^{\f 16}$ norm and infinite critical Strichartz norm.

\subsection*{Acknowledgments}
S.-J. Oh is a Miller Research Fellow, and thanks the Miller Institute for support.

\section{Preliminaries}\label{sec.prelim}
In this section, we further explain the geometric setup of the problem and introduce the notation that we will use for the rest of the paper.

\subsection{Setup}\label{setup}

As discussed in the introduction, \eqref{eq:SSESF} arises as a reduction of the (3+1)-dimensional Einstein-scalar field equation under spherical symmetry, written in a double null coordinate system. Here we describe \eqref{eq:SSESF} as a (1+1)-dimensional system, which is the point of view we adopt in our analysis throughout this paper until Section~\ref{sec:cgc}.

Consider the $(1+1)$-dimensional domain
\begin{equation*}
	\PD = \set{(u, v) \in \bbR^{1+1} : u \in (-\infty, \infty), \, v \in [u, \infty)},
\end{equation*}
with partial boundary
\begin{equation*}
	\Gmm = \set{(u, u) \in \PD : u \in (-\infty, \infty)}.
\end{equation*}
We define causality in $\PD$ with respect to the ambient metric $m = - \ud u \cdot \ud v$ of $\bbR^{1+1}$, and the time orientation in $\PD$ so that $\rd_{u}$ and $\rd_{v}$ are future pointing. We use the notation $C_{u}$ and $\uC_{v}$ for constant $u$ and $v$ curves in $\PD$, respectively. We call $C_{u}$ an \emph{outgoing} null curve and $\uC_{v}$ as an \emph{incoming} null curves, in reference to their directions (to the future) relative to $\Gmm$.
Moreover, given $- \infty < u_{0} < u_{1} < \infty$, let
\begin{align*}
	\PD_{[u_{0}, u_{1}]} =& \set{(u, v) \in \PD : u \in [u_{0}, u_{1}]}, \\
	\PD_{[u_{0}, \infty)} =& \set{(u, v) \in \PD : u \in [u_{0}, \infty)}.
\end{align*}

We introduce the notion of a $C^{k}$ solution to \eqref{eq:SSESF} as follows.
\begin{definition} \label{def:C1-sol}
Let $- \infty < u_{0} < u_{1} < \infty$. We say that a triple $(r, \phi, m)$ of real-valued functions on $\PD_{[u_{0}, u_{1}]}$ is a \emph{$C^{k}$ solution to \eqref{eq:SSESF}} if it satisfies this system of equations and the following conditions hold:
\begin{enumerate}
\item The following functions are $C^{k}$ in $\PD_{[u_{0}, u_{1}]}$:
\begin{equation*}
	\rd_{u} r, \ \rd_{v} r, \ \phi, \ \rd_{v}(r \phi), \ \rd_{u}(r \phi).
\end{equation*}
\item For $\rd_{v} r$ and $\rd_{u} r$, we have
\begin{equation*}
	\inf_{\PD_{[u_{0}, u_{1}]}} \rd_{u} r > -\infty, \quad \inf_{\PD_{[u_{0}, u_{1}]}} \rd_{v} r > 0.
\end{equation*}
\item For each point $(a, a) \in \Gmm \cap \PD_{[u_{0}, u_{1}]}$, the following boundary conditions hold:
\begin{align}
\label{eq:bc4r}
	r (a, a) =& 0, \\
\label{eq:bc4m}
	m(a, a) =& 0.
\end{align}
\end{enumerate}

Moreover, if $(r, \phi, m)$ is a $C^{k}$ solution on $\PD_{[u_{0}, u_{1}]}$ for every $u_{1}$ greater than $u_{0}$, then we say that it is a \emph{global $C^{k}$ solution on $\PD_{[u_{0}, \infty)}$}.
\end{definition}

The boundary condition \eqref{eq:bc4r} can be combined with the regularity assumption to deduce higher order boundary order conditions for $r$ and $r \phi$. More precisely, let $(r, \phi, m)$ be a $C^{k}$ solution on $\PD_{[u_{0}, u_{1}]}$. Since $u = v$ on $\Gmm = \set{ r = 0}$, we have
\begin{equation} \label{eq:bc4high-d}
	(\rd_{v} + \rd_{u})^{\ell} r (a, a) = 0, \quad
	(\rd_{v} + \rd_{u})^{\ell} (r \phi) (a, a) = 0,
\end{equation}
for every $\ell = 0, \ldots, k$ and $(a, a) \in \Gmm \cap \PD_{[u_{0}, u_{1}]}$.

Consider the characteristic initial value problem for \eqref{eq:SSESF} with data
\begin{equation} \label{eq:ini-dvrphi}
	(\rd_v r)^{-1} \rd_{v} (r \phi) \restriction_{C_{u_{0}}} = \Phi,
\end{equation}
and initial gauge condition\footnote{We call \eqref{eq:ini-dvr} an initial gauge condition since it can be enforced for an arbitrary initial data set by a suitable reparametrization of the coordinate $v$, which is a gauge symmetry of the problem. See Remark~\ref{rem:gauge}.}
\begin{equation} \label{eq:ini-dvr}
	(\rd_v r) \restriction_{C_{u_{0}}} = \frac{1}{2},
\end{equation}
on some outgoing null curve $C_{u_{0}}$. This problem is locally well-posed for $C^{k}$ data $(k \geq 1)$ in the following sense: Given any $C^{k}$ data $\Phi$ with $k \geq 1$, there exists a unique $C^{k}$ solution to \eqref{eq:SSESF} on $\PD_{[u_{0}, u_{1}]}$ for some $u_{1} > u_{0}$, which only depends on $u_{0}$ and the $C^{k}$ norm of $\Phi$. We omit the proof, which is a routine iteration argument using, for instance, the equations stated in Section~\ref{subsec:eqs} below.

\begin{remark} \label{rem:gauge}
The system \eqref{eq:SSESF} is invariant under reparametrizations of the form $(u, v) \mapsto ( U(u), V(v) )$; this is the \emph{gauge invariance} of \eqref{eq:SSESF}. Note that we have implicitly fixed a gauge in the setup above, by requiring that $u = v$ on $\Gmm$ and imposing the initial gauge condition \eqref{eq:ini-dvr}.
\end{remark}

\begin{remark} \label{rem:3+1d}
As discussed in the introduction, reduction of the Einstein-scalar field system under spherical symmetry yields the above (1+1)-dimensional setup, where $\Gmm$ corresponds to the axis of symmetry $\set{r = 0}$. Furthermore, the boundedness of the function $\nb^{\alp} r \nb_{\alp} r = 1 - \frac{2m}{r}$ on $\Gmm$ translates to the boundary condition $m = 0$ on $\Gmm$. Conversely, any suitably regular solution $(r, \phi, m)$ on $\PD_{0} \subseteq \PD$ gives rise to a spherically symmetric (3+1)-dimensional solution $(g, \phi)$ to the Einstein-scalar field system on $\calM = \PD_{0} \times \bbS^{2}$, where $g$ is as in the introduction.
\end{remark}

\begin{remark}
Finally, although it is stated slightly differently, it can be checked that the notion of $C^{1}$ solution in \cite{LO1} is equivalent to the present definition.
\end{remark}

\subsection{Structure of \eqref{eq:SSESF}} \label{subsec:eqs}

Following \cite{Christodoulou:1993bt}, we introduce the shorthands
\begin{equation*}
	\dvr = \rd_{v} r, \quad
	\dur = \rd_{u} r, \quad
	\mu = \frac{2 m }{r}.
\end{equation*}
These dimensionless quantities will play an important role in this paper, as they encode key geometric information about the spacetime.

In what follows, we will rewrite \eqref{eq:SSESF} using normalized derivatives $\dvr^{-1} \rd_{v}$ and $\dur^{-1} \rd_{u}$ instead of $\rd_{v}$ and $\rd_{u}$.
Unlike $\rd_{v}$ and $\rd_{u}$, these normalized derivatives are invariant under reparametrizations of $v$ and $u$. Moreover, it turns out that writing \eqref{eq:SSESF} in such a form leads to \emph{decoupling} of the evolutionary equations under mild assumptions on the quantities $\dvr$, $\dur$ and $\mu$, which is convenient for analysis; see Remark~\ref{rem:decouple} below for a more detailed discussion.

The wave equation for $\phi$, in terms of $\dvr^{-1} \rd_{v} \phi$ and $\dur^{-1} \rd_{u} \phi$, takes the form
\begin{align}
\label{eq:wave4dvrphi}
	\rd_{u} \bb( \dvr^{-1} \rd_{v} (r \phi) \bb) = & - \frac{2 m \dur}{(1-\mu) r^{2}} \bb( \dvr^{-1} \rd_{v} (r \phi) \bb) + \frac{2 m \dur}{(1-\mu) r^{2}} \phi, \\
\label{eq:wave4durphi}
	\rd_{v} \bb( \dur^{-1} \rd_{u} (r \phi) \bb) = & - \frac{2 m \dvr}{(1-\mu) r^{2}} \bb( \dur^{-1} \rd_{u} (r \phi) \bb) + \frac{2 m \dvr}{(1-\mu) r^{2}} \phi.
\end{align}
The wave equation for $r$, in terms of $\log \dvr$ and $\log \dur$, takes the form
\begin{align}
\label{eq:wave4dvr}
	\rd_{u} \log \dvr = & \frac{2 m \dur}{(1-\mu) r^{2}}, \\
\label{eq:wave4dur}
	\rd_{v} \log \dur = & \frac{2 m \dvr}{(1-\mu) r^{2}}.
\end{align}
The equations for the Hawking mass $m$ read
\begin{align}
\label{eq:dv-m}
	\dvr^{-1} \rd_{v} m = & \frac{1}{2} (1-\mu) r^{2} (\dvr^{-1} \rd_{v} \phi)^{2}, \\
\label{eq:du-m}
	\dur^{-1} \rd_{u} m = & \frac{1}{2} (1-\mu) r^{2} (\dur^{-1} \rd_{u} \phi)^{2}.
\end{align}
Moreover, the following \emph{Raychaudhuri equations} can be derived from \eqref{eq:SSESF}:
\begin{align}
	\label{eq:raych4v}
	\dvr^{-1} \rd_{v} \log \abs{\frac{\dur}{1-\mu}} = & r (\dvr^{-1} \rd_{v} \phi)^{2}, \\
	\label{eq:raych4u}
	\dur^{-1} \rd_{u} \log \abs{\frac{\dvr}{1-\mu}} = & r (\dur^{-1} \rd_{u} \phi)^{2}.
\end{align}
By the wave equation for $r$, we also have the commutator formulae
\begin{align}
\label{eq:comm-du-dv}
	[\rd_{u}, \dvr^{-1} \rd_{v}] = & - \frac{2 m \dur}{(1-\mu) r^{2}} \dvr^{-1} \rd_{v}, \\
\label{eq:comm-dv-du}
	[\rd_{v}, \dur^{-1} \rd_{u}] = & - \frac{2 m \dvr}{(1-\mu) r^{2}} \dur^{-1} \rd_{u}.
\end{align}

\begin{remark} \label{rem:decouple}
Once we have a suitable control of the underlying geometry, namely upper and lower bounds for $\dvr, \dur$ and $(1-\mu)$, the evolutionary equations \eqref{eq:wave4dvrphi}, \eqref{eq:wave4durphi}, \eqref{eq:wave4dvr} and \eqref{eq:wave4dur} are essentially all decoupled from each other.
This observation allows us to close bounds for $\dvr^{-1} \rd_{v}$ derivatives of $r \phi$ first, and then derive bounds for other variables (such as $\log \dvr$, $\dur^{-1} \rd_{u}(r \phi)$ and $\log \dur$) afterwards.
Moreover, from \eqref{eq:wave4dvrphi}, \eqref{eq:wave4dvr} and \eqref{eq:comm-du-dv}, it is clear that a key step in propagating the incoming waves $\dvr^{-1} \rd_{v} (r \phi)$ and $\log \dvr$ is to control the factor $\frac{2 m \dur}{(1-\mu) r^{2}}$. Similarly, controlling $\frac{2 m \dvr}{(1-\mu) r^{2}}$ is important for propagating the outgoing waves $\dur^{-1} \rd_{u} (r \phi)$ and $\log \dur$.
\end{remark}

Finally, for a $C^{k}$ solution $(r, \phi, m)$, note that the boundary conditions in \eqref{eq:bc4high-d} imply
\begin{equation} \label{eq:bc}
	(\dvr^{-1} \rd_{v} - \dur^{-1} \rd_{u})^{\ell} r(a, a) = 0, \quad
	(\dvr^{-1} \rd_{v} - \dur^{-1} \rd_{u})^{\ell} (r \phi)(a, a) = 0,
\end{equation}
for $\ell = 0, \ldots, k$ on the axis. These equations, along with the wave equations stated above, can be used to compute
$(\dur^{-1} \rd_{u})^{k} (r \phi)$ and $(\dur^{-1} \rd_{u})^{k-1} \log \dur$
in terms of $(\dvr^{-1} \rd_{v})^{k} (r \phi)$, $(\dvr^{-1} \rd_{v})^{k-1} \log \dvr$ and lower order terms.

\subsection{Averaging operators and commutation with $\dvr^{-1} \rd_{v}$}
Observe that a number of quantities in the nonlinearity of \eqref{eq:SSESF} are given in terms of averaging formulae.
For instance, by the boundary conditions \eqref{eq:bc4r}, \eqref{eq:bc4m} and the equation \eqref{eq:dv-m}, we have
\begin{align}
\label{eq:avg:phi}
	\phi (u, v) =& \frac{1}{r(u, v)} \int_{u}^{v} \dvr^{-1} \rd_{v} (r \phi) \, \dvr (u, v') \, \ud v', \\
\label{eq:avg:2m-over-r2}
	 \frac{2m}{r^{2}} (u, v)=& \frac{1}{r^{2}(u, v)} \int_{u}^{v} (1-\mu) r (\dvr^{-1} \rd_{v} \phi)^{2} \, r \dvr \, (u, v') \ud v'.
\end{align}

Motivated by these formulae, we define the $s$-order averaging operator $I_{s}$ on outgoing null curves by
\begin{equation*}
	I_{s}[f](u, v) = \frac{1}{r^{s}} \int_{u}^{v} f(v') \, r^{s-1} \dvr (u, v') \, \ud v'.
\end{equation*}
By pulling out $f$ outside the integral and using the fundamental theorem of calculus, we obtain the basic estimate
\begin{equation} \label{eq:avg-est}
	\abs{I_{s} [f](u, v)} \leq \frac{1}{s} \sup_{v' \in [u, v]} \abs{f}
\end{equation}
The averaging operator $I_{s}$ turns out to obey a nice differentiation formula with respect to $\dvr^{-1} \rd_{v}$.
\begin{lemma} \label{lem:avg-d}
For any real number $s \geq 1$, the following identity holds.
\begin{equation} \label{eq:avg-d}
	\dvr^{-1} \rd_{v} I_{s}[f](u, v) = I_{s + 1}[\dvr^{-1} \rd_{v} f](u, v).
\end{equation}
\end{lemma}
\begin{proof}
In what follows, we will often omit writing $u$, which is fixed throughout the proof.
Making the change of variable $\rho  = r^{s}(u, v)$ so that
\begin{equation*}
	s r^{s-1} \dvr \, \ud v = \ud \rho, \quad
	\dvr^{-1} \rd_{v} = s \rho^{\frac{s-1}{s}} \rd_{\rho},
\end{equation*}
we may rewrite $I_{s}[f]$ and $\dvr^{-1} \rd_{v} I_{s}[f]$ as
\begin{align*}
	I_{s}[f](\rho) =& \frac{1}{s \rho} \int_{0}^{\rho} f(\rho') \, \ud \rho', \\
	(\dvr^{-1} \rd_{v} I_{s}[f]) (\rho) =& \rho^{\frac{s-1}{s}} \rd_{\rho} \bb( \frac{1}{\rho} \int_{0}^{\rho} f(\rho') \, \ud \rho' \bb),
\end{align*}
where we abuse notation and write $I_{s}[f](\rho) = I_{s}[f](v(\rho))$, $f(\rho') = f(v(\rho'))$ etc.
Note that
\begin{equation*}
	\rd_{\rho} \bb(\frac{1}{\rho} \int_{0}^{\rho} f(\rho') \, \ud \rho' \bb) = \frac{1}{\rho^{2}} \int_{0}^{\rho} (\rd_{\rho} f)(\rho') \, \rho' \ud \rho'.
\end{equation*}
The previous identity follows quickly by, say, making a further change of variables $\sgm' = \rho' / \rho$. Plugging this in the expression for $\dvr^{-1} \rd_{v} I_{s}[f]$ and changing the variable back to $v$, we arrive at \eqref{eq:avg-d}. \qedhere
\end{proof}

Applying Lemma~\ref{lem:avg-d} to the formulae \eqref{eq:avg:phi} and \eqref{eq:avg:2m-over-r2}, we obtain
\begin{align}
	\dvr^{-1} \rd_{v} \phi (u, v) = & \frac{1}{r^{2}(u, v)} \int_{u}^{v} (\dvr^{-1} \rd_{v})^{2} (r \phi) \, r \dvr (u, v') \, \ud v', \label{eq:dvphi:avg} \\
	\dvr^{-1} \rd_{v} \bb( \frac{2m}{r^{2}} \bb) (u, v) = & \frac{1}{r^{3}(u, v)} \int_{u}^{v} (\dvr^{-1} \rd_{v}) \bb( (1-\mu) r (\dvr^{-1} \rd_{v} \phi)^{2} \bb) \, r^{2} \dvr (u, v') \, \ud v'.	\label{eq:dv-2m-over-r2}
\end{align}
Such differentiated averaging identities are useful near the axis. On the other hand, far away from the axis, it is more effective to simply commute $\dvr^{-1} \rd_{v}$ with $r$, as in the following identities:
\begin{align}
	r \, \dvr^{-1} \rd_{v} \phi =& \dvr^{-1} \rd_{v} (r \phi) - \phi,  \label{eq:dvphi:large-r}\\
\label{eq:dv-m-over-r2:large-r}
	r^{2} \, \dvr^{-1} \rd_{v} \bb( \frac{2m}{r^{2}} \bb)
	=& \dvr^{-1} \rd_{v} (2 m) - \frac{4m}{r}.
\end{align}

An entirely analogous discussion holds with the roles of $u$ and $v$ interchanged. Indeed, with the definition
\begin{equation} \label{eq:avg:conj}
	\underline{I}_{s}[g](u, v) = \frac{1}{r^{s}} \int_{u}^{v} g(u') \, r^{s-1} \dur(u', v) \, \ud u',
\end{equation}
the following analogue of Lemma~\ref{lem:avg-d} can be proved.
\begin{lemma} \label{lem:avg-d:conj}
For any real number $s \geq 1$, the following identity holds.
\begin{equation} \label{eq:avg-d:conj}
	\dur^{-1} \rd_{u} \underline{I}_{s}[g](u,v) = \underline{I}_{s+1}[\dur^{-1} \rd_{u} g](u,v).
\end{equation}
\end{lemma}
%

\section{Forward-in-time global solution}\label{sec.proof}
The main goal of this section is to establish Theorem~\ref{thm:main.finite} modulo future causal geodesic completeness, which is proved in Section~\ref{sec:cgc}. We also formulate and prove uniform estimates for higher derivatives (Proposition~\ref{prop:high-d}), which will be useful in the proof of Theorem~\ref{thm:main} in the next section.

This section is structured as follows. In Sections~\ref{subsec:main.finite:btstrp}--\ref{subsec:main.finite:btstrp-close}, we carry out the main bootstrap argument, which lies at the heart of our proof of Theorem~\ref{thm:main.finite}. The proof of Theorem~\ref{thm:main.finite} is then completed in Section~\ref{subsec:main.finite:du}. Finally, in Section~\ref{subsec:high-d}, we prove estimates for higher derivatives (Proposition~\ref{prop:high-d}), which are uniform with respect to the initial curve $u_{0}$.

\subsection{Bootstrap assumptions} \label{subsec:main.finite:btstrp}
Suppose that $(r, \phi, m)$ is a $C^{1}$ solution to \eqref{eq:SSESF} on $\PD_{[u_{0}, u_{1}]}$.
We introduce the following bootstrap assumptions on $\PD_{[u_{0}, u_{1}]}$:
\begin{enumerate}
\item {\it Assumptions on the geometry.}
\begin{align} \label{eq:btstrp:geom}
	\dvr > \frac{1}{3}, \quad-\frac{1}{6}> \nu > - \frac{2}{3}, \quad 1-\mu > \frac{1}{2}.
\end{align}

\item {\it Assumptions on the inhomogeneous part of $\dvr^{-1} \rd_{v} \phi$.}
\begin{align}
	\int_{u_0}^{u} \abs{\frac{2m \dur}{(1-\mu) r^{2}} \phi (u', v)} \, \ud u'
	& \leq 2 \eps r_+^{-\dc}, \label{eq:btstrp:wave:1}
\end{align}
where $r_{+} := \max \set{1, r}$.
\item {\it Assumption on $(\dvr^{-1} \rd_{v})^2 \phi$.}
\begin{align}
       \label{eq:btstrp:wave:2}
       \abs{(\dvr^{-1} \rd_{v})^{2} \phi} \leq 3 \eps.
\end{align}

\end{enumerate}

Henceforth until Section~\ref{subsec:main.finite:btstrp-close}, the domain for each bound is $\PD_{[u_{0}, u_{1}]}$ unless otherwise specified.
We will use the convention that unless otherwise stated, the constants $C$ depend only on $\dc$. Moreover, we will also use the notation $\aleq$ such that the implicit constants are allowed to depend only on $\dc$.
\subsection{Preliminary estimates}
Recall that  $r$ vanishes on the boundary $\Gamma$ and $\dvr > 1/3$ by the bootstrap assumption. Moreover, by the bootstrap assumptions on $\dur$ and $1-\mu$, we have $\rd_{u} \dvr \leq 0$. It follows that
\begin{equation} \label{eq:dvr}
	\frac{1}{3} < \dvr \leq \frac{1}{2}.
\end{equation}
This bound implies that at the point $(u, v)$ the radius $r(u, v)$ is comparable to the difference $v-u$ up to a constant:
  \begin{equation}
 \label{eq:rvminusu}
 \frac{1}{3}(v-u)\leq \int_{u}^{v}\dvr(u, v')\, \ud v'=r(u, v)\leq \frac{1}{2}(v-u).
 \end{equation}
In the proof, we will frequently need estimates for integrals of powers of $r$. We will collect these estimates in Lemma \ref{lem:est4noverr}.
To this end, the notation
$$r_{+} := \max \set{1, r}$$
introduced above will be convenient.

The following lemma holds also due to \eqref{eq:dvr} and the assumption that $r$ vanishes on the boundary $\Gamma$.
\begin{lemma} \label{lem:est4noverr}
Assume the bootstrap assumption on the geometry \eqref{eq:btstrp:geom}.
Then for all $k> 1$, we have
\begin{align}
\label{eq:est4noverr:Cu}
\int_{u}^v r_+^{-k}(u, v')\, \ud v'&\leq C \min\set{1, r}(u, v) ,\\
\label{eq:est4noverr:barCv}
\int_{u_0}^{u}r_+^{-k}(u', v)\, \ud u'&\leq Cr_+^{-k+1}(u, v)
\end{align}
for some constant $C$ depending only on $k$.
\end{lemma}
\begin{proof}
For \eqref{eq:est4noverr:Cu}, the case when $r(u, v)\leq 1$ is easy to verify. The lower bound for $\dvr$ implies that
 \begin{align*}
 \int_{u}^{v}r_+^{-k}(u, v')\,\ud v'\leq \int_{u}^{v}(\dvr^{-1}\dvr)(u, v')\,\ud v'\leq 3\int_{u}^{v}\rd_v r(u, v')\,\ud v'=3 r(u, v).
 \end{align*}
When $r(u, v)>1$, let $v^*$ be the unique $v$ value such that $r(u, v^*)=1$. Then we have
 \begin{align*}
 \int_{u}^{v}r_+^{-k}(u, v')\, \ud v'&\leq 3 r(u, v^*)+3\int_{v^*}^{v}(\dvr r_+^{-k})(u, v')\, \ud v' \leq \frac{3k}{k-1}.
 \end{align*}
 The proof for \eqref{eq:est4noverr:barCv} is very similar where we make use of the bootstrap assumption on the lower bound of $-\dur$. More precisely, for $r(u,v)>1$, we have
\begin{align*}
\int_{u_0}^{u} r_+^{-k}(u', v)\, \ud u'\leq 6\int_{u_0}^{u}(-\dur r_+^{-k})(u', v)\, \ud u' \leq \frac{6}{k-1} r_+^{-k+1}(u, v).
\end{align*}
On the other hand, if $r(u,v)\leq 1$, we defined $u^*$ to be the unique $u$ value such that $r(u^*,v)=1$. We then obtain
\begin{align*}
\int_{u_0}^{u} r_+^{-k}(u', v)\, \ud u' \leq \frac{6}{k-1} +6\int_{u^*}^u (-\rd_u r(u',v))\,\ud u'\leq \f{6k}{k-1}=\f{6k}{k-1}r_+^{-k+1}(u, v).
\end{align*}
\end{proof}
Estimate \eqref{eq:est4noverr:Cu} bounds the integral from the axis to the given point $(u, v)$ on the outgoing null hypersurface $C_u$. It will be used if we want to control some quantity by using the data on the axis, e.g., the mass $m$.
 Estimate \eqref{eq:est4noverr:barCv} controls the integral from the point $(u_0, v)$ on the initial hypersurface $C_{u_0}$ to the given point $(u, v)$. We will use it when we want to control the solution from the data given on $C_{u_0}$.

\subsection{Estimates for $\phi$}
The following lemma 
 will be crucial for many estimates to follow.
\begin{lemma} \label{lem:no-shift}
For any $u_{1} \leq u_{2}$, we have
\begin{equation} \label{eq:no-shift}
\int_{u_{1}}^{u_{2}} \frac{2 m (-\dur)}{(1-\mu) r^{2}} (u, v) \, \ud u
= \log \frac{\dvr(u_{1}, v)}{\dvr(u_{2}, v)}.
\end{equation}
Hence, under the bootstrap assumptions \eqref{eq:btstrp:geom}--\eqref{eq:btstrp:wave:2}, we have
\begin{equation} \label{eq:est4shift}
	0 \leq \int_{u_{1}}^{u_{2}} \frac{2 m (-\dur)}{(1-\mu) r^{2}} (u, v) \, \ud u \leq \log \frac{3}{2}.
\end{equation}
\end{lemma}
\begin{proof}
Equation \eqref{eq:no-shift} is an immediate consequence of the equation \eqref{eq:wave4dvr}.
Then \eqref{eq:est4shift} follows from \eqref{eq:btstrp:geom} and \eqref{eq:dvr}. \qedhere
\end{proof}

We now derive estimates for the scalar field $\phi$ and its $\dvr^{-1} \rd_{v}$ derivatives.
\begin{proposition}
 \label{prop:Est4phi}
Under the bootstrap assumptions \eqref{eq:btstrp:geom}--\eqref{eq:btstrp:wave:2}, we have the following estimates for the scalar field:
\begin{align}
\label{eq:btstrp:pf:dvrphi}
 \abs{\dvr^{-1}\rd_v(r\phi)(u, v)}&\les \abs{\Phi(v)}+\eps r_+^{-\dc},\\
 \label{eq:btstrp:pf:phi}
 \abs{\phi(u,v)} &\les \eps r_+^{-\dc},\\
 \label{eq:btstrp:pf:dvphi}
	\abs{\dvr^{-1} \rd_{v} \phi(u, v)} &\les \min\{\frac{|\Phi(v)|+\eps r_+^{-\dc}}{r}, \eps\}, \\
\label{eq:btstrp:pf:rdvdvphi}
	\abs{r (\dvr^{-1} \rd_{v})^{2} \phi(u, v)} &\les \eps.
\end{align}
Here the implicit constants depend only on $\dc$.
\end{proposition}

\begin{proof}
By \eqref{eq:wave4dvrphi}, we have the integral formula
\begin{align*}
	(\dvr^{-1} \rd_{v}) (r \phi)(u, v)
	= & e^{-\int_{u_{0}}^{u} \frac{2 m \dur}{(1-\mu) r^{2}} (u', v) \, \ud u'} (\dvr^{-1} \rd_{v}) (r \phi)(u_{0}, v) \\
	& + \int_{u_{0}}^{u} e^{-\int_{u'}^{u} \frac{2 m \dur}{(1-\mu) r^{2}} (u'', v) \, \ud u''} \frac{2 m \dur}{(1-\mu) r^{2}} \phi (u', v) \, \ud u'.
\end{align*}
Then using Lemma~\ref{lem:no-shift} and the bootstrap assumption \eqref{eq:btstrp:wave:1}, we have
\begin{align}\label{rdvrphi.diff.est}
 \abs{\dvr^{-1}\rd_v(r\phi)(u, v) - \bb( \frac{\dvr(u_{0}, v)}{\dvr(u, v)} \bb) \dvr^{-1}\rd_v(r\phi)(u_0, v)}
 		\leq \frac{3}{2}\int_{u_0}^{u} \abs{\frac{2m \dur}{(1-\mu) r^{2}} \phi (u', v)} \, \ud u'
                        \les \eps r_+^{-\dc}.
\end{align}
Recalling that $\dvr^{-1} \rd_{v} (r \phi)(u_{0}, v) = \Phi(v)$, the desired estimate \eqref{eq:btstrp:pf:dvrphi} follows.

Once we have estimate \eqref{eq:btstrp:pf:dvrphi}, we can then use the averaging formula \eqref{eq:avg:phi} to control the scalar field $\phi$:
\begin{equation*}
\begin{split}
 \abs{\phi(u,v)} &\leq \frac{1}{r}\int_{u}^{v}\abs{\dvr^{-1}\rd_v(r\phi)(u, v')} \dvr \,\ud v'\\
 &\les \frac{1}{r}\int_{u}^{v}|\Phi(v')|\, \ud v'+\frac{\eps}{r}\int_{u}^{v}(\dvr r_+^{-\dc})(u, v')\, \ud v'\\
&\les \frac{\eps}{r} \min \set{v-u, (v-u)^{1-\dc}}+\frac{\eps(r_+^{1-\dc}-1)}{r}\\
&\les \eps r_+^{-\dc}
\end{split}
\end{equation*}
Here we have used the condition \eqref{eq:IDcond} and the relation \eqref{eq:rvminusu} to estimate the integral of $\Phi(v')$. The inequality
\[
 \frac{r_+^{1-\dc}-1}{r}\leq r_+^{-\dc}
\]
follows from the fact that $r_+=\max\{1,r\}$, $r\geq 0$, $0<\dc<1$.

Estimate \eqref{eq:btstrp:pf:dvrphi} for $\rd_v(r\phi)$ and estimate \eqref{eq:btstrp:pf:phi} together with \eqref{eq:dvphi:large-r} give us the following bound
for $\dvr^{-1} \rd_{v} \phi$:
\begin{equation*}
\begin{split}
	\abs{\dvr^{-1} \rd_{v} \phi(u, v)} &\leq  \frac{1}{r} (\abs{\dvr^{-1} \rd_{v} (r \phi)} + \abs{\phi}) \\
	&\les \frac{1}{r} \left(|\Phi(v)|+\eps r_+^{-\dc}\right).
\end{split}
\end{equation*}
Such an estimate is favorable in the region far away from the axis.


On the other hand, using the differentiated averaging formula \eqref{eq:dvphi:avg} and the bootstrap assumption \eqref{eq:btstrp:wave:2}, we are able to show that $\rd_{v}\phi$ is uniformly bounded near the axis, which completes the proof of \eqref{eq:btstrp:pf:dvphi}:
\begin{equation*}
\begin{split}
\abs{\dvr^{-1}\rd_{v}\phi(u, v)}&\leq \frac{1}{r^2}\int_{u}^{v}|(\dvr^{-1}\rd_v)^2(r\phi)(u, v')|(r\dvr)(u, v')\,\ud v'\\
&\les \eps r^{-2}\int_{u}^{v}\,\ud r^2\les \eps.
\end{split}
\end{equation*}

Finally, \eqref{eq:btstrp:pf:rdvdvphi} follows from the commutation formula
\begin{equation} \label{eq:comm-rdvdvphi}
	r (\dvr^{-1} \rd_{v})^{2} \phi = (\dvr^{-1} \rd_{v})^{2} (r \phi) - 2 \dvr^{-1} \rd_{v} \phi,
\end{equation}
as well as the bootstrap assumption \eqref{eq:btstrp:wave:2} and estimate \eqref{eq:btstrp:pf:dvphi}. \qedhere
\end{proof}

\subsection{Estimates for the Hawking mass}
Once we have estimate for the solution $\rd_v\phi$, we can derive bounds for the mass $m$.
\begin{proposition}
 \label{prop:Est4m}
 Under the bootstrap assumptions \eqref{eq:btstrp:geom}--\eqref{eq:btstrp:wave:2}, we have
 \begin{equation}
  \label{eq:btstrp:pf:m}
  m(u, v)\les \eps^2\min\{r^3, r^{1-\dc}\}.
 \end{equation}
\end{proposition}
We remark that the gain of the positive power in $r$ is crucial to close the bootstrap assumptions on the nonlinearity near the axis. Indeed, as a quick consequence of \eqref{eq:btstrp:pf:m}, we have
 \begin{equation}
 \label{eq:btstrp:pf:moverr}
 \frac{m}{r^k}\leq C \eps^2 r_+^{-k+1-\dc}
\end{equation}
for all $0\leq k\leq 3$, where the constant $C$ depends only on $k$ and $\dc$.
\begin{proof}
By \eqref{eq:bc4m} and \eqref{eq:dv-m}, we have
\begin{equation} \label{eq:dvm:int}
	m (u,v)= \frac{1}{2} \int_{u}^{v} (1-\mu) r^{2} (\dvr^{-1} \rd_{v} \phi)^{2}  \dvr (u, v') \, \ud v'.
\end{equation}
Recall that $\frac{1}{2} \leq 1-\mu \leq 1$ by the bootstrap assumption \eqref{eq:btstrp:geom}. From estimates \eqref{eq:btstrp:pf:dvrphi}--\eqref{eq:btstrp:pf:dvphi}, we can show that
\begin{equation*}
\begin{split}
m(u,v) &\leq \frac{1}{2} \int_{u}^{v} \bb(\frac{r \rd_{v} \phi}{\dvr} \bb)^{2} \dvr (u, v') \, \ud v'\\
	&\les \min\{\int_{u}^{v}|\Phi(v')|^2\, \ud v'+\eps^2\int_{u}^{v} (\dvr r_+^{-2\dc})(u, v')\, \ud v', \eps^2 \int_{u}^{v} \dvr r^2(u, v')\,\ud v' \}\\
	&\les \min\{\eps \int_{u}^{v}|\Phi(v')|\, \ud v'+\eps^2 r_+^{1-2\dc}, \eps^2 r^3\}\\
	&\les \eps^2 \min\{r_+^{1-\dc}, r^3\}.
\end{split}
	\end{equation*}
Here we have used the condition \eqref{eq:IDcond} to control the integral of $|\Phi(v')|$.
\end{proof}

We also derive estimates for $\dvr^{-1} \rd_{v}$ of $2m$ and $2m / r^{2}$, which will be needed for closing the bootstrap assumption for $(\dvr^{-1} \rd_{v})^{2} (r \phi)$.
\begin{proposition} \label{prop:Est4dv-m}
Under the bootstrap assumptions \eqref{eq:btstrp:geom}--\eqref{eq:btstrp:wave:2}, we have
\begin{align}
\label{eq:est4dv-m}
	\abs{\dvr^{-1} \rd_{v} (2m)} \aleq & \eps^{2} \min \set{r^{2}, 1} \\
\label{eq:est4dv-m-over-r2}
	\abs{\dvr^{-1} \rd_{v} \bb( \frac{2m}{r^{2}} \bb)} \aleq & \eps^{2} \min \set{1, r^{-2}}.
\end{align}
\end{proposition}

An important point is that $\dvr^{-1} \rd_{v} \bb( \frac{2m}{r^{2}} \bb)$ is uniformly bounded near the axis; this fact will be clear by the use of the differentiated averaging formula \eqref{eq:dv-2m-over-r2}.

\begin{proof}
Estimate \eqref{eq:est4dv-m} is a simple consequence of the equation \eqref{eq:dv-m}, as well as the bootstrap assumption \eqref{eq:btstrp:geom} and estimate \eqref{eq:btstrp:pf:dvphi}.

To establish \eqref{eq:est4dv-m-over-r2}, we begin by showing that
\begin{equation} \label{eq:est4dv-m-over-r2:1}
 \abs{\dvr^{-1} \rd_{v} \bb( \frac{2m}{r^{2}} \bb)} \aleq \eps^{2},
\end{equation}
which is acceptable in the region $\set{r \leq 1}$ near the axis. Using the differentiated averaging formula \eqref{eq:dv-2m-over-r2}, we have
\begin{align*}
 \abs{\dvr^{-1} \rd_{v} \bb( \frac{2m}{r^{2}} \bb)(u,v)}
 \leq & \frac{1}{r^{3}} \int_{u}^{v} \bb( \dvr^{-1} \rd_{v} \bb( (1-\mu) r (\dvr^{-1} \rd_{v} \phi)^{2} \bb) \bb) r^{2} \dvr \, \ud v' \\
 \leq & \sup_{v' \in [u, v]} \abs{\dvr^{-1} \rd_{v} \bb( (1-\mu) r (\dvr^{-1} \rd_{v} \phi)^{2} \bb)}.
\end{align*}
To estimate the last line, we expand
\begin{align*}
\dvr^{-1} \rd_{v} \bb( (1-\mu) r (\dvr^{-1} \rd_{v} \phi)^{2} \bb)
= & \bb( 1 - \dvr^{-1} \rd_{v} (2m) \bb) (\dvr^{-1} \rd_{v} \phi)^{2} \\
	& + 2 (1-\mu) \dvr^{-1}  r (\dvr^{-1} \rd_{v})^{2} \phi \, \dvr^{-1} \rd_{v} \phi
\end{align*}
Then by \eqref{eq:btstrp:pf:dvphi}, \eqref{eq:btstrp:pf:rdvdvphi}, \eqref{eq:est4dv-m} and the fact that $\mu \geq 0$, it follows that the absolute value of the preceding expression is uniformly bounded by $\aleq \eps^{2}$. Hence \eqref{eq:est4dv-m-over-r2:1} is proved.

In order to complete the proof of \eqref{eq:est4dv-m-over-r2}, it suffices to prove
\begin{equation} \label{eq:est4dv-m-over-r2:2}
	 \abs{\dvr^{-1} \rd_{v} \bb( \frac{2m}{r^{2}} \bb)} \aleq \eps^{2} r^{-2},
\end{equation}
which is favorable in the region $\set{r \geq 1}$ away from the axis. In this case, recall that by \eqref{eq:dv-m-over-r2:large-r}, we have
\begin{equation*}
	r^{2} \, \dvr^{-1} \rd_{v} \bb( \frac{2m}{r^{2}} \bb)
	= \dvr^{-1} \rd_{v} (2 m) - \frac{4m}{r}.
\end{equation*}
The desired estimate \eqref{eq:est4dv-m-over-r2:2} now follows from \eqref{eq:btstrp:pf:m} and \eqref{eq:est4dv-m}. \qedhere
\end{proof}



\subsection{Closing the bootstrap assumptions} \label{subsec:main.finite:btstrp-close}

The purpose of this subsection is to improve the bootstrap assumptions \eqref{eq:btstrp:geom}--\eqref{eq:btstrp:wave:2}, using the estimates for the scalar field in Proposition \ref{prop:Est4phi} and the bounds for the mass in Propositions~\ref{prop:Est4m} and \ref{prop:Est4dv-m}. Combined with local well-posedness of \eqref{eq:SSESF} for $C^{1}$ solutions, global existence of the solution then follows.

We begin by improving the bootstrap assumption \eqref{eq:btstrp:geom} on the geometry.
A corollary of Proposition~\ref{prop:Est4m} is that $\mu = \frac{2m}{r}$ is small for sufficiently small $\eps$; this improves the bootstrap assumption on $1-\mu$.
\begin{corollary}
\label{cor:btstrp:pf:mu}
 Under the bootstrap assumptions \eqref{eq:btstrp:geom}--\eqref{eq:btstrp:wave:2}, we have
\begin{equation}
\label{eq:btstrp:pf:mu}
1-C\eps^2\leq  1-\mu\leq 1
\end{equation}
for some constant $C$ depending only on $\dc$.
\end{corollary}

To close the bootstrap for $\dvr$, as well as for \eqref{eq:btstrp:wave:1} below, a key role is played by the following lemma.
\begin{lemma} \label{lem:Est4shift}
 Under the bootstrap assumptions \eqref{eq:btstrp:geom}--\eqref{eq:btstrp:wave:2}, we have
 \begin{equation} \label{eq:est4shift:imp}
 	\abs{\frac{2 m \dur}{(1-\mu) r^{2}}} \leq C \eps^{2} r_{+}^{-1-\gmm},
\end{equation}
for some constant $C$ depending only on $\dc$.
\end{lemma}
\begin{proof}
The desired estimate follows from \eqref{eq:btstrp:geom} and \eqref{eq:btstrp:pf:m} in Proposition~\ref{prop:Est4m}. \qedhere
\end{proof}

With Lemma~\ref{lem:Est4shift}, we can immediately prove an improved bound for $\dvr$.
\begin{proposition}
 \label{prop:dvr:im}
 Under the bootstrap assumptions \eqref{eq:btstrp:geom}--\eqref{eq:btstrp:wave:2}, we have
 \begin{equation}
  \label{eq:btstrp:pf:geom:dvr}
\dvr(u, v)\geq \frac{1}{2} e^{-C \eps^2 r_+^{-\dc}}
 \end{equation}
for some constant $C$ depending only on $\dc$.
\end{proposition}
\begin{proof}
By integrating \eqref{eq:wave4dvr} and using \eqref{eq:est4shift:imp}, we can show that
\begin{equation*}
\begin{split}
\log \dvr(u, v) - \log\dvr(u_0, v) &= \int_{u_0}^{u}\frac{2m\dur}{(1-\mu)r^2}(u', v)\, \ud u'\\
&\geq -C\eps^2\int_{u_0}^{u}(\dvr r_+^{-1-\dc})(u', v)\, \ud u'\\
&\geq -C \eps^2 r_+^{-\dc},
\end{split}
\end{equation*}
where we have used Lemma~\ref{lem:est4noverr} on the last line. Recalling our initial gauge condition that $\dvr=\frac{1}{2}$ on the initial hypersurface $C_{u_0}$, \eqref{eq:btstrp:pf:geom:dvr} follows. \qedhere
\end{proof}

In order to estimate $\dur$, the following lemma is needed.
\begin{lemma}
 \label{lem:est4rPhi}
 Under the assumption \eqref{eq:IDcond} on the initial data, for all $u<v^*<v$ we have
 \begin{equation}
  \label{eq:rPhi}
\int_{v^*}^{v}(v'-u)^{-1}|\Phi(v')|\, \ud v'\leq C\eps (v^*-u)^{-\dc}
 \end{equation}
for some constant $C$ depending only on $\dc$.
\end{lemma}
\begin{proof}
Let
\[
F(s)=\int_{u}^{s}|\Phi(v')|\, \ud v'
\]
for $s\geq v^*$. Then $F'(s)=|\Phi(s)|$. From the assumption \eqref{eq:IDcond}, we have
\[
F(s)\leq \eps (s-u)^{1-\dc}.
\]
Therefore we can show that
\begin{equation*}
\begin{split}
\int_{v^*}^{v}(v'-u)^{-1}|\Phi(v')|\, \ud v'&=\int_{v^*}^{v}(v'-u)^{-1}F'(v')\, \ud v'\\
&=\left.(v'-u)^{-1}F(v')\right|_{v^*}^{v}+\int_{v^*}^{v}(v'-u)^{-2}F(v')\, \ud v'\\
&\leq \eps (v-u)^{-\dc}+\eps\int_{v^*}^{v}(v'-u)^{-1-\dc}\, \ud v'\\
&\leq \eps(1+\dc^{-1})(v^*-u)^{-\dc},
\end{split}
\end{equation*}
as desired. \qedhere
\end{proof}
Using the above lemma, we now estimate $\dur$.
\begin{proposition}
 \label{prop:dur:im}
 Under the bootstrap assumptions \eqref{eq:btstrp:geom}--\eqref{eq:btstrp:wave:2}, we have
 \begin{equation}
\label{eq:btstrp:pf:geom:dur}
\frac{1}{2}e^{-C\eps^2}\leq -\dur(u, v)\leq \frac{1}{2}e^{C\eps^2}
\end{equation}
for some constant $C$ depending only on $\dc$.
\end{proposition}
\begin{proof}
We rely on the Raychaudhuri equation \eqref{eq:raych4v}. From that we obtain the representation for $\dur$:
\begin{equation}\label{Ray}
\log \frac{-\dur}{1-\mu}(u, v)=\log \frac{-\dur}{1-\mu}(u, u)+\int_{u}^{v}\bb( \frac{r \rd_{v} \phi}{\dvr} \bb)^{2} \frac{\dvr}{r}(u, v')\,\ud v'.
\end{equation}
To control the integral on the right-hand side, define $v^*$ to be the unique $v$ value such that $r(u, v^*)=1$. We then divide the integral into the regions $[u, v^*]$ and $[v^*, v]$.
By \eqref{eq:IDcond}, \eqref{eq:btstrp:pf:dvphi} and \eqref{eq:rPhi}, we have
\begin{align*}
\int_{u}^{v}\bb( \frac{r \rd_{v} \phi}{\dvr} \bb)^{2} \frac{\dvr}{r}(u, v')\,\ud v'&=\int_{u}^{v^*}\bb( \frac{r \rd_{v} \phi}{\dvr} \bb)^{2} \frac{\dvr}{r}(u, v')\,\ud v'+\int_{v^*}^{v}\bb( \frac{r \rd_{v} \phi}{\dvr} \bb)^{2} \frac{\dvr}{r}(u, v')\,\ud v'\\
&\les \eps^2\int_{u}^{v^*}(r\dvr)(u, v')\, \ud v'+\int_{v^*}^{v}r^{-1}(|\Phi|^2+\eps^2 r_+^{-2\dc})(u, v')\, \ud v'\\
&\les \eps^2(v^*-u)+\eps\int_{v^*}^{v}(v'-u)^{-1}|\Phi(v')|\,\ud v' +\eps^2\int_{v^*}^{v}r_+^{-1-2\dc}(u, v')\,\ud v'\\
&\les \eps^2 (1+r(u, v^*)^{-\dc})\les \eps^2.
\end{align*}
Here we have used estimate \eqref{eq:est4noverr:barCv} to bound the integral of $r_+^{-1-2\dc}$.
Using also the following identities on the axis $\Gamma$
\[
\dur+\dvr=0, \quad \mu=0,
\]
we can bound $\dur$ as follows:
\begin{equation*}
\begin{split}
\log \frac{-\dur}{1-\mu}(u, v)&=\log \frac{-\dur}{1-\mu}(u, u)+\int_{u}^{v}\bb( \frac{r \rd_{v} \phi}{\dvr} \bb)^{2} \frac{\dvr}{r}(u, v')\,\ud v'\\
&\leq \log \dvr(u, u)+C\eps^2
\end{split}
\end{equation*}
for some constant $C$ depending only on $\dc$. Moreover, since the last term in \eqref{Ray} is positive, we have the trivial bound
\[
 \log \frac{-\dur}{1-\mu}(u, v)\geq\log \frac{-\dur}{1-\mu}(u, u)=\log \dvr(u, u).
\]
Then from Corollary \ref{cor:btstrp:pf:mu} and estimate \eqref{eq:btstrp:pf:geom:dvr} in Proposition \ref{prop:dvr:im}, we have
\begin{equation*}
\frac{1}{2}e^{-C\eps^2} (1-C\eps^2)\leq\dvr(u, u)(1-\mu)\leq -\dur(u, v)\leq \dvr(u, u)e^{C\eps^2}\leq \frac{1}{2}e^{C\eps^2}
\end{equation*}
for some constant $C$ depending only on $\dc$. Estimate \eqref{eq:btstrp:pf:geom:dur} in the proposition then follows. \qedhere
\end{proof}

Next, we establish an estimate for the inhomogeneous part of $\dvr^{-1} \rd_{v} (r \phi)$, which improves the bootstrap assumption \eqref{eq:btstrp:wave:1}. The key ingredient is Lemma~\ref{lem:Est4shift}.
\begin{proposition}
 \label{prop:Est4nonlin}
 Under the bootstrap assumptions \eqref{eq:btstrp:geom}--\eqref{eq:btstrp:wave:2}, we can show that
 \begin{align}
  \label{eq:btstrp:pf:wave:1:im}
 \int_{u_0}^{u}\abs{\frac{2m\dur}{(1-\mu)r^2}\phi(u', v)}\, \ud u' &\leq C \eps^3 r_+^{-2\dc}(u, v),
 \end{align}
 for some constant $C$ depending only on $\dc$.
\end{proposition}
\begin{proof}
Using Lemma~\ref{lem:Est4shift} and \eqref{eq:btstrp:pf:phi}, we may estimate
\begin{equation*}
 \begin{split}
 \int_{u_0}^{u}\abs{\frac{2m\dur}{(1-\mu)r^2}\phi(u', v)}\, \ud u'\
 \les \eps^3\int_{u_0}^{u}(r_+^{-1-\dc} r_+^{-\dc})(u', v)\, \ud u'
 \les \eps^3 r_+^{-2\dc}(u, v),
\end{split}
 \end{equation*}
where we have used estimate \eqref{eq:est4noverr:barCv} to bound the integral. \qedhere

\end{proof}

%
%
%

It remains to close the bootstrap assumption \eqref{eq:btstrp:wave:2} for $(\dvr^{-1} \rd_{v})^{2} (r \phi)$. Analogous to the role played by Lemma~\ref{lem:Est4shift} in the preceding proof, this task requires a good bound on the factor
\begin{equation*}
	\dvr^{-1} \rd_{v} \bb( \frac{2 m \dur}{(1-\mu) r^{2}} \bb).
\end{equation*}
This is the subject of the following lemma.
\begin{lemma} \label{lem:est4dv-shift}
Under the bootstrap assumptions \eqref{eq:btstrp:geom}--\eqref{eq:btstrp:wave:2}, we have
\begin{align}
\label{eq:est4dv-dur-over-1-mu}
	\abs{\dvr^{-1} \rd_{v} \bb( \frac{\dur}{1-\mu} \bb)}
	\aleq & \eps^{2} \min \set{r, \frac{1}{r}}, \\
\label{eq:est4dv-shift}
	\abs{\dvr^{-1} \rd_{v} \bb( \frac{2 m \dur}{(1-\mu) r^{2}} \bb)}
	\aleq & \eps^{2} \min\set{1, \frac{1}{r^{2}}} + \eps^{4} \min \set{r^{2}, \frac{1}{r^{2}}}.
\end{align}
\end{lemma}
\begin{proof}
Estimate \eqref{eq:est4dv-dur-over-1-mu} is an immediate consequence of the Raychaudhuri equation \eqref{eq:raych4v}, the bootstrap assumption \eqref{eq:btstrp:geom} and estimate \eqref{eq:btstrp:pf:dvphi}. Estimate \eqref{eq:est4dv-shift} then follows from \eqref{eq:btstrp:pf:m}, \eqref{eq:est4dv-m-over-r2} and the preceding bound. \qedhere
\end{proof}
We are ready to prove an improved estimate for $(\dvr^{-1} \rd_{v})^{2} (r \phi)$.
\begin{proposition} \label{prop:Est4dvdvrphi}
Under the bootstrap assumptions \eqref{eq:btstrp:geom}--\eqref{eq:btstrp:wave:2}, we have
\begin{equation} \label{eq:btstrp:pf:wave:2:im}
	\abs{(\dvr^{-1} \rd_{v})^{2} (r \phi) } \leq \bb(\frac{3}{2} \bb)^{2} \eps + C (\eps^{3} + \eps^{5}).
\end{equation}
for some constant $C$ depending only on $\dc$.
\end{proposition}

\begin{proof}
Commuting $\dvr^{-1} \rd_{v}$ with the equation \eqref{eq:wave4dvrphi} for $(\dvr^{-1} \rd_{v}) \phi$, we arrive at the equation
\begin{align*}
	\rd_{u} \bb( (\dvr^{-1} \rd_{v})^{2} (r \phi) \bb)
	= & - \frac{4 m \dur}{(1-\mu) r^{2}} (\dvr^{-1} \rd_{v})^{2} (r \phi) \\
	& - \dvr^{-1} \rd_{v} \bb( \frac{2m \dur}{(1-\mu) r^{2}} \bb) \dvr^{-1} \rd_{v} (r \phi)
	+ \dvr^{-1} \rd_{v} \bb( \frac{2m \dur}{(1-\mu) r^{2}}  \phi \bb)  \\
	=: &- \frac{4 m \dur}{(1-\mu) r^{2}} (\dvr^{-1} \rd_{v})^{2} (r \phi) + N_{2}.
\end{align*}
Hence we have
\begin{equation} \label{eq:dvdvrphi-formula}
\begin{aligned}
	(\dvr^{-1} \rd_{v})^{2} (r \phi)(u, v)
	= & e^{- \int_{u_{0}}^{u} \frac{4m \dur}{(1-\mu) r^{2}} (u', v) \, \ud u'} (\dvr^{-1} \rd_{v})^{2} (r \phi)(u_{0}, v) \\
	& + \int_{u_{0}}^{u}  e^{- \int_{u'}^{u} \frac{4m \dur}{(1-\mu) r^{2}} (u'', v) \, \ud u''} N_{2}(u', v) \, \ud u'.
\end{aligned}
\end{equation}
By \eqref{eq:est4shift} and the bootstrap assumption \eqref{eq:btstrp:geom}, the integration factor is bounded by
\begin{equation}\label{int.factor.bd}
e^{- \int_{u_{0}}^{u} \frac{4m \dur}{(1-\mu) r^{2}} (u', v) \, \ud u'}
\leq \bb(\frac{\frac{1}{2}}{\frac{1}{3}} \bb)^{2} = \bb( \frac{3}{2} \bb)^{2}.
\end{equation}
Combined with the initial condition $\abs{\dvr^{-1} \rd_{v} \Phi} \leq \eps$, we see that the contribution of the data on $C_{u_{0}}$ is acceptable.
Hence, using \eqref{int.factor.bd} again, recalling the definition of $N_{2}(u, v)$ and using Leibniz's rule, it only remains to establish
\begin{equation*}
	\int_{u_{0}}^{u} \abs{\dvr^{-1} \rd_{v} \bb( \frac{2m \dur}{(1-\mu) r^{2}} \bb)} \bb( \abs {\dvr^{-1} \rd_{v} (r \phi)} + \abs{\phi} \bb)  \, \ud u'
	+ \int_{u_{0}}^{u} \abs{\frac{2m \dur}{(1-\mu) r^{2}}} \abs{\dvr^{-1} \rd_{v} \phi} \, \ud u' \aleq \eps^{3} + \eps^{5},
\end{equation*}
uniformly in $u \in [u_{0}, v]$ and $v$. This bound is an immediate consequence of the Leibniz rule, \eqref{eq:btstrp:pf:dvrphi}, \eqref{eq:btstrp:pf:dvphi}, \eqref{eq:est4shift:imp} and \eqref{eq:est4dv-shift}. \qedhere
\end{proof}

By Corollary~\ref{cor:btstrp:pf:mu} and Propositions~\ref{prop:dvr:im}, \ref{prop:dur:im}, \ref{prop:Est4nonlin} and \ref{prop:Est4dvdvrphi}, there exists a constant $0 < \eps_{1} < 1$ (depending only on $\dc$) such that the bootstrap assumptions \eqref{eq:btstrp:geom}--\eqref{eq:btstrp:wave:2} for $\PD_{[u_{0}, u_{1}]}$ are improved if $\eps \leq \eps_{1}$. Then by a standard continuity argument, the $C^{1}$ solution $(r, \phi, m)$ exists globally on $\PD_{[u_{0}, \infty)}$, which satisfies the bootstrapped bounds \eqref{eq:btstrp:geom}--\eqref{eq:btstrp:wave:2} as well as the estimates derived in this section so far.

In the remainder of this section, {\bf we require that $\eps \leq \eps_{1}$ and take $(r, \phi, m)$ to be such a global $C^{1}$ solution obeying \eqref{eq:btstrp:geom}--\eqref{eq:btstrp:wave:2}.}

\subsection{Estimate for $\rd_{v}$ derivatives of $r$  and $\phi$} \label{subsec:main.finite:dv}
Let $(r, \phi, m)$ be the global $C^{1}$ solution constructed above obeying \eqref{eq:btstrp:geom}--\eqref{eq:btstrp:wave:2}; we assume furthermore that $\eps < \eps_{1} < 1$. Here we show that $(r, \phi, m)$ obeys the estimates for $\rd_{v}^{2} (r \phi)$ and $\rd_{v}^{2} r$ stated in Theorem~\ref{thm:main.finite}; see Corollary~\ref{cor:Est4dvdv}.

We start by establishing an estimate for $\dvr^{-1} \rd_{v} \log \dvr$, which follows from essentially the same estimates used in Proposition~\ref{prop:Est4dvdvrphi}.

\begin{proposition} \label{prop:Est4dvdvr}
For the global solution we have constructed for \eqref{eq:SSESF}, we have
\begin{equation} \label{eq:est4dvdvr}
	\abs{\dvr^{-1} \rd_{v} \log \dvr} \aleq \eps^{2},
\end{equation}
where the implicit constant depends only on $\dc$.
\end{proposition}
\begin{proof}
Taking $\dvr^{-1} \rd_{v}$ of the equation for $\rd_{u} \log \dvr$, we obtain
\begin{equation*}
	\rd_{u} \bb( \dvr^{-1} \rd_{v} \log \dvr \bb)  = - \frac{2 m \dur}{(1-\mu) r^{2}} \dvr^{-1} \rd_{v}\log \dvr
								+ \dvr^{-1} \rd_{v} \bb( \frac{2 m \dur}{(1-\mu) r^{2}} \bb).
\end{equation*}
Note furthermore that $\rd_{v} \log \dvr (u_{0}, v)= 0$, due to the initial gauge condition $\dvr = \frac{1}{2}$. Hence
\begin{equation*}
	\dvr^{-1} \rd_{v} \log \dvr = \int_{u_{0}}^{u} e^{- \int_{u'}^{u} \frac{2 m \dur}{(1-\mu) r^{2}} (u'', v) \, \ud u''}  \dvr^{-1} \rd_{v} \bb( \frac{2 m \dur}{(1-\mu) r^{2}} \bb) (u', v) \, \ud u'.
\end{equation*}
As before, the integration factor can be bounded by \eqref{eq:est4shift} in Lemma~\ref{lem:no-shift} and the bound \eqref{eq:btstrp:geom} on the geometry i.e.,
\begin{equation*}
	e^{- \int_{u}^{u'} \frac{2 m \dur}{(1-\mu) r^{2}} \, \ud u'} \leq \frac{3}{2}.
\end{equation*}
Therefore it only remains to prove
\begin{equation*}
\int_{u_{0}}^{u} \abs{\dvr^{-1} \rd_{v} \bb( \frac{2 m \dur}{(1-\mu) r^{2}} \bb) (u', v)} \, \ud u' \aleq \eps^{2},
\end{equation*}
uniformly in $u$, which in turn is a quick consequence of \eqref{eq:est4dv-shift}. \qedhere
\end{proof}

Since $\dvr^{-1} \rd_{v} \log \dvr = \dvr^{-2} \rd_{v} \dvr$, the previous proposition gives an estimate for $\rd_{v} \dvr$. In turn, this estimate can be used bound $\rd_{v}^{2}(r \phi)$; indeed
\begin{equation*}
(\dvr^{-1} \rd_{v})^{2} (r \phi) = \dvr^{-2} \rd_{v}^{2} (r \phi) - (\dvr^{-2} \rd_{v} \dvr )\dvr^{-1} \rd_{v} (r \phi),
\end{equation*}
and we have estimates \eqref{eq:btstrp:wave:2} and \eqref{eq:btstrp:pf:dvrphi} for $(\dvr^{-1} \rd_{v})^{2} (r \phi)$ and $\dvr^{-1} \rd_{v} (r \phi)$, respectively. We record these bounds in the following corollary.
\begin{corollary} \label{cor:Est4dvdv}
For the global solution we have constructed for \eqref{eq:SSESF}, we have
\begin{equation*}
	\abs{\rd_{v}^{2} (r \phi)} \aleq \eps, \quad \abs{\rd_{v}^{2} r} \aleq \eps^{2},
\end{equation*}
where the implicit constant depends only on $\dc$.
\end{corollary}
\subsection{Estimates for $\rd_u$ derivatives of $r$ and $\phi$} \label{subsec:main.finite:du}
As before, let $\eps < \eps_{1} < 1$ and take $(r, \phi, m)$ to be a $C^{1}$ solution obeying \eqref{eq:btstrp:geom}--\eqref{eq:btstrp:wave:2} on $u \in [u_{0}, \infty)$. We now derive estimates for the outgoing wave $\dur^{-1} \rd_{u}(r \phi)$ and $(\dur^{-1} \rd_{u})^{2}(r \phi)$, as well as $\dur^{-1} \rd_{u} \log \dur$.

\begin{proposition}
 \label{prop:est4duphi}
 For the global solution we have constructed for \eqref{eq:SSESF}, we have the following estimates for the $\rd_u$ derivatives of $r$ and $\phi$:
 \begin{align}
 \label{eq:est4durphi}
 |\dur^{-1}\rd_u(r\phi)(u, v)|&\leq C\eps+C \eps^3\min\set{r^{2}, 1},\\
  \label{eq:est4duphi}
 |\dur^{-1}\rd_u \phi (u, v)|&\leq C\eps \min\set{1, r^{-1}}, \\
 \label{eq:est4dudurphi}
|(\dur^{-1}\rd_u)^2(r\phi)(u, v)|&\leq C \eps, \\
 \label{eq:est4dudur}
|\dur^{-1}\rd_u \log \dur (u, v)|&\leq C \eps^{2}.
 \end{align}
for some constant $C$ depending only on $\dc$.
\end{proposition}
\begin{proof}
We start with \eqref{eq:est4durphi}.  By the boundary condition $(\dvr^{-1} \rd_{v} - \dur^{-1} \rd_{u}) (r \phi) (u, u) = 0$ and \eqref{eq:btstrp:pf:dvrphi}, it follows that
\begin{equation} \label{eq:est4durphi:axis}
	\abs{\dur^{-1} \rd_{u} (r \phi)(u, u)} \aleq \eps.	
\end{equation}
Similar to the case of $\dvr^{-1}\rd_v(r\phi)$, the equation for $\rd_{v}(\dur^{-1} \rd_{u} (r \phi) )$ leads to the following integral formula for $\dur^{-1}\rd_u(r\phi)$:
\begin{align*}
	\dur^{-1} \rd_{u} (r \phi)(u, v)
	=& e^{- \int_{u}^{v} \frac{2 m \dvr}{(1-\mu) r^{2}}(u, v') \ud v'} \dur^{-1} \rd_{u} (r \phi)(u, u) \\
	&	+ \int_{u}^{v} e^{- \int_{v'}^{v} \frac{2 m \dvr}{(1-\mu) r^{2}}(u, v'') \ud v''}  \frac{2m \dvr}{(1-\mu) r^{2}} \phi (u, v') \, \ud v'.
\end{align*}
For any $v_1<v_2$, we have the identity
\begin{equation} \label{eq:no-shift:conj}
	\int_{v_{1}}^{v_{2}} \frac{2 m \dvr}{(1-\mu) r^{2}} (u, v) \, \ud v
	= \int_{v_{1}}^{v_{2}} \rd_{u} \log (-\dur) (u, v) \, \ud v = \log \frac{(-\dur)(u, v_{2})}{(-\dur)(u, v_{1})}.
\end{equation}
Then from the bound \eqref{eq:btstrp:geom}, we derive
\begin{equation}
 \label{eq:est4durphid}
 \begin{split}
& \hskip-2em
|\dur^{-1}\rd_u(r\phi)(u, v)-\bb( \frac{\dur(u, u)}{\dur(u, v)} \bb) \dur^{-1}\rd_{u}(r\phi)(u, u)|\\
\les & \int_{u}^{v} \frac{m}{r^{2}} |\phi|(u, v')\, \ud v'\\
\les & \eps^3\int_{u}^{v} \min \set{r(u, v'), r^{-1-2\dc}(u, v')}\, \ud v'\les  \eps^3\min\set{r^{2}, 1}.
 \end{split}
\end{equation}
Here we used estimate \eqref{eq:btstrp:pf:m} to control $\frac{2m}{r^{2}}$, \eqref{eq:btstrp:pf:phi} for $\abs{\phi}$ and \eqref{eq:est4noverr:Cu} of Lemma \ref{lem:est4noverr} to bound the integral.


Next, we turn to $\dur^{-1} \rd_{u} \phi$. Simply by commuting $\dur^{-1} \rd_{u}$ with $r$, we have
\begin{equation*}
	r \dur^{-1} \rd_{u} \phi = \dur^{-1} \rd_{u}(r \phi) - \phi.
\end{equation*}
Hence by \eqref{eq:btstrp:pf:phi} and \eqref{eq:est4durphi}, we obtain
\begin{equation*}
	\abs{r \dur^{-1} \rd_{u} \phi} \aleq \eps,
\end{equation*}
which is favorable away from the axis.

To obtain the uniform boundedness of $\dur^{-1} \rd_u \phi$ near the axis, one option is to use an averaging formula for $\dur^{-1} \rd_{u} \phi$ (as in the proof of \eqref{eq:btstrp:pf:dvphi} for $\dvr^{-1} \rd_{v} \phi$), which relies on proving a bound for $(\dur^{-1} \rd_{u})^{2} (r \phi)$. Alternatively, using the
bound that we already have for $\dvr^{-1}\rd_v\phi$, we can derive the desired uniform bound of $\dur^{-1} \rd_u\phi$ near the axis from the previous estimate \eqref{eq:est4durphi}.
Indeed, commuting $r$ and $\rd_{u}$ in estimate \eqref{eq:est4durphid}, we can show that
\begin{equation}
\label{eq:est4duphirough}
\begin{split}
|r \dur^{-1}\rd_u\phi (u, v)|&\les \eps^3\min\set{1, r^{2}(u, v)}+|\dur^{-1}(u, v)||(\dur\phi)(u, v)-(\dur\phi)(u, u)|\\
                         &\les \eps^3\min\set{1, r^{2}(u, v)}+\int_{u}^{v}|\rd_v(\dur\phi)(u, v')|\, \ud v'\\
                         &\les \eps^3\min\set{1, r^{2}(u, v)}+\int_{u}^{v}(|\phi\frac{2m\dur\dvr}{(1-\mu)r^2}|+|\dur\rd_v\phi|)(u, v')\, \ud v'\\
                         &\les \eps^3\min\set{1, r^{2}(u, v)}+\int_{u}^{v} \frac{m}{ r^{2}} |\phi|(u, v') \, \ud v' + \int_{u}^{v} |\dvr^{-1}\rd_v\phi|\, \dvr (u, v') \, \ud v'\\
                         &\les \eps^3\min\set{1, r^{2}(u, v)}+\eps r(u,v).
\end{split}
 \end{equation}
Here we have used the equation \eqref{eq:SSESF} on the geometry $\rd_v\dur=\rd_v\rd_u r$, the estimate for the integral of $mr^{-2}\phi$, which has been carried out in the previous estimate \eqref{eq:est4durphi}, and \eqref{eq:btstrp:pf:dvphi} for $\dvr^{-1} \rd_{v} \phi$. Dividing by $r$ on both sides, it follows that
\begin{equation*}
	\abs{\dur^{-1} \rd_{u} \phi} \aleq \eps,
\end{equation*}
which proves \eqref{eq:est4duphi}.


Finally, we prove the bounds \eqref{eq:est4dudurphi} and \eqref{eq:est4dudur} for $(\dur^{-1} \rd_{u})^{2} (r \phi)$ and $\dur^{-1} \rd_{u} \log \dur$, respectively. As before, one may proceed in analogy with the cases of $(\dvr^{-1} \rd_{v})^{2} (r \phi)$ (Proposition~\ref{prop:Est4dvdvrphi}) and $\dvr^{-1} \rd_{v} \log \dvr$ (Proposition~\ref{prop:Est4dvdvr}), using averaging formulae near the axis and commutation of $r$ and $\dur^{-1} \rd_{u}$ far away. However, for the sake of simplicity, we take a more direct route here, exploiting the bound \eqref{eq:est4duphi} that is already closed.

We begin by bounding the data for $(\dur^{-1} \rd_{u})^{2} (r \phi)$ and $\dur^{-1} \rd_{u} \log \dur$ on the axis. By the regularity of $(r, \phi, m)$, it follows that
\begin{equation*}
	(\dvr^{-1} \rd_{v} - \dur^{-1} \rd_{u})^{2}(r \phi) (u, u) = 0, \quad
	(\dvr^{-1} \rd_{v} - \dur^{-1} \rd_{u})^{2} r (u, u) = 0.
\end{equation*}
By the boundary condition \eqref{eq:bc4r}, the wave equations for $\phi$ and $r$ and the estimates proved so far, we see that the mixed derivative terms involving $(\dvr^{-1} \rd_{v}) (\dur^{-1} \rd_{u})$ and $(\dur^{-1} \rd_{u}) (\dvr^{-1} \rd_{v})$ vanish. By \eqref{eq:btstrp:wave:2} and \eqref{eq:est4dvdvr}, it follows that
\begin{align*}
	\abs{(\dur^{-1} \rd_{u})^{2} (r \phi)(u, u)} =& \abs{(\dvr^{-1} \rd_{v})^{2}(r \phi)(u, u)} \aleq \eps, \\
	\abs{\dur^{-1} \rd_{u} \log \dur(u, u)} = & \abs{\dvr^{-1} \rd_{v} \log \dvr(u, u)} \aleq \eps^{2}.
\end{align*}
Next, a direct computation shows that
\begin{align*}
	\rd_{v} \bb( (\dur^{-1} \rd_{u})^{2} (r \phi) \bb)
	=& - \frac{4 m \dvr}{(1-\mu) r^{2}} (\dur^{-1} \rd_{u})^{2} (r \phi)  \\
	& + \frac{6 m \dvr}{(1-\mu) r^{2}} \dur^{-1} \rd_{u} \phi - \frac{\dvr}{(1-\mu)} r (\dur^{-1} \rd_{u} \phi)^{3}, \\
	\rd_{v} \bb( \dur^{-1}\rd_{u} \log \dur \bb)
	=& - \frac{2 m \dvr}{(1-\mu) r^{2}} \dur^{-1} \rd_{u} \log \dur \\
	& - \frac{4m \dvr}{(1-\mu) r^{3}} + \frac{\dvr}{1-\mu} (\dvr^{-1} \rd_{u} \phi)^{2}.
\end{align*}
These lead to integral formulae for $(\dur^{-1} \rd_{u})^{2} (r \phi)$ and $\dur^{-1} \rd_{u} \log \dur$, with integrating factors that are bounded by \eqref{eq:no-shift:conj}. Using \eqref{eq:btstrp:geom}, \eqref{eq:btstrp:pf:m} and \eqref{eq:est4duphi}, we may estimate
\begin{equation*}
 \begin{split}
|(\dur^{-1}\rd_u)^2(r\phi)(u, v)|&\les |(\dur^{-1}\rd_u)^2(r\phi)(u, u)|
						+|\int_{u}^{v} \frac{6 m \dvr}{(1-\mu) r^{2}} \dur^{-1} \rd_{u} \phi (u, v')\, \ud v'|\\
                                 &\phantom{\les |(\dur^{-1}\rd_u)^2(r\phi)(u, u)|}
                                 			 +|\int_{u}^{v} \frac{\dvr}{(1-\mu)} r (\dur^{-1} \rd_{u} \phi)^{3} (u, v')\, \ud v'|\\
                                 &\les \eps + \eps^3 \int_{u}^{v}\min\set{r, \frac{1}{r^{2 + \dc}}}(u, v')\, \ud v' +\eps^3 \int_{u}^{v}\min\set{r, \frac{1}{r^2}}(u, v')\, \ud v'\\
                                 &\les \eps,
 \end{split}
 \end{equation*}
 as well as
 \begin{align*}
\abs{\dur^{-1} \rd_{u} \log \dur (u, v)}
\aleq & \abs{\dur^{-1} \rd_{u} \log \dur(u, u)}
		+ \abs{\int_{u}^{v} \frac{4m \dvr}{(1-\mu) r^{3}}(u, v') \, \ud v' } \\
	& \phantom{\abs{\dur^{-1} \rd_{u} \dur(u, u)}}
		+ \abs{\int_{u}^{v} \frac{\dvr}{1-\mu} (\dur^{-1} \rd_{u} \phi)^{2}(u, v') \, \ud v'} \\
\aleq & \eps^{2} + \eps^{2} \int_{u}^{v} \min \set{1, \frac{1}{r^{2 + \dc}}} (u, v') \, \ud v' + \eps^{2} \int_{u}^{v} \min \set{1, \frac{1}{r^{2}}}(u, v') \, \ud v' \\
\aleq & \eps^{2}.
\end{align*}
This completes the proof of estimates \eqref{eq:est4dudurphi} and \eqref{eq:est4dudur}. \qedhere
\end{proof}

%

By an argument similar to that for Corollary~\ref{cor:Est4dvdv}, we obtain the following bounds on $\rd_{u}^{2}(r \phi)$ and $\rd_{u}^{2} r$ from Proposition~\ref{prop:est4duphi}.
\begin{corollary} \label{cor:Est4dudu}
For the global solution we have constructed for \eqref{eq:SSESF}, we have
\begin{equation*}
	\abs{\rd_{u}^{2} (r \phi)} \aleq \eps, \quad \abs{\rd_{u}^{2} r} \aleq \eps^{2},
\end{equation*}
where the implicit constant depends only on $\dc$.
\end{corollary}
This completes the proof of the second order derivative bounds for $r \phi$ and $r$ stated in Theorem~\ref{thm:main.finite}.

\subsection{Estimates for higher derivatives of $r$ and $\phi$} \label{subsec:high-d}
Finally, we derive estimates for higher derivatives of $r$ and $\phi$, which are \emph{uniform} with respect to choice of an initial null curve $C_{u_{0}}$. These estimates require an additional $C^{2}$ regularity assumption on the initial data $\Phi$, as well as possibly taking $\eps$ smaller compared to a constant depending only on $\dc$.
They will be crucial for passing to the limit $u_{0} \to - \infty$ in the next section.

Given $u_{0} \in \bbR$ and $\Phi$ satisfying \eqref{eq:IDcond} with $\eps \leq \eps_{1}$, let $(r, \phi, m)$ be the global $C^{1}$ solution to \eqref{eq:SSESF} that we have constructed earlier. Suppose furthermore that $\Phi$ belongs to $C^{2}$.
By a routine persistence of regularity argument, it follows that $\log \dvr$, $\log \dur$, $\dvr^{-1} \rd_{v}(r \phi)$ and $\dur^{-1} \rd_{u} (r \phi)$ are $C^{2}$ on their domain; in short, we will say that $(r, \phi, m)$ is a global $C^{2}$ solution to \eqref{eq:SSESF}. Our goal is to show that, by taking $\eps$ smaller if necessary, the $C^{2}$ norm of these variables obeys a uniform bound independent of $u_{0}$. A more precise statement is as follows.

\begin{proposition} \label{prop:high-d}
Given $v_{0} \in [u_{0}, \infty)$, let
\begin{equation*}
	A = \sup_{v \in [u_{0}, v_{0}]} \abs{(\dvr_{0}^{-1} \rd_{v} )^{2} \Phi(v)}
\end{equation*}
where $\dvr_{0} = \frac{1}{2}$, and let $\calD(u_{0}, v_{0})$ be the domain of dependence of the curve $\set{u_{0}} \times [u_{0}, v_{0}]$, i.e.,
\begin{equation*}
	\calD(u_{0}, v_{0}) = \set{(u, v) \in \PD : u \in [u_{0}, v_{0}], \, v \in [u, v_{0}]}.
\end{equation*}
There exists a constant $\eps_{2} > 0$, which is independent of $u_{0}, v_{0}$ and $A$, such that if $\eps \leq \eps_{2}$ then we have the uniform bounds
\begin{align}
\label{eq:est4dvdvdvrphi}
	\sup_{\calD(u_{0}, v_{0})} \abs{(\dvr^{-1} \rd_{v})^{3} (r \phi)} \aleq & A + \eps, \\
\label{eq:est4dvdvdvr}	
	\sup_{\calD(u_{0}, v_{0})} \abs{(\dvr^{-1} \rd_{v})^{2} \log \dvr} \aleq & \eps A + \eps^{2}, \\
\label{eq:est4dududurphi}
	\sup_{\calD(u_{0}, v_{0})} \abs{(\dur^{-1} \rd_{u})^{3} (r \phi)} \aleq & A + \eps, \\
\label{eq:est4dududur}
	\sup_{\calD(u_{0}, v_{0})} \abs{(\dur^{-1} \rd_{u})^{2} \log \dur} \aleq & \eps A + \eps^{2},
\end{align}
with an implicit constant independent of $u_{0}$, $v_{0}$, $A$ and $\eps$.
\end{proposition}

We begin by establishing \eqref{eq:est4dvdvdvrphi} and \eqref{eq:est4dvdvdvr}. As in the proofs of Propositions~\ref{prop:Est4dvdvrphi} and \ref{prop:Est4dvdvr}, the key step is to bound
\begin{equation} \label{eq:dvdv-shift}
	(\dvr^{-1} \rd_{v})^{2} \bb( \frac{2 m \dur}{(1-\mu) r^{2}} \bb).
\end{equation}
To achieve this end, we need a few preliminary estimates for $\dvr^{-1} \rd_{v}$ derivatives of $\phi$ and $m$. The ensuing computation is somewhat tedious, but the principle is simple: We rely on the differentiated averaging formulae (see Lemma~\ref{lem:avg-d}) to derive estimates which are favorable near the axis $\set{r = 0}$, whereas we simply commute $r$ with $\dvr^{-1} \rd_{v}$ in the region $\set{r \ageq 1}$ away from the axis.
\begin{lemma} \label{lem:est4dvdvdv:prelim}
For the global $C^{2}$ solution considered above, the following estimates hold.
\begin{align}
\label{eq:est4dvdvphi}
	\abs{(\dvr^{-1} \rd_{v})^{2} \phi (u, v)}
	\aleq & \sup_{v' \in [u, v]} \abs{(\dvr^{-1} \rd_{v})^{3} (r \phi)(u, v')} ,	\\
\label{eq:est4rdvdvdvphi}
	\abs{r \, (\dvr^{-1} \rd_{v})^{3} \phi (u, v)}
	\aleq & \sup_{v' \in [u, v]} \abs{(\dvr^{-1} \rd_{v})^{3} (r \phi)(u, v')} ,	\\
\label{eq:est4dvdv-2m}
	\abs{(\dvr^{-1} \rd_{v})^{2} (2m)}
	\aleq &  \eps^{2}, \\
\label{eq:est4dvdv-2m-over-r2}
	\abs{(\dvr^{-1} \rd_{v})^{2}\bb( \frac{2m}{r^{2}} \bb)(u,v)}
	\aleq & \eps \sup_{v' \in [u, v]} \abs{(\dvr^{-1} \rd_{v})^{3} (r \phi)(u, v')} + \eps^{4}, \\
\label{eq:est4dvdv-2m-over-r2:large-r}
	\abs{r^{2} \, (\dvr^{-1} \rd_{v})^{2}\bb( \frac{2m}{r^{2}} \bb)(u,v)}
	\aleq & \eps^{2}, \\
\label{eq:est4dvdv-dur-over-1-mu}
	\abs{(\dvr^{-1} \rd_{v})^{2} \bb( \frac{\dur}{1 - \mu} \bb)(u, v)}
	\aleq & \eps^{2} \min \set{1, \frac{1}{r}}.
\end{align}
\end{lemma}

\begin{proof}
Since it is rather routine, we will only sketch the proof of each estimate, specifying the relevant computation and previous bounds needed.

Estimate \eqref{eq:est4dvdvphi} follows directly from taking $(\dvr^{-1}\rd_{v})^{2}$ of the averaging formula \eqref{eq:avg:phi} for $\phi$ using Lemma~\ref{lem:avg-d} and bounding the resulting term using \eqref{eq:avg-est}.

For \eqref{eq:est4rdvdvdvphi}, we simply commute $r$ with $(\dvr^{-1} \rd_{v})^{3}$ to arrive at the formula
\begin{equation*}
	r (\dvr^{-1} \rd_{v})^{3} \phi
	= (\dvr^{-1} \rd_{v})^{3}(r \phi) - 3 (\dvr^{-1} \rd_{v})^{2} \phi,
\end{equation*}
from which \eqref{eq:est4rdvdvdvphi} follows using \eqref{eq:est4dvdvphi}.

For \eqref{eq:est4dvdv-2m}, we compute
\begin{align*}
	(\dvr^{-1} \rd_{v})^{2} (2m)
	= & (\dvr^{-1} \rd_{v}) \bb( (1-\mu) r^{2} (\dvr^{-1} \rd_{v} \phi)^{2} \bb) \\
	= & (1- \dvr^{-1} \rd_{v} (2m)) r (\dvr^{-1} \rd_{v} \phi)^{2} + (1-\mu) r (\dvr^{-1} \rd_{v} \phi)^{2} \\
	& + 2 (1-\mu) r^{2} \dvr^{-1} \rd_{v} \phi (\dvr^{-1} \rd_{v}) ^{2} \phi,
\end{align*}
then use \eqref{eq:btstrp:pf:dvphi}, \eqref{eq:btstrp:pf:rdvdvphi} and \eqref{eq:est4dv-m} to estimate the right-hand side.

To prove \eqref{eq:est4dvdv-2m-over-r2}, we first use Lemma~\ref{lem:avg-d} to take $(\dvr^{-1} \rd_{v})^{2}$ of the averaging formula \eqref{eq:avg:2m-over-r2}, which leads to
\begin{equation*}
	\abs{(\dvr^{-1} \rd_{v})^{2} \bb( \frac{2m}{r^{2}} \bb)(u, v)}
	\aleq \sup_{v' \in [u, v]} \abs{(\dvr^{-1} \rd_{v})^{2} \bb( (1-\mu) r (\dvr^{-1} \rd_{v} \phi)^{2} \bb)(u, v')}.
\end{equation*}
Expanding the right-hand side, we obtain the formula
\begin{align*}
	(\dvr^{-1} \rd_{v})^{2} \bb( (1-\mu) r (\dvr^{-1} \rd_{v} \phi)^{2} \bb)
	= & 2 (1-\mu) \dvr^{-1} \rd_{v} \phi \, r (\dvr^{-1} \rd_{v})^{3} \phi \\
		& + 2 (1-\mu) r \bb( (\dvr^{-1} \rd_{v})^{2} \phi \bb)^{2} \\
	& + 2 \bb(1- \dvr^{-1} \rd_{v} (2m) \bb) \dvr^{-1} \rd_{v} \phi (\dvr^{-1} \rd_{v})^{2} \phi \\
	& - \bb( (\dvr^{-1} \rd_{v})^{2} (2m) \bb) (\dvr^{-1} \rd_{v} \phi)^{2}.
\end{align*}
Then the desired estimate follows using \eqref{eq:btstrp:pf:dvphi}, \eqref{eq:btstrp:pf:rdvdvphi}, \eqref{eq:est4dv-m}, \eqref{eq:est4dvdvphi}, \eqref{eq:est4rdvdvdvphi} and \eqref{eq:est4dvdv-2m}.


For \eqref{eq:est4dvdv-2m-over-r2:large-r}, we commute $r^{2}$ with $\dvr^{-1} \rd_{v}$ and obtain
\begin{align*}
	r^{2} (\dvr^{-1} \rd_{v})^{2} \bb( \frac{2m}{r^{2}} \bb)
	= & (\dvr^{-1} \rd_{v})^{2} (2m) - 4 \frac{\dvr^{-1} \rd_{v} (2m)}{r} + \frac{12m}{r^{2}}.
\end{align*}
Then the desired estimate follows from \eqref{eq:btstrp:pf:m}, \eqref{eq:est4dv-m} and \eqref{eq:est4dvdv-2m}.

Finally, for \eqref{eq:est4dvdv-dur-over-1-mu}, we first compute
\begin{align*}
	(\dvr^{-1} \rd_{v})^{2} \bb( \frac{\dur}{1-\mu} \bb)
	=& (\dvr^{-1} \rd_{v}) \bb( \frac{\dur}{1-\mu} r (\dvr^{-1} \rd_{v} \phi)^{2} \bb) \\
	=& \frac{\dur}{1-\mu} r^{2} (\dvr^{-1} \rd_{v} \phi)^{4}
	 + \frac{\dur}{1-\mu} (\dvr^{-1} \rd_{v} \phi)^{2}
	 + 2 \frac{\dur}{1-\mu} r \dvr^{-1} \rd_{v} \phi (\dvr^{-1} \rd_{v})^{2} \phi,
\end{align*}
and then use \eqref{eq:btstrp:geom}, \eqref{eq:btstrp:pf:dvphi}, \eqref{eq:btstrp:pf:rdvdvphi} to estimate the right-hand side.
\end{proof}

As an immediate consequence of Lemma~\ref{lem:est4dvdvdv:prelim}, we have
\begin{align*}
	\abs{(\dvr^{-1} \rd_{v})^{2} \bb( \frac{2m \dur}{(1-\mu) r^{2}} \bb)(u, v)}
	\aleq & \eps \sup_{v' \in [u, v]} \abs{(\dvr^{-1} \rd_{v})^{3} (r \phi) (u, v')} + \eps^{4}, \\
	\abs{(\dvr^{-1} \rd_{v})^{2} \bb( \frac{2m \dur}{(1-\mu) r^{2}} \bb)(u, v)}
	\aleq & \frac{\eps^{2}}{r^{2}}.
\end{align*}
These can be combined to a single slightly weaker but more convenient bound as follows.
\begin{corollary} \label{cor:est4dvdv-shift}
For the global $C^{2}$ solution considered above, we have
\begin{align}
\label{eq:est4dvdv-shift}
	\abs{(\dvr^{-1} \rd_{v})^{2} \bb( \frac{2m \dur}{(1-\mu) r^{2}} \bb)(u, v)}
	\aleq & \eps \bb( \eps + \sup_{v' \in [u, v]} \abs{(\dvr^{-1} \rd_{v})^{3} (r \phi) (u, v')} \bb) \min \set{1, \frac{1}{r^{2}}}.
\end{align}
\end{corollary}

We are ready to establish \eqref{eq:est4dvdvdvrphi} and \eqref{eq:est4dvdvdvr}. The proofs are similar to those of Propositions~\ref{prop:Est4dvdvrphi} and \ref{prop:Est4dvdvr}, respectively.
\begin{proof} [Proof of \eqref{eq:est4dvdvdvrphi} and \eqref{eq:est4dvdvdvr}]
We first prove \eqref{eq:est4dvdvdvrphi}. Commuting $(\dvr^{-1} \rd_{v})^{2}$ with the equation \eqref{eq:wave4dvrphi} using \eqref{eq:comm-du-dv}, we obtain
\begin{align*}
	\rd_{u} \bb( (\dvr^{-1} \rd_{v})^{3} (r \phi) \bb)
	= & - \frac{6 m \dur}{(1-\mu) r^{2}} (\dvr^{-1} \rd_{v})^{3} (r \phi) \\
	& - \bb( \dvr^{-1} \rd_{v} \bb( \frac{6 m \dur}{(1-\mu) r^{2}} \bb) \bb) (\dvr^{-1} \rd_{v})^{2} (r \phi) \\
	& - \bb( (\dvr^{-1} \rd_{v})^{2} \bb( \frac{2 m \dur}{(1-\mu) r^{2}} \bb) \bb) (\dvr^{-1} \rd_{v}) (r \phi) \\
	& + (\dvr^{-1} \rd_{v})^{2} \bb( \frac{2m \dur}{(1-\mu) r^{2}}  \phi \bb) \\
	=: & - \frac{6 m \dur}{(1-\mu) r^{2}} (\dvr^{-1} \rd_{v})^{3} (r \phi) + N_{3}.
\end{align*}

As in the proof of Proposition~\ref{prop:Est4dvdvrphi}, we may derive an integral formula for $(\dvr^{-1} \rd_{v})^{3} (r \phi)$, where the integration factor is uniformly bounded by \eqref{eq:est4shift}. Then we have
\begin{equation*}
	\abs{(\dvr^{-1} \rd_{v})^{3} (r \phi)(u, v)}
	\leq \bb( \frac{3}{2} \bb)^{3} A + \bb( \frac{3}{2} \bb)^{3} \int_{u_{0}}^{u} \abs{N_{3}(u', v)} \, \ud u'.
\end{equation*}
For $(u, v) \in \calD(u_{0}, v_{0})$, we claim that
\begin{equation} \label{eq:est4dvdvdvrphi:key}
	\int_{u_{0}}^{u} \abs{N_{3}(u', v)} \, \ud u'
	\aleq \eps^{2} \bb( \eps + \sup_{\calD(u_{0}, v_{0})} \abs{(\dvr^{-1} \rd_{v})^{3} (r \phi)} \bb).
\end{equation}
Once \eqref{eq:est4dvdvdvrphi:key} is proved, the desired estimate \eqref{eq:est4dvdvdvrphi} follows by taking $\eps > 0$ sufficiently small and absorbing the term $\eps^{2} \sup_{\calD(u_{0}, v_{0})} \abs{(\dvr^{-1} \rd_{v})^{3} (r \phi)}$ into the left-hand side.

To establish \eqref{eq:est4dvdvdvrphi:key}, we first expand $N_{3}$ as
\begin{align*}
N_{3}
= 	& - \bb( \dvr^{-1} \rd_{v} \bb( \frac{6 m \dur}{(1-\mu) r^{2}} \bb) \bb) (\dvr^{-1} \rd_{v})^{2} (r \phi)
	- \bb( (\dvr^{-1} \rd_{v})^{2} \bb( \frac{2 m \dur}{(1-\mu) r^{2}} \bb) \bb) (\dvr^{-1} \rd_{v}) (r \phi) \\
	&+ \bb( (\dvr^{-1} \rd_{v})^{2} \bb( \frac{2m \dur}{(1-\mu) r^{2}} \bb)  \bb) \phi
	+ 2\bb( \dvr^{-1} \rd_{v} \bb( \frac{2m \dur}{(1-\mu) r^{2}} \bb) \bb) \dvr^{-1} \rd_{v} \phi + \frac{2 m \dur}{(1-\mu) r^{2}} (\dvr^{-1} \rd_{v})^{2} \phi.
\end{align*}
Then using \eqref{eq:btstrp:wave:2}, \eqref{eq:btstrp:pf:dvrphi}, \eqref{eq:btstrp:pf:phi}, \eqref{eq:btstrp:pf:dvphi}, \eqref{eq:est4shift:imp}, \eqref{eq:est4dv-shift} and \eqref{eq:est4dvdv-shift}, the desired estimate \eqref{eq:est4dvdvdvrphi:key} follows.

 Next, to establish \eqref{eq:est4dvdvdvr}, we first commute $(\dvr^{-1} \rd_{v})^{2}$ with the equation \eqref{eq:wave4dvr} using \eqref{eq:comm-du-dv} to obtain
\begin{align*}
	\rd_{u} \bb( (\dvr^{-1} \rd_{v})^{2} \log \dvr \bb) =& - \frac{4 m \dur}{(1-\mu) r^{2}} (\dvr^{-1} \rd_{v})^{2} \log \dvr \\
								& - \bb( \dvr^{-1} \rd_{v} \bb(\frac{2 m \dur}{(1-\mu) r^{2}}\bb) \bb) \dvr^{-1} \rd_{v} \log \dvr
								+ (\dvr^{-1} \rd_{v})^{2} \bb( \frac{2 m \dur}{(1-\mu) r^{2}} \bb) \\
	=: & - \frac{4 m \dur}{(1-\mu) r^{2}} (\dvr^{-1} \rd_{v})^{2} \log \dvr + M_{3}.
\end{align*}
Since $(\dvr^{-1} \rd_{v})^{2} \log \dvr (u_{0}, v) = 0$ thanks to the initial gauge condition $\dvr (u_{0}, v) = \frac{1}{2}$, we have
\begin{equation*}
	\abs{(\dvr^{-1} \rd_{v})^{2} \log \dvr(u, v)} \leq \bb(\frac{3}{2} \bb)^{2} \int_{u}^{u_{0}} \abs{M_{3}(u', v)} \, \ud u',
\end{equation*}
where we again used \eqref{eq:est4shift} to bound the integration factor. By \eqref{eq:est4dv-shift}, \eqref{eq:est4dvdvr} and \eqref{eq:est4dvdv-shift}, as well as \eqref{eq:est4dvdvdvrphi} that we just proved, we have	
\begin{equation*}
	\int_{u_{0}}^{u} \abs{M_{3}(u', v)} \, \ud u' \aleq \eps (\eps + A)
\end{equation*}
which proves \eqref{eq:est4dvdvdvr}.
\end{proof}

It remains to prove \eqref{eq:est4dududurphi} and \eqref{eq:est4dududur}. This can be done by a similar argument as in the proofs of\eqref{eq:est4dvdvdvrphi} and \eqref{eq:est4dvdvdvr}, with the roles of $u$ and $v$ interchanged. To avoid repetition, we only sketch the argument.

\begin{proof} [Sketch of proof of \eqref{eq:est4dududurphi} and \eqref{eq:est4dududur}]
As in Lemma~\ref{lem:est4dvdvdv:prelim}, we can prove that
 \begin{align}
\label{eq:est4duduphi}
	\abs{(\dur^{-1} \rd_{u})^{2} \phi (u, v)}
	\aleq & \sup_{u' \in [u, v]} \abs{(\dur^{-1} \rd_{u})^{3} (r \phi)(u', v)} ,	\\
\notag
	\abs{r \, (\dur^{-1} \rd_{u})^{3} \phi (u, v)}
	\aleq & \sup_{u' \in [u, v]} \abs{(\dur^{-1} \rd_{u})^{3} (r \phi)(u', v)}.
\end{align}
Proceeding as in the proofs of Lemmas~\ref{lem:Est4shift}, \ref{lem:est4dv-shift} and Corollary~\ref{cor:est4dvdv-shift}, we also obtain
 \begin{align*}
 	\abs{\frac{2 m \dvr}{(1-\mu) r^{2}}  (u, v)} \aleq & \eps^{2} r_{+}^{-1 -\dc}, \\
	\abs{\dur^{-1} \rd_{u} \bb( \frac{2 m \dvr}{(1-\mu) r^{2}} \bb) (u, v)} \aleq & \eps^{2} \min\set{1, \frac{1}{r^{2}}}, \\
	\abs{(\dur^{-1} \rd_{u})^{2} \bb( \frac{2 m \dvr}{(1-\mu) r^{2}} \bb) (u, v)} \aleq & \eps \bb( \eps + \sup_{u' \in [u, v]} \abs{(\dur^{-1} \rd_{u})^{3} (r \phi)(u', v)} \bb) \min \set{1, \frac{1}{r^{2}}}.
\end{align*}
Furthermore, since $\rd_{v} r$, $\rd_{u} r$, $\rd_{v}(r \phi)$ and $\rd_{u}(r \phi)$ are $C^{2}$ up to the axis $\Gmm = \set{u = v}$, we have
\begin{equation*}
	(\dvr^{-1} \rd_{v} - \dur^{-1} \rd_{u})^{3} (r \phi) (u, u) = 0, \quad
	(\dvr^{-1} \rd_{v} - \dur^{-1} \rd_{u})^{3} r (u, u) = 0.
\end{equation*}
Then by the wave equations for $r$ and $\phi$, as well as \eqref{eq:est4dvdvdvrphi}--\eqref{eq:est4dvdvdvr}, we obtain
\begin{align*}
\abs{(\dur^{-1} \rd_{u})^{3} (r \phi) (u, u)} \aleq & A + \eps, \\
\abs{(\dur^{-1} \rd_{u})^{2} \log \dur (u, u)} \aleq & \eps A + \eps^{2}.
\end{align*}
Commuting $(\dur^{-1} \rd_{u})^{2}$ with \eqref{eq:wave4durphi} and \eqref{eq:wave4dur}, estimating the initial data at $v = u$ by the preceding bounds and estimating the inhomogeneous terms using the earlier bounds, the desired estimates \eqref{eq:est4dududurphi} and \eqref{eq:est4dududur} follow as in the proofs of \eqref{eq:est4dvdvdvrphi} and \eqref{eq:est4dvdvdvr}. \qedhere
\end{proof}


\section{Forward- and backward-in-time global solution}\label{sec.proof.inf}
The goal of this section is to deduce Theorem~\ref{thm:main} from Theorem~\ref{thm:main.finite} and Proposition~\ref{prop:high-d}. The proof of the causal geodesic completeness assertions are again postponed to Section~\ref{sec:cgc}.

Let $\Phi$ be a $C^{2}$ function on $\bbR$ satisfying the hypothesis of Theorem~\ref{thm:main}. Define also \footnote{The finiteness of course follows also from the hypothesis of Theorem~\ref{thm:main}.}
\begin{equation*}
	A := \sup_{v \in \bbR} \abs{(\dvr_{0}^{-1} \rd_{v})^{2} \Phi(v)} < \infty,
\end{equation*}
where $\dvr_{0} = \frac{1}{2}$. Consider a sequence $u_{n} \in \bbR$ tending to $-\infty$ as $n \to \infty$. For each $n = 1, 2, \ldots$, let $(r^{(n)}, \phi^{(n)}, m^{(n)})$ be the solution of \eqref{eq:SSESF} with $\dvr^{(n)} \restriction_{C_{u_{n}}} = \rd_{v} r^{(n)} \restriction_{C_{u_{n}}} = \frac{1}{2}$ and
\begin{equation*}
	(\dvr^{(n)})^{-1} \rd_{v}(r^{(n)} \phi^{(n)})(u_{n}, v) = \Phi(v).
\end{equation*}
Let $\eps > 0$ be sufficiently small (depending on $\dc > 0$), so that Theorem~\ref{thm:main.finite} applies to each solution $(r^{(n)}, \phi^{(n)}, m^{(n)})$. Our aim now is to show that $(r^{(n)}, \phi^{(n)}, m^{(n)})$ tends to a solution $(r, \phi, m)$ that obeys the conclusions of Theorem~\ref{thm:main}.

As a consequence of Theorem~\ref{thm:main.finite}, Proposition~\ref{prop:high-d} and the estimates \eqref{eq:btstrp:pf:dvphi}, \eqref{eq:est4duphi}, \eqref{eq:est4dvdvphi} and \eqref{eq:est4duduphi}, for $\eps>0$ sufficiently small we have the uniform bounds
\begin{align*}
	\sum_{k=0}^{2} \sup_{\PD_{[u_{n}, \infty)}} \bb( \abs{\rd_{v}^{k} \phi^{(n)}} + \abs{\rd_{u}^{k} \phi^{(n)}} \bb)
	\aleq& A + \eps, \\
	\sum_{k=0}^{2} \sup_{\PD_{[u_{n}, \infty)}} \bb( \abs{\rd_{v}^{k+1}(r^{(n)} \phi^{(n)})} + \abs{\rd_{u}^{k+1}(r^{(n)} \phi^{(n)})} \bb)
	\aleq& A + \eps, \\
	\sum_{k=0}^{2} \sup_{\PD_{[u_{n}, \infty)}} \bb( \abs{\rd_{v}^{k+1} r^{(n)}} + \abs{\rd_{u}^{k+1} r^{(n)}} \bb)
	\aleq& \eps(A + \eps).
\end{align*}
Uniform bounds for a corresponding number of mixed derivatives follow from the wave equations for $\phi$ and $r$. Using the equations for $m$ (see also the bounds \eqref{eq:btstrp:pf:m}, \eqref{eq:est4dv-m} and \eqref{eq:est4dvdv-2m}), it also follows that the $C^{2}$ norm of $m^{(n)}$ is uniformly bounded on any compact subset of $\PD_{[u_{n}, \infty)}$. Hence by the Arzela-Ascoli theorem, there exists a limit $(r, \phi, m)$ on $\PD$ such that $r = m = 0$ on $\Gmm$, and
\begin{align*}
	\bb( \phi, \rd_{v} (r^{(n)} \phi^{(n)}), \rd_{u} (r^{(n)} \phi^{(n)}) \bb)
	 \to& \, \bb( \phi, \rd_{v} (r \phi), \rd_{u} (r \phi) \bb) \quad \hbox{ in } C^{1}(\Omg), \\
	r^{(n)} \to&\,  r \hbox{ in } C^{2}(\Omg), \\
	 m^{(n)} \to& \, m \quad \hbox{ in } C^{1}(\Omg),
\end{align*}
on every compact subset $\Omg \subseteq \PD$. By this convergence, it is clear that $(r, \phi, m)$ solves \eqref{eq:SSESF} in the classical sense.
Moreover, the a priori bounds we have proved in the finite $u_{0}$ case (e.g., \eqref{eq:main:geo} and \eqref{eq:main:apriori}) still hold for the limiting solution $(r, \phi, m)$, as long as they are uniform in $u_{0}$.

It remains to justify that the limiting solution $(r, \phi, m)$ assumes $\Phi$ as the data on the past null infinity.
More precisely, we claim that
\begin{equation} \label{eq:main:pf:id}
	\lim_{u \to -\infty} \dvr (u, v) = \frac{1}{2}, \quad
	\lim_{u \to -\infty} \dvr^{-1} \rd_{v} (r \phi)(u, v) = \Phi(v),
\end{equation}
for every $v \in (-\infty, \infty)$.

Recalling the proof of Proposition~\ref{prop:dvr:im}, for any $u \geq u_{n}$ we have
\begin{equation*}
	\abs{\log \dvr^{(n)}(u, v)- \log \frac{1}{2}} \leq \abs{\int_{u_{n}}^{u} \frac{2 m^{(n)} \dur^{(n)}}{(1-\mu^{(n)}) (r^{(n)})^{2}} (u', v) \, \ud u'} \aleq \eps^{2} \left(\max\{(v-u),1\}\right)^{-\dc}.
\end{equation*}
Taking the limit $n \to \infty$ first and then letting $u \to -\infty$, we obtain the desired statement for $\dvr$. Similarly, proceeding as in the proof of Proposition~\ref{prop:Est4nonlin},
\begin{equation*}
	\abs{(\dvr^{(n)})^{-1} \rd_{v} (r^{(n)} \phi^{(n)}) (u, v) -  \frac{\dvr^{(n)}(u_{n}, v)}{\dvr^{(n)}(u, v)}\Phi(v) }
	\aleq \eps^{3} \left(\max\{(v-u),1\}\right)^{-2\dc},
\end{equation*}
where we recall that $\dvr^{(n)}(u_{n}, v) = \frac{1}{2}$. Taking the limits $n \to \infty$ and $u \to -\infty$ in order as before, we obtain \eqref{eq:main:pf:id}.

\begin{remark}
As a byproduct of the Arzela-Ascoli theorem, observe that $\rd_{v}^{2} (r \phi)$, $\rd_{u}^{2} (r \phi)$, $\rd_{v}^{2} r$ and $\rd_{u}^{2} r$ are Lipschitz.
Their weak derivatives obey the bounds
\begin{align*}
	\esssup_{\PD} \bb( \abs{\rd_{v}^{3}(r \phi)} + \abs{\rd_{u}^{3}(r \phi)} \bb) \aleq  A + \eps, \quad
	\esssup_{\PD} \bb( \abs{\rd_{v}^{3} r} + \abs{\rd_{u}^{3} r} \bb) \aleq  \eps(A + \eps).
\end{align*}
\end{remark}

\section{Proof of Corollary \ref{cor.infinite.BV.mass}}\label{sec.infinite.BV.mass}
In this section, we prove Corollary \ref{cor.infinite.BV.mass}, i.e., we show that the initial data $\Phi$ can be chosen to satisfy the assumptions of Theorem \ref{thm:main.finite} while at the same time having infinite BV norm and infinite Bondi mass.

\begin{proof}[Proof of Corollary \ref{cor.infinite.BV.mass}]
Let $\chi:\mathbb R\to [0,1]$ be a non-negative smooth bump function such that $\chi$ is compactly supported in $[1,6]$ and $\chi(x)=1$ for $x\in [3,4]$. For some $\ep'>0$ to be fixed later, let $\Phi$ be defined by the following sum of translated bump functions:
\begin{equation}\label{Phi.infinite.def}
\Phi(v):=\eps'\sum_{k=3}^\infty \chi(v-u_0-2^k).
\end{equation}
We will show that $\Phi$ satisfies the assumptions of Theorem \ref{thm:main.finite} and have infinite BV norm and initial Bondi mass.

\pfstep{Step~1:~Verifying \eqref{eq:IDcond}} Since $k\geq 3$, for every $v$ at most one term in the sum \eqref{Phi.infinite.def} is non-zero. Therefore, for every $\ep>0$, one can choose $\ep'>0$ sufficiently small such that $|\Phi(v)|+|\Phi'(v)|\leq \ep$. This gives the second condition in \eqref{eq:IDcond}.

Fix any $\gamma>0$. Given $u$ and $v$ such that $u_0\leq u\leq v$, we consider two cases. If $v-u\leq 5$, then we just use the bound
$$\int_{u}^v \Phi(v') \,\ud v'\leq \ep'(v-u)\leq 5^{\gamma}\ep'(v-u)^{1-\gamma}.$$
If $v-u>5$, we use the fact that the support of at most $(\log_2 \lfloor v-u \rfloor)+100$ bumps intersect the interval $(u,v)$. Since $\int_{-\infty}^{\infty} \chi(x) \, \ud x\leq 5$, we thus have
$$\int_{u}^v \Phi(v') \, \ud v'\leq 5\ep'\left(\log_2 \lfloor v-u \rfloor+100\right)\leq C_{\gamma}\ep'(v-u)^{1-\gamma}$$
for some $C_{\gamma}>0$ depending only on $\gamma$ as long as $\gamma<1$. In both cases, the first condition in \eqref{eq:IDcond} is satisfied after choosing $\ep'$ to be sufficiently small depending on $\ep$ and $\gamma$.

\pfstep{Step~2:~Infinite BV norm}
We now show that $\Phi$ gives rise to data with infinite BV norm. To this end, one observes that $\int_{-\infty}^\infty |\chi'(x)| \, \ud x\geq c_{\chi}$ for some $c_{\chi}>0$. Hence
$$\lim_{v\to \infty}\int_{u_0}^v |\rd_v\Phi|(v') \, \ud v'\geq  \ep' \lim_{N\to \infty}(\sum_{i=0}^N c_{\chi})=\infty.$$

\pfstep{Step~3:~Infinite Bondi mass}
Finally, we prove that the data have infinite initial Bondi mass, i.e., the limit of the Hawking mass as $v\to \infty$ is infinite. We first recall that the Hawking mass obeys the following equation
$$\rd_v m=\f 12 r^2(1-\f{2m}{r})\f{(\rd_v\phi)^2}{\lambda},$$
i.e.,
$$\rd_v( m e^{\int_{u_0}^v r\f{(\rd_v\phi)^2}{\lambda}(v') \, \ud v'})=\f 12 r^2\f{(\rd_v\phi)^2}{\dvr}e^{\int_{u_0}^v r\f{(\rd_v\phi)^2}{\lambda}(v') \, \ud v'} .$$
This implies
\begin{equation}\label{m.data.formula}
 m(u_0, v) =\f 12 e^{-\int_{u_0}^v r\f{(\rd_v\phi)^2}{\lambda}(u_0,v') \, \ud v'}\int_{u_0}^v r^2\f{(\rd_v\phi)^2}{\lambda}(u_0,v') e^{\int_{u_0}^{v'} r\f{(\rd_v\phi)^2}{\lambda}(u_0,v'') \, \ud v''} \, \ud v' .
\end{equation}
To compute the limit as $v\to \infty$, we first write $\rd_v\phi$ in terms of $\Phi$. We then show that with the choice of $\Phi$ in \eqref{Phi.infinite.def}, $r\f{(\rd_v\phi)^2}{\lambda}$ is integrable, while $r^2\f{(\rd_v\phi)^2}{\lambda}$ is not, thus demonstrating that $\lim_{v\to \infty} m(u_0,v)=\infty$.

To compute $\rd_v\phi$, we note that
\begin{equation}\label{data.rdvphi}
\rd_v\phi(u_0,v)=\f{1}{r}\rd_v(r\phi)(u_0,v)-\f{\lambda\phi}{r}(u_0,v)
=\f \Phi{2r}(u_0,v) -\f{\int_{u_0}^v \Phi(v') \, \ud v'}{4r^2(u_0,v)}.
\end{equation}
In other words, using also the following condition on $C_{u_0}$
$$\rd_v r=\lambda=\f 12,\, r(u_0,u_0)=0 \implies r(u_{0}, v) =\f 12 (v-u_0),$$
we get
$$r(\rd_v\phi)(u_0,v)=\f {\Phi(u_0,v)} 2 -\f{\int_{u_0}^v \Phi(v') \, \ud v'}{2(v-u_0)}.$$
Therefore, for some $C>0$ independent of $\ep'$, we have
\begin{equation*}
\begin{split}
\lim_{v\to \infty}\int_{u_0}^v r\f{(\rd_v\phi)^2}{\lambda}&(v') \, \ud v' \leq C\lim_{v\to \infty}\left(\int_{u_0}^v \f{\Phi^2(u_0,v')}{v'-u_0} \, \ud v'+\int_{u_0}^v \f{(\int_{u_0}^{v'} \Phi(v'') \, \ud v'')^2}{(v'-u_0)^3} \, \ud v'\right)\\
\leq &C(\ep')^2\lim_{N\to \infty}\sum_{k=1}^N  2^{-k}+C(\ep')^2\lim_{v\to\infty}\int_{u_0 + 9}^v \f{(1+\log (v'-u_0+2))^2}{(v'-u_0)^3} \, \ud v'\leq C(\ep')^2.
\end{split}
\end{equation*}
On the last line, we used the bound $\int_{u_{0}}^{v'} \Phi(v'') \, \ud v'' \aleq \eps' (1 + \log (v' - u_{0} + 2))$, as well as the fact that $\Phi(v) = 0$ for $u_{0} \leq v \leq u_{0} + 9$ by definition.
Moreover, in a similar manner, we have\footnote{Notice that up to a constant factor, this is the contribution to the integral of $r^2(\rd_v\phi)^2$ by the second term in \eqref{data.rdvphi}.}
$$\lim_{v\to \infty}\int_{u_0}^v \f{(\int_{u_0}^{v'} \Phi(v'') \, \ud v'')^2}{(v'-u_0)^2} \, \ud v'\leq C(\ep')^2.$$
Therefore, by \eqref{m.data.formula}, we obtain
\begin{equation}\label{m.data.lower.bound}
m(v)\geq c\int_{u_0}^v r^2(\rd_v\phi)^2(u_0,v') \, \ud v'\geq c\int_{u_0}^v \Phi^2(u_0,v') \, \ud v' -C(\ep')^2
\end{equation}
for some $0<c<C$. On the other hand, by a similar argument as the proof of the infinitude of the BV norm, we get
\begin{equation}\label{m.lower.bound.limit}
\lim_{v\to\infty}\int_{u_0}^v \Phi^2(u_0,v') \, \ud v'\to \infty.
\end{equation}
Combining \eqref{m.data.lower.bound} and \eqref{m.lower.bound.limit} gives
$$m(v)\to \infty$$
as $v\to \infty$, as is to be proved.
\end{proof}

\begin{remark}
We note that for the data given by \eqref{Phi.infinite.def} in the proof of Corollary \ref{cor.infinite.BV.mass}, the global solution that arises from the data (which exists by Theorem \ref{thm:main.finite}) in fact has infinite BV norm on each $C_{u}$, as well as infinite Bondi mass everywhere along future null infinity. More precisely, for every $u\geq u_0$, we have
\begin{equation*}
\lim_{v \to \infty} \int_{u}^{v} \abs{(\dvr^{-1} \rd_{v})^{2} (r \phi) (u, v')} \, \ud v' = \infty
,\quad
\lim_{v\to \infty} m(u,v)=\infty.
\end{equation*}

To establish the infinitude of the BV norm, note first that by \eqref{eq:dvdvrphi-formula} and \eqref{int.factor.bd}, we have
\begin{equation*}
	\int_{u}^{v} \abs{(\dvr^{-1} \rd_{v})^{2} (r \phi)(u, v')} \, \ud v'
	\geq \int_{u}^{v} \abs{(\dvr^{-1} \rd_{v}) \Phi(v')} \, \ud v'
		- \bb(\frac{3}{2} \bb)^{2} \int_{u}^{v} \int_{u_{0}}^{u} \abs{N_{2}(u', v')} \, \ud u' \, \ud v'
\end{equation*}
where
\begin{equation*}
	N_{2} = - \dvr^{-1} \rd_{v} \bb( \frac{2m \dur}{(1-\mu) r^{2}} \bb) \dvr^{-1} \rd_{v} (r \phi) + \dvr^{-1} \rd_{v} \bb( \frac{2m \dur}{(1-\mu) r^{2}} \phi \bb).
\end{equation*}
Note furthermore that $\Phi$ as given by \eqref{Phi.infinite.def} satisfies the assumptions of Theorem \ref{thm:main.finite} with any $\gamma\in (0,1)$. Using estimates \eqref{eq:btstrp:pf:dvrphi}, \eqref{eq:btstrp:pf:phi}, \eqref{eq:btstrp:pf:dvphi}, \eqref{eq:est4shift:imp} and \eqref{eq:est4dv-shift}, as well as exploiting the explicit form of $\abs{\Phi(v)}$ in \eqref{Phi.infinite.def}, it can be shown that
\begin{equation*}
	\sup_{u, v : u_{0} \leq u \leq v < \infty} \int_{u}^{v} \int_{u_{0}}^{u} \abs{N_{2}(u', v')} \, \ud u' \, \ud v'
	\aleq \int_{u_{0}}^{\infty} \int_{u_{0}}^{v'} \bb( \eps^{3} r_{+}^{-2-\gmm} + \eps^{2} r_{+}^{-2} \abs{\Phi(v')} \bb) \, \ud u' \, \ud v' < \infty.
\end{equation*}
On the other hand, since $\lim_{v \to \infty} \int_{u}^{v} \abs{(\dvr^{-1} \rd_{v}) \Phi(v')} \, \ud v' = \infty$ as in the proof of Corollary~\ref{cor.infinite.BV.mass}, the desired conclusion follows.

Next, to see that the Bondi mass is infinite everywhere along future null infinity, we again apply Theorem~\ref{thm:main.finite}, but now with $\gmm  > \f 12$. Then according to \eqref{eq:dvphi:large-r}, \eqref{eq:dvr}, \eqref{eq:btstrp:pf:phi}, \eqref{rdvrphi.diff.est} and \eqref{eq:btstrp:pf:geom:dvr}, it can be seen that that the main contribution to the Bondi mass is given by $\lim_{v\to\infty}\int_{u}^v \Phi^2(u,v') \, \ud v'$ in a similar manner as in the proof of Corollary \ref{cor.infinite.BV.mass}. Therefore, one can argue as in the proof of Corollary \ref{cor.infinite.BV.mass} to show that the Bondi mass is infinite for every $u\geq u_0$.
\end{remark}


\section{Causal geodesic completeness} \label{sec:cgc}
In this section, we complete the proof of Theorems~\ref{thm:main.finite} and \ref{thm:main} by establishing causal geodesic completeness of the solutions constructed in Sections~\ref{sec.proof} and \ref{sec.proof.inf}. We will first show the \underline{future} causal geodesic completeness of these spacetimes. Once this is achieved, it is easy to see that the past causal geodesic completeness for solutions constructed in Theorem~\ref{thm:main} can be proven in an almost identical manner. We will return to this at the end of the section.

\subsection{Geodesics in $\calM$}

Let $\gmm : I \to \calM$ (where $I$ is an interval) be a future pointing causal geodesic on the spacetime $\calM$ constructed in Theorems~\ref{thm:main.finite} or \ref{thm:main}. Given any function $f$ on $\calM$, we adopt the convention of denoting by $f(s)$ the value of $f$ at the point $\gmm(s)$ i.e., $f(s) = f(\gmm(s))$. We also write $\dot{f}(s) = \frac{\ud}{\ud s} f(s)$ and $\ddot{f}(s) = \frac{\ud^{2}}{\ud s^{2}} f(s)$.

In order to describe the geodesic $\gmm$, it is convenient to use the double null coordinates $(u, v, \tht, \varphi)$ whenever possible, since then we can directly use the bounds in Theorem \ref{thm:main.finite}. Under our convention, we may write
\begin{equation*}
	\gmm(s) = (u(s), v(s), \tht(s), \varphi(s)), \quad \dot{\gmm}(s) = (\dot{u}(s), \dot{v}(s), \dot{\tht}(s), \dot{\varphi}(s)).
\end{equation*}
Let us note that these are only defined away from the axis $\Gmm$. On the other hand, it is easy to verify that in fact $v(s)$ and $u(s)$ can be extended to continuous functions in $\calM$. (Notice that in contrast, $\dot{u}(s)$ and $\dot{v}(s)$ may be discontinuous.)

As $\gamma$ is future pointing causal, we have 
\[
\dot u\geq 0,\quad \dot v\geq 0.
\]
We now discuss conserved quantities of a geodesic. We denote by $\bfC^{2}$ (the minus of) the magnitude of the 4-velocity $\dot{\gmm}(s)$, i.e.,
\begin{equation*}
	\bfC^{2}(s) = - g_{\alp \bt} \,{\!\rst}_{\gmm(s)} \dot{\gmm}^{\alp}(s) \dot{\gmm}^{\bt}(s),
\end{equation*}
where we recall the metric
\begin{equation*}
	g = - \Omg^{2} \, \ud u \cdot \ud v + r^{2} (\ud \tht^{2} + \sin^{2} \tht \, \ud \varphi^{2})
\end{equation*}
on $\calM$.
The quantity $\bfC^{2}$ is conserved (i.e., constant in $s$). The choice of the sign is so that $\bfC^{2} > 0$ when $\gmm$ is a time-like geodesic. In the double null coordinates, it takes the form
\begin{equation} \label{eq:c2-dnull:al}
	\bfC^{2} = \Omg^{2} \dot{u} \dot{v} - r^{2} (g_{\bbS^{2}, \tht \tht} \dot{\tht} \dot{\tht} + g_{\bbS^{2}, \varphi \varphi} \dot{\varphi} \dot{\varphi}).
\end{equation}
Since the spacetime $(\calM, g)$ is spherically symmetric, conservation of angular momentum holds for geodesics. Let $\Omg_{x}$, $\Omg_{y}$, $\Omg_{z}$ be the standard generators of the rotation group $\mathrm{SO}(3)$ (i.e., infinitesimal rotations about the $x$-, $y$-, and $z$-axes). Let
\begin{equation*}
	J_{k} (s)= g_{\alp \bt} \!\rst_{\gmm(s)} \Omg_{k}^{\alp} \!\rst_{\gmm(s)} \dot{\gmm}^{\bt}(s), \quad k=x,\; y \textnormal{ or } z.
\end{equation*}
It can be easily verified that $J_{x}$, $J_{y}$, $J_{z}$ are conserved. We define the (conserved) total angular momentum squared as
\begin{equation*}
	\bfJ^{2} := J_{x}^{2} + J_{y}^{2} + J_{z}^{2}.
\end{equation*}
In the double null coordinates, $\bfJ^{2}$ takes the form
\begin{equation} \label{eq:j2-dnull:al}
	\bfJ^{2} = r^{4} \bb( g_{\bbS^{2}, \tht \tht} \dot{\tht} \dot{\tht} + g_{\bbS^{2}, \varphi \varphi} \dot{\varphi} \dot{\varphi} \bb).
\end{equation}
This statement is an immediate consequence of the identity
\begin{equation*}
	\Omg_{x} \otimes \Omg_{x} + \Omg_{y} \otimes \Omg_{y} + \Omg_{z} \otimes \Omg_{z}
	= \rd_{\tht} \otimes \rd_{\tht} + \frac{1}{\sin^{2} \tht} \rd_{\varphi} \otimes \rd_{\varphi}
\end{equation*}
which concerns only the standard sphere $\bbS^{2}$. This identity, in turn, can be verified by observing that each side defines a contravariant symmetric 2-tensor on $\bbS^{2}$ which is invariant under rotations (i.e., Lie derivatives with respect to $\Omg_{x}$, $\Omg_{y}$, $\Omg_{y}$ vanish) and yields $1$ when tested against $\ud \varphi \otimes \ud \varphi$ on the equator $\tht = \frac{\pi}{2}$; such a tensor is clearly unique.

By \eqref{eq:c2-dnull:al} and \eqref{eq:j2-dnull:al}, we obtain the useful identity
\begin{equation} \label{eq:dudv-dnull:al}
	\Omg^{2} \dot{u} \dot{v} = \bfC^{2} + r^{-2} \bfJ^{2}
\end{equation}
where $\bfC^{2}$ and $\bfJ^{2}$ are conserved.

 A basic tool for studying completeness of future pointing causal geodesics is the following lemma.
\begin{lemma}[Continuation of future pointing causal geodesics] \label{lem:cont-crit}
Any future pointing causal geodesic $\gmm : [0, s_{f}) \to \calM$ can be continued past $s_{f}$ if there exists a compact subset $K \subseteq \calM$ such that $\set{\gmm(s) \in \calM : s \in [0, s_{f})} \subseteq K$.
\end{lemma}

\begin{proof}
First, observe that it suffices to consider a future causal geodesic $\gmm$ whose image intersects $\calM \setminus C_{u_{0}}$. Otherwise, the image of $\gmm$ is contained in $C_{u_{0}}$. Since the unique future pointing causal vector tangent to $C_{u_{0}}$ is its null generator, $\gmm$ is a radial null geodesic contained in $C_{u_{0}}$, which is complete thanks to uniform boundedness of $\abs{\log \Omg}$ in \eqref{eq:main:geo}.

Without loss of generality, we may assume that $\gmm(0) \in \calM \setminus C_{u_{0}}$. Since $\gmm$ is future pointing causal, the closure of its image $\gmm([0, s_{f})) = \set{\gmm(s) \in \calM : s \in [0, s_{f})}$ is disjoint from $C_{u_{0}}$. Then by the compactness assumption, it follows that there exists a sequential limit point $p \in \calM \setminus C_{u_{0}}$ of $\gmm([0, s_{f}))$, i.e., there exists a sequence $s_{n} \to s_{f}$ such that $\gmm(s_{n}) \to p$.

Recall now the standard result that there exists a geodesically convex neighborhood around any non-boundary point in a smooth Lorentzian manifold. Let $\calU$ be a geodesically convex neighbordhood of $p \in \calM \setminus C_{u_{0}}$. By definition, there exists $s' \in [0, s_{f})$ such that $\gmm(s') \in \calU$; since $\gmm$ is a future pointing causal geodesic, it follows that the $\gmm([s', s_{f})) \subseteq \calU$.
Then $\gmm$ can be continued as the unique geodesic in $\calU$ passing through $\gmm(s')$ and $p$.
\end{proof}

\subsection{Preliminary discussions}
Our strategy for proving geodesic completeness is to argue by contradiction, i.e., we assume that there is a future pointing causal geodesic $\gamma$ which is not complete and terminates at some finite time $s_f$ and derive a contradiction (to Lemma \ref{lem:cont-crit} or otherwise). The following is the simplest case: 
\begin{lemma}\label{C.J.not.0}
If $\gamma(s): [0, s_f)\mapsto \calM$ is incomplete, then either $\bfC\neq 0$ or $\bfJ\neq 0$.
\end{lemma}
\begin{proof}
If $\bfC=0$ and $\bfJ=0$, then the geodesic $\gamma$ is a spherically symmetric constant $u$ curve or constant $v$ curve. These geodesics are  complete since $\f{1}{\Omg^2}\rd_u$ and $\f{1}{\Omg^2}\rd_v$ are geodesic vector fields and $|\log\Omg|$ is uniformly bounded. 
\end{proof}

Before proceeding to the other cases, we need some preliminary considerations. First, we will see that some difficulties arise because the $(u,v,\theta,\varphi)$ coordinate system is not regular at the axis. It is therefore useful to have the following:
\begin{lemma}\label{r.0.discrete}
If $\gamma(s): [0, s_f)\mapsto \calM$ is incomplete, then the set $\{s: r(s)=0\}$ is a discrete subset of $[0,s_f)$ (with a possible accumulation point at $s_f$).
\end{lemma}
\begin{proof}
By standard existence and uniqueness theory for ODEs, it suffices to show that the axis $\Gamma$ is a complete geodesic. (Indeed, then if $\{s: r(s)=0\}$ is not a discrete subset of $[0,s_f)$, then the image of $\gamma$ coincides with $\Gamma$ and contradicts the incompleteness of $\gamma$.) To see that $\Gamma$ is a complete geodesic, we note that $\lambda\restriction_\Gmm=-\nu\restriction_\Gmm$, $\rd_v\lambda\restriction_\Gmm=-\rd_u\nu\restriction_\Gmm$ together with $\f{m}{r^2}\restriction_{\Gmm}=0$ imply that $\rd_v\log\Omg\restriction_\Gmm=\rd_u\log\Omg\restriction_\Gmm$. Then by an explicit calculation, one checks that the vector field $\f{1}{\Omg}(\rd_v+\rd_u)\restriction_\Gmm$, which is tangent to $\Gmm$, is a geodesic vector field. Since $|\log\Omg|$ is uniformly bounded, $\Gmm$ is a complete geodesic.
\end{proof}

Consider the following (smooth) quantity
\begin{equation*}
E(s) =- g_{\alp \bt} \!\rst_{\gmm(s)} T^{\alp} \!\rst_{\gmm(s)} \dot{\gmm}^{\bt}(s),
\end{equation*}
where $T$ is a smooth vector field on $\calM$ which is given in the $(u, v, \theta, \varphi)$ coordinate system by
\begin{equation} \label{eq:T-uv}
T =  - \frac{2 \dur}{\Omg^{2}} \rd_{v} + \frac{2 \dvr}{\Omg^{2}} \rd_{u}.
\end{equation}
Observe that $T$ is radial, future time-like, and tangent to the constant $r$ hypersurfaces, i.e., $T r = 0$. In the $(u, v, \theta, \varphi)$ coordinates, $E$ takes the form
\begin{equation}\label{E.uv}
	E(s) := \dvr \dot v(s)-\dur \dot u(s).
\end{equation}
In particular, this shows that away from the axis $\Gamma$, $E(s)$ is non-negative as $\dvr$, $\dot v$, $-\dur$, $\dot u$ are non-negative. 

It will be useful to have the following slightly stronger statement:
\begin{lemma}\label{r.positive}
Let $\gamma(s): [0, s_f)\mapsto \calM$ be a future causal geodesic with either $\bfJ\neq 0$ or $\bfC\neq 0$, then $E(s)>0$ for $s\in [0, s_f)$.
\end{lemma}
\begin{proof}
\emph{Case 1. $\bfJ\neq 0$.} In this case, $r(s)>0$ for $s\in [0, s_f)$ and hence we can use the $(u,v,\theta,\varphi)$ coordinate system. If $E(s)=0$, then by \eqref{E.uv}, $\dot{v}(s)=\dot{u}(s)=0$. By \eqref{eq:dudv-dnull:al}, $0=\Omg^2(\dot{u}\dot{v})(s)\geq r^{-2}(s)\bfJ^2>0$, which is a contradiction.

\emph{Case 2. $\bfC\neq 0$.} If $r(s)>0$, the proof proceeds in the same way as in Case 1. If $r(s)=0$ and $E(s)=0$, by Lemma \ref{r.0.discrete}, there exists a sequence $\{s_i\}$ with $s_i\to s$ such that $r(s_i)\to 0$ and $E(s_i)\to 0$. On the other hand, by \eqref{E.uv}, \eqref{eq:dudv-dnull:al} and the boundedness of $|\log\Omg|$, $E(s_i)\gtrsim \sqrt{\dot{v}\dot{u}}(s_i)\gtrsim \bfC$, which is uniformly bounded below. This is again a contradiction.
\end{proof}

Recall that 
\begin{equation}\label{rd.uv}
\dot r(s)=\dvr\dot v(s)+\dur\dot u(s).
\end{equation}
Our analysis heavily relies on the evolution equations for $E(s)$, $\dot r(s)$ and $\gamma(s)$.
\begin{lemma} \label{lem:geod-eq-v:al}
Let $\gmm$ be a geodesic on $\calM$. If $\gmm(s) = (u, v, \tht, \varphi)(s)$ lies outside the axis $\Gmm$, then
\begin{equation} \label{eq:geod-eq-v}
\begin{split}
	 \ddot{v}(s)& = - \Omg^{-2} \rd_{v} \Omg^{2} (s) \dot{v}^{2}(s) - 2 r^{-3} \dur \Omg^{-2} (s) \bfJ^{2},\\
    \ddot {u}(s)&=-\Omg^{-2}\rd_u\Omg^2(s) \dot {u}^2(s)-2r^{-3}\dvr\Omg^{-2}(s)\bfJ^2,\\
    \dot E(s)&=-r\dot{v}^2(\rd_v\phi)^2 +r\dot{u}^2(\rd_u\phi)^2,\\
    \ddot {r}(s)&=-r\dot{v}^2(\rd_v\phi)^2 -r\dot{u}^2(\rd_u\phi)^2-4r^{-3}\dur\dvr\Omg^{-2}(s)\bfJ^2+2\rd_u\dvr \dot u\dot v.
\end{split}
\end{equation}
\end{lemma}
\begin{proof}
In a coordinate patch, recall that the geodesic equation reads
$\ddot{\gmm}^{\lmb} = - \Gmm_{\alp \bt}^{\lmb} \dot{\gmm}^{\alp} \dot{\gmm}^{\bt}$.
Hence in order to find the equation for $\ddot{v}$, it suffices to compute the Christoffel symbols of the form $\Gmm_{\alp \bt}^{v}$.
By explicit computation, it can be verified that
\begin{align*}
	\Gmm^{v}_{vv}
	= \Omg^{-2} \rd_{v} \Omg^{2}, \quad
	\Gmm^{v}_{\tht \tht}
	= 2 r^{-3} \dur \Omg^{-2} r^{4} g_{\bbS^{2}, \tht \tht}, \quad
	\Gmm^{v}_{\varphi \varphi}
	= 2 r^{-3} \dur \Omg^{-2} r^{4} g_{\bbS^{2}, \varphi \varphi},
\end{align*}
while all the other components vanish. Recalling the identity \eqref{eq:j2-dnull:al}, the equation for $\ddot v$ follows. Similarly we have the equation for 
$\ddot u$.

Next we use the equations for $\dot u$, $\dot v$ to derive the evolution equations for $\dot r$ and $E$. By using the Raychaudhuri equations \eqref{eq:raych4v}, \eqref{eq:raych4u}, we can compute that 
\begin{align*}
 \frac{d}{ds}(\dvr \dot v\pm\dur \dot u)&=\dvr\ddot v+\dot \dvr\dot{ v}\pm(\dur \ddot{u}+\dot \dur \dot{u})\\
 &=\dot{v}(\dot u \rd_u\dvr+ \dot v \rd_v\dvr)-\dvr\left(\Omg^{-2}\rd_v\Omg^2(s) \dot {v}^2(s)+2r^{-3}\dur\Omg^{-2}(s)\bfJ^2\right)\\
 &\quad \pm\dot{u}(\dot u \rd_u\dur + \dot v \rd_v \dur)\mp\dur\left(\Omg^{-2}\rd_u\Omg^2(s) \dot {u}^2(s)+2r^{-3}\dvr\Omg^{-2}(s)\bfJ^2\right)\\
 &=-\dot{v}^2 \dvr \rd_v \log \abs{\frac{\dur}{1-\mu}} \mp \dot{u}^2\dur \rd_u\log \abs{\frac{\dvr}{1-\mu}} +(1\pm 1)(-2r^{-3}\dur\dvr\Omg^{-2}(s)\bfJ^2+\rd_u\dvr \dot u\dot v)\\
 &=-r\dot{v}^2(\rd_v\phi)^2 \mp r\dot{u}^2(\rd_u\phi)^2+(1\pm 1)(-2r^{-3}\dur\dvr\Omg^{-2}(s)\bfJ^2+\rd_u\dvr \dot u\dot v).
 \end{align*}
 Here note that $\rd_u\dvr=\rd_v \dur$. By the definition of $E(s)$ and the equation for $\dot r(s)$, the plus case leads to the equation for $\dot r(s)$ while the minus case gives the equation for $E(s)$. 
\end{proof}

\subsection{Basic features of incomplete geodesics in $\calM$}
Now a very basic feature of an incomplete geodesic is that the quantity $E(s)$ blows up.
\begin{lemma} \label{lem:cont-crit-e:al}
If $\gamma(s): [0, s_f)\mapsto \calM$ is incomplete with $\bfJ\neq 0$ or $\bfC\neq 0$, then 
\begin{equation}\label{E.lower.bd}
 E(s)\geq \frac{C}{s_f-s},\quad \forall s<s_f
\end{equation}
for some constant $C$ depending only on the constants in Theorem \ref{thm:main.finite}.
\end{lemma}
\begin{proof}
\pfstep{Step~1} We first claim that 
 \[
  \limsup\limits_{s\rightarrow s_f}E(s)=\infty.
 \]
To see this, first note the estimate $\dot{u}(s)+\dot{v}(s)\leq CE(s)$ for some $C>0$, which is a consequence of \eqref{E.uv}, and holds away from the axis. By Lemma \ref{r.0.discrete} and using the continuity of $u(s)$ and $v(s)$ (which also holds at the axis), we thus deduce that if $\limsup\limits_{s\rightarrow s_f} E(s)$ is bounded, then $u$, $v$ are bounded.
In particular the geodesic $\gamma$ lies in a compact set in $\calM$. By Lemma \ref{lem:cont-crit}, the geodesic can be continued beyond $s_f$ which contradicts the assumption. 

\pfstep{Step~2} We next make use of the evolution equation for $E(s)$ obtained in the previous lemma. The bounds \eqref{eq:main:apriori} on $\phi$ imply that $ r(\rd_u\phi)^2+r(\rd_v\phi)^2$ is bounded above. Therefore, for any $s$ such that $r(s)>0$, we have
\begin{align}\label{E.dot.est}
 \dot E(s)\leq C_*(\dot u^2+\dot v^2)(s)\leq 36 C_*(\dvr \dot v-\dur \dot u)^2(s)=36C_*E^2(s)
\end{align}
for some constant $C_*>0$. We now divide into two cases. 

\emph{Case 1. There exists $s_0\in [0, s_f)$ such that $r(s)>0$ whenever $s\geq s_0$.} Let $s_0<s_*<s_{**}<s_f$. Integrating \eqref{E.dot.est} from $s_*$ to $s_{**}$, we get
\begin{equation}\label{Ei.diff}
 E(s_*)^{-1}-E(s_{**})^{-1}\leq 36C_* (s_{**}-s_{*}).
\end{equation}
Notice that this makes sense thanks to Lemma \ref{r.positive}. Taking $\liminf_{s_{**}\to s_f}$ and using Step 1, we thus obtain
$$E(s_*)^{-1}\leq 36C_* (s_f-s_{*})$$
for every $s_* \in (s_0,s_f)$, as desired\footnote{We note of course that by choosing $C$ large if necessary, we only need to obtain \eqref{E.lower.bd} for $s$ sufficiently close to $s_f$.}.

\emph{Case 2. There exists a sequence $\{s_k\}$ with $s_k\to s_f$ such that $r(s_k)=0$.} By Lemma \ref{r.0.discrete}, we can assume that $r(s)>0$ if $s\neq s_k$. In this case we need to be slightly more careful since \eqref{E.dot.est} only holds when $s\neq s_k$.

Let $s_*, s_{**}\in [0,s_f)$ be such that $s_{**}> s_{*}\geq s_2$ and let $k_*=\min\{k : s_k\geq s_*\}$ and $k_{**}=\max\{k : s_k\leq s_{**}\}$. Assume that $k_{**}>k_*$. We then compute
\begin{equation}
\begin{split}
&\left(E(s_*)^{-1}-E(s_{k_*})^{-1}\right)+\left(E(s_{k_{**}})^{-1}-E(s_{**})^{-1}\right)+\sum_{k=k_*}^{k_{**}-1}\left(E(s_k)^{-1}-E(s_{k+1})^{-1}\right)\\
\leq &36C_*\left(s-s_{k_*}+\sum_{k=k_*}^{k_{**}-1} (s_{k+1}-s_k)\right).
\end{split}
\end{equation}
This again leads to \eqref{Ei.diff} and gives the desired conclusion as in Case 1.

\end{proof}
Another feature of future causal incomplete geodesics is that they approach the axis (at least along a sequence of times). More precisely, we have
\begin{lemma} \label{lem:leave:al}
If $\gamma(s): [0, s_f)\mapsto \calM$ is incomplete, then for any $r_0>0$ and any $s_0\in [0,s_f)$, there exists $s\in[s_0, s_f)$ such that $r(s)<r_0$.
\end{lemma}
\begin{proof}
 Suppose not, i.e., we assume that $r(s) \geq r_{0}$ for all $s\in [s_0, s_f)$ and some constants $r_0>0$ and $s_0\in [0,s_f)$. Consider the geodesic equation \eqref{eq:geod-eq-v} for $\dot v$. We can write
\begin{align*}
	- \Omg^{-2} \rd_{v} \Omg^{2} \dot{v}^{2} - 2 r^{-3} \dur \Omg^{-2} \bfJ^{2}
	= & - (\Omg^{-2} \rd_{v} \Omg^{2} \dot{v} + \Omg^{-2} \rd_{u} \Omg^{2} \dot{u}) \dot{v} \\
	& + \Omg^{-4} \rd_{u} \Omg^{2} (\bfC^{2} + r^{-2} \bfJ^{2})- 2 r^{-3} \dur \Omg^{-2} \bfJ^{2}.
\end{align*}
Hence, we have $\ddot{v} (s) = - \dot{R}(s) \dot{v}(s) + F(s)$, where
\begin{align*}
	R = \log \Omg^{2}, \quad
	F = \Omg^{-4} \rd_{u} \Omg^{2} (\bfC^{2} + r^{-2} \bfJ^{2})- 2 r^{-3} \dur \Omg^{-2} \bfJ^{2}.
\end{align*}
It follows that
\begin{equation*}
	\frac{\ud}{\ud s} (\Omg^{2} \dot{v})(s) = \Omg^{2} F(s).
\end{equation*}
By the bounds in Theorem \ref{thm:main.finite} (which also holds for solutions constructed in Theorem~\ref{thm:main}), as well as conservation of $\bfC^{2}$ and $\bfJ^{2}$, $\abs{\log \Omg}$ and $\abs{F}$ are uniformly bounded on $[s_0, s_{f})$. It follows that $\dot{v}$ is uniformly bounded. In particular 
$v$ is uniformly bounded. Since $u\leq v$, we derive that $\gamma(s)$ lies in a compact set in $\calM$. This contradicts Lemma ~\ref{lem:cont-crit}.
\end{proof}

\subsection{Proof of geodesic completeness}

We can now rule out the case when the geodesic is spherically symmetric.
\begin{lemma}
 \label{lem:J:al}
 Assume $\gamma(s): [0, s_f)\mapsto \calM$ is incomplete. Then $\bfJ\neq 0$.
\end{lemma}
\begin{proof}
 Assume for the sake of contradiction that $\bfJ=0$ (and by Lemma \ref{C.J.not.0}, we can assume without loss of generality that $\bfC\neq 0$). We consider the following cases (Notice that by Lemma \ref{r.0.discrete}, they exhaust all possibilities):

 \emph{Case 1. There exists a sequence $\{s_k\}$ with $s_k\to s_f$ such that $\dot{r}(s_k)=0$ and $r(s_k)>0$.} This condition, $\bfJ=0$ and \eqref{eq:dudv-dnull:al} together imply that $\dot{u}(s_k)$ and $\dot{v}(s_k)$ are uniformly bounded. This contradicts Lemma \ref{lem:cont-crit-e:al}.
 
\emph{Case 2. There exists $s_0\in [0,s_f)$ such that\footnote{Notice that if there exists a sequence $\{s_k'\}$ with $s_k'\to s_f$ such that $r(s_k')=0$, then by Lemma \ref{r.0.discrete} we are necessarily in Case $1$.} $\dot r(s)>0$ and $r(s)>0$, $\forall s\in [s_0,s_f)$.}  Without loss of generality, we may assume $r(s_0)>0$. Then $r(s)\geq r(s_0)$ on $[s_0, s_f)$ which contradicts Lemma \ref{lem:leave:al}.
 
\emph{Case 3. There exists $s_0\in [0,s_f)$ such that $\dot r(s)<0$ and $r(s)>0$, $\forall s\in [s_0,s_f)$.}  By definition and the bounds on $\dur$, $\dvr$, we have $\dot v\leq 4 \dot u$. Combine this with the uniform bound on $\dot u\dot v$ (which follows from \eqref{eq:dudv-dnull:al}). We conclude that $\dot v$ is uniformly bounded. In particular $v$ 
 is uniformly bounded. Since $u\leq v$, the geodesic $\gamma(s)$ lies in a compact set, which contradicts Lemma \ref{lem:cont-crit}. 
\end{proof}

It now remains to rule out the possibility of $\bfJ\neq 0$. It is convenient to note that in this case, we have $r(s)>0$ for $s\in [0,s_f)$. As a first step, we observe that if $\bfJ\neq 0$, then Lemma \ref{lem:geod-eq-v:al} implies that $\ddot u$, $\ddot v$ and $\ddot r$ have a particular sign if $\dot u$, $\dot v$ have size $r^{-1}$ and $r$ is sufficiently small.
\begin{lemma}\label{lem:uvr:sign}
There exists $r_0>0$ such that if $\bfJ\neq 0$ and at some time $s_0$
\[
\frac{1}{100}r^{-1}(s_0)\Omg^{-1}(s_0)\bfJ\leq \dot u(s_0), \quad\dot v(s_0)\leq 100 r^{-1}(s_0)\Omg^{-1}(s_0)\bfJ, \quad r(s_0)<r_0,
\]
then
\begin{align*}
&\ddot v(s_0)>0,\quad \ddot u(s_0)<0,\quad \ddot r(s_0)>0.
\end{align*}
\end{lemma}
\begin{proof}
The lemma follows from the equations \eqref{eq:geod-eq-v} for $\ddot u$, $\ddot v$, $\ddot r$ together with the bounds on the geometry from Theorem \ref{thm:main.finite}.
\end{proof}

Given Lemma \ref{lem:uvr:sign}, one sees that an incomplete geodesic with $\bfJ\neq 0$ cannot stay inside a small cylinder around the axis.
\begin{lemma} \label{lem:exit}
Assume $\gamma(s): [0, s_f)\mapsto \calM$ is incomplete and $\bfJ\neq 0$. Then there exists $r_0>0$ such that for every $s_0\in [0,s_f)$ the geodesic $\gamma(s)$ exits the cylinder with radius $r_0$ at some time to the future of $s_0$, that is, there exists $s_1\in (s_0,s_f)$ such that 
$r(s_1)> r_0$.
\end{lemma}
\begin{proof}
Since $\bfJ\neq 0$, the geodesic does not intersect the axis and in particular we can use the $(u,v,\theta,\varphi)$ coordinate system. 
Take $r_0$ be the constant in Lemma \ref{lem:uvr:sign}. We prove this lemma 
by a contradiction argument. Assume the geodesic $\gamma(s)$ lies in the cylinder with radius $r_0$ 
for all $s\in [s_0, s_f)$. 

\pfstep{Step~1}  We claim that
\begin{equation} 
 \label{eq:dotr:neg}
 \dot r(s)\leq 0,\quad \forall  s\in [s_0, s_f).
 \end{equation}
 Otherwise there exists $s'\in [s_0,s_f)$ such that $\dot r(s')>0$. Then by Lemma \ref{lem:leave:al}, there exists $s''>s'$ such that $\dot r(s'')<0$. Take
\[
s^*=\sup \{s: s' \leq s\leq s'',\quad \dot r(s)\geq 0\}.
\]
Then $\dot r(s^*)=0$, $\dot r(s)<0$ when $s^*<s\leq s''$. Recall that $\dot r(s^*)=\dur \dot u(s^*)+\dvr\dot v(s^*)$. Then the identity \eqref{eq:dudv-dnull:al} implies that $\dot u(s^*)$, $\dot v(s^*)$ satisfy
the conditions in Lemma \ref{lem:uvr:sign}. In particular, $\ddot r(s^*)>0$ which contradicts $\dot r(s)<0$ when $s^*<s\leq s''$. Hence \eqref{eq:dotr:neg} holds.

\pfstep{Step~2} 
Next we prove that there exists $t_0\in (s_0, s_f)$ such that
\begin{equation}\label{eq:dotu:dotv:s0}
\frac{1}{10}\dot u(t_0) \leq \dot v(t_0).
\end{equation}

Otherwise $\frac{1}{10}\dot u>\dot v$ for all $s \in (s_0,s_f)$ which implies that
\[
\dot r(s)=(\dvr \dot v+\dur \dot u)(s)<-\frac{1}{10}\dot u(s).
\]
Integrating from time $s_0$ to $s$, we derive that $u$ is uniformly bounded. From the relation $\dot v< \frac{1}{10}\dot u$, we derive that $v$ is also uniformly bounded. It then violates Lemma \ref{lem:cont-crit}.

\pfstep{Step~3}  Given $t_0$ satisfying \eqref{eq:dotu:dotv:s0}, we now show that
\begin{equation}\label{eq:dotu:dotv}
\frac{1}{10}\dot u(s)\leq \dot v(s),\quad \forall s\in[t_0, s_f).
\end{equation}
Define:
\[
s^*=\sup\{ s : \frac{1}{10}\dot u(s')\leq \dot v(s'), \quad \forall s'\in [t_0, s] \}.
\]
If $s^*=s_f$, then \eqref{eq:dotu:dotv} holds. Otherwise by continuity, $\frac{1}{10}\dot u(s^*)\leq \dot v(s^*)$. Since $\dot r(s) \leq 0$ (by Step 1), we have $\dot v(s^*)\leq 10 \dot u(s^*)$. Then 
from the identity \eqref{eq:dudv-dnull:al}, we see that $\dot u(s^*)$, $\dot v(s^*)$ satisfy the conditions in Lemma \ref{lem:uvr:sign}. In particular, $\ddot v(s^*)>0$, $\ddot u(s^*)<0$. Therefore there exists $t_1 >s^*$ such that
\[
\frac{1}{10}\dot u(s)\leq \frac{1}{10}\dot u(s^*)\leq \dot v(s^*)\leq \dot v(s),\quad \forall s\in[s^*, t_1].
\]
This contradicts the definition of $s^*$. Hence the inequality \eqref{eq:dotu:dotv} holds. 

\pfstep{Step~4} The argument above using Lemma \ref{lem:uvr:sign} also implies that $\ddot u(s)<0$, $\dot v(s)\leq 10 \dot u(s)$ for all $s\in [t_0, s_f)$. In particular both $\dot u$ and
 $\dot v$ are uniformly bounded. Hence the geodesic $\gamma(s)$ lies in a compact set. This contradicts Lemma \ref{lem:cont-crit} and thus concludes the proof of the lemma.
\end{proof}

Lemmas \ref{lem:leave:al} and \ref{lem:exit} together show that as $s\to s_f$, the geodesic must enter and leave the cylinder $\{r=r_0\}$ infinitely many times. However, this will be prohibited by the following lemma:
\begin{lemma} \label{cor:leave:al}
Assume $\gamma(s): [0, s_f)\mapsto \calM$ is incomplete and $r(s)>0$ for all $s\in [0,s_f)$. Suppose there exists a sequence $\{s_n\}$ with $s_n\rightarrow s_f$ such that $\dot r(s_n)=0$. Then $\lim\limits_{n\rightarrow \infty}(\dot u\dot v)(s_n)=\infty$.
\end{lemma}
\begin{proof}
By \eqref{rd.uv}, $\dot{r}(s_n)=0$ implies that $C^{-1} \dot{u}(s_{n}) \leq \dot{v}(s_{n}) \leq C\dot{u}(s_{n})$ for some constant $C>0$ independent of $n$. This therefore implies (by \eqref{E.uv}) $E(s_n) \aleq \sqrt{(\dot u\dot v)}(s_n)$. The conclusion follows from Lemma \ref{lem:cont-crit-e:al}.
\end{proof}



\begin{proof}[Proof of future geodesic incompleteness]
We are now ready to obtain a contradiction by assuming that $\gamma(s)$ is an incomplete future pointing causal geodesic. By Lemma \ref{lem:J:al}, we can assume $\bfJ\neq 0$. Take $r_0$ sufficiently small so that Lemma 
\ref{lem:exit} holds. Consider the cylinder $\{r=r_0\}$. Lemma \ref{lem:leave:al} together with Lemma \ref{lem:exit} imply that $\gamma(s)$ intersects $\{r=r_0\}$ infinitely many times. 
In particular we can find a sequence $s_n\rightarrow s_f$ such that $\dot r(s_n)=0$ but $r(s_n)\geq r_0$. Therefore the identity \eqref{eq:dudv-dnull:al} shows that $\dot u(s_n)\dot v(s_n)$ are uniformly bounded which contradicts Lemma~\ref{cor:leave:al}.
\end{proof}

This concludes the proof of the future causal geodesic completeness in Theorems~\ref{thm:main.finite} and \ref{thm:main}. The remaining past causal geodesic completeness statement in Theorems~\ref{thm:main} can be proved in a completely identical manner after reversing the time-orientation and noticing that we have very similar bounds for the scalar field, the metric components and their derivatives.

\appendix
\section{Nonlinear wave equation with seventh-order nonlinearity}\label{appendix}
In this appendix, we consider the equation
\begin{equation} \label{eq:nlw} \tag{NLW}
	\Box_{\bbR^{1+3}} \phi = \pm \phi^{7}.
\end{equation}
The nonlinearity is said to be \emph{defocusing} if the sign on the right-hand side is plus, and \emph{focusing} if the sign is minus. The critical Sobolev space is $\dot{H}^{\frac{7}{6}} \times \dot{H}^{\frac{1}{6}}$, and hence \eqref{eq:NLW} is called energy supercritical.

Our aim in this appendix is to apply the techniques developed in this paper to construct a solution with infinite critical Sobolev norm that exists globally in the future and the past. A more precise statement is as follows.
\begin{theorem} \label{thm:nlw}
Consider either the focusing or the defocusing \eqref{eq:nlw}.
There exists a smooth solution $\phi$ to the equation \eqref{eq:nlw} (in both the focusing and defocusing cases) which exists globally on $\bbR^{1+3}$ and has infinite $\dot{H}^{\frac{7}{6}} \times \dot{H}^{\frac{1}{6}}$ norm on each constant $t$ hypersurface, i.e.,
\begin{equation} \label{eq:nlw-infinite-Hs}
	\nrm{(\phi, \rd_{t} \phi) (t, x)}_{\dot{H}^{\frac{7}{6}}_{x} \times \dot{H}^{\frac{1}{6}}_{x}} = \infty \quad \hbox{ for every } t \in \bbR.
\end{equation}
Moreover, the space-time $L^{12}$ norm, which is a scale invariant Strichartz norm, is finite on every bounded time interval, i.e.,
\begin{equation} \label{eq:eq:nlw-finite-Str}
	\nrm{\phi}_{L^{12}([-T, T] \times \bbR^{3})} < \infty \quad \hbox{ for all } T > 0.
\end{equation}
However, the space-time $L^{12}$ norm  is infinite towards the future and the past, i.e.,
\begin{equation} \label{eq:nlw-infinite-Str}
	\nrm{\phi}_{L^{12}([0, \infty) \times \bbR^{3})} = \infty, \quad
	\nrm{\phi}_{L^{12}((-\infty, 0] \times \bbR^{3})} = \infty.
\end{equation}
\end{theorem}

\begin{remark}
More precisely, by \eqref{eq:nlw-infinite-Hs}, we mean that $(\phi, \rd_{t} \phi)(t, x)$ does not belong to the space $\dot{H}^{\frac{7}{6}}_{x} \times \dot{H}^{\frac{1}{6}}_{x}$, which in turn is defined to as the completion of $\calS \times \calS$ under the $\dot{H}^{\frac{7}{6}}_{x} \times \dot{H}^{\frac{1}{6}}_{x}$ norm defined in \eqref{eq:frac-Hs}.
\end{remark}

\begin{remark} [Comparison with Krieger-Schlag \cite{KrSc} and Beceanu-Soffer \cite{BeSo}]\label{rem:nlw}
A brief comparison of Theorem~\ref{thm:nlw} with the previous works \cite{KrSc} and \cite{BeSo} is in order.
Both papers, among other results, established the forward-in-time global existence of solutions to \eqref{eq:NLW} arising from a class of initial data with infinite $\dot{H}^{\f76}\times \dot{H}^{\f 16}$ norms. We emphasize that of course the solutions we construct are in a different regime as that of \cite{KrSc} and \cite{BeSo}. On the one hand, unlike in \cite{KrSc}, our solutions are not close to any self-similar solutions. On the other hand, as opposed to the scattering solutions in \cite{BeSo}, our solutions do not scatter; see \eqref{eq:nlw-infinite-Str}.
Moreover, we do not obtain the statement in \cite{KrSc} and \cite{BeSo} that there are some subclass of solutions with large $L^\infty$ norms (in the defocusing case for \cite{KrSc}). Finally, we mention that the solutions that we construct are manifestly stable in some sufficiently regular topology in spherical symmetry but we do not obtain stability in $\dot{H}^{\f76}\times \dot{H}^{\f 16}$ as in \cite{KrSc} and \cite{BeSo}.
\end{remark}

In Section~\ref{subsec:nlw-apriori}, we first prove analogues of Theorems~\ref{thm:main.finite} and \ref{thm:main} for \eqref{eq:NLW}, which are the main tools for our proof of Theorem~\ref{thm:nlw}. Then in Sections~\ref{subsec:nlw-id} and \ref{subsec:nlw-pf}, we construct an initial data set $\Phi$ at the past null infinity and show that it gives rise to a global solution with properties claimed in Theorem~\ref{thm:nlw}.

Due to the simplicity of the nonlinearity, the proof of the existence theorems in this case is considerably simpler than for the Einstein-scalar-field system (see the proof of Proposition \ref{prop:nlw-main.finite}). Most of the work in this appendix in fact goes into verifying that the critical norms are infinite, i.e., \eqref{eq:nlw-infinite-Hs} and \eqref{eq:nlw-infinite-Str}.

\subsection{Main existence statements} \label{subsec:nlw-apriori}
For the sake of concreteness, we fix the sign in \eqref{eq:nlw} to be $-$; it will however be clear that our argument does not depend on this sign.

As in the case of \eqref{eq:SSESF}, we work with spherically symmetric solutions to \eqref{eq:nlw}. Using the double null coordinates $(u, v)$ defined\footnote{We note that the null variables $u$ and $v$ are normalized in slightly differently from Theorem~\ref{thm:main.finite} to simplify the constants in the expressions.} by the formula $(t, r) = (v+u, v-u)$,
the equation reduces to
\begin{equation} \label{eq:nlw:null}
	\rd_{u} \rd_{v} (r \phi) = r \phi^{7}.
\end{equation}
As before, $\phi$ can be recovered from $\rd_{v}(r \phi)$ by the averaging formula
\begin{equation} \label{eq:nlw:avg:phi}
	\phi(u, v) = \frac{1}{r} \int_{u}^{v} \rd_{v}(r \phi) (u, v') \, \ud v'.
\end{equation}
We remark that $(u, v)$ was chosen so that $\rd_{v} r = - \rd_{u} r = 1$.

For $k = 0, 1, \ldots$, we say that $\phi$ is a \emph{(spherically symmetric) $C^{k}$ solution to \eqref{eq:nlw} on $\PD_{[u_{0}, u_{1}]}$} if it obeys \eqref{eq:nlw:null} and $\phi, \rd_{v}(r \phi)$ and $\rd_{u}(r \phi)$ are $C^{k}$ on $\PD_{[u_{0}, u_{1}]}$. If these conditions hold for every $u_{1}$ which is larger than $u_{0}$, we say that $\phi$ is a \emph{global $C^{k}$ solution} on $\PD_{[u_{0}, \infty)}$.

Consider the characteristic initial value problem from an outgoing curve $C_{u_{0}}$, where we prescribe
\begin{equation*}
\rd_{v}(r \phi)(u_{0}, v) = \Phi(v).
\end{equation*}
By a routine iteration argument employing integration along characteristics, it follows that \eqref{eq:nlw} is locally well-posed for any $\Phi \in C^{k}$ with $k \geq 0$, i.e., there exists a unique $C^{k}$ solution to \eqref{eq:nlw} on $\PD_{[u_{0}, u_{1}]}$ with the given data, where $u_{1} > u_{0}$ depends only on the $C^{k}$ norm of $\Phi$.

The analogue of Theorem~\ref{thm:main.finite} for \eqref{eq:nlw} reads as follows.
\begin{proposition} \label{prop:nlw-main.finite}
Consider the characteristic initial value problem from an outgoing curve $C_{u_{0}}$ with data $\Phi$.
Suppose that the following condition holds for some $0 \leq \dc < \frac{2}{3}$:
\begin{equation} \label{eq:IDcond:nlw}
	\int_{u}^{v} \abs{\Phi(v')} \, \ud v' \leq \eps (v - u)^{\frac{2}{3} - \dc}, \quad \sup \abs{\Phi} \leq \eps.
\end{equation}
Then there exists $\eps_{1} > 0$ depending only on $\dc$ such that if $\eps \leq \eps_{1}$, then the data above give rise to a unique global $C^{0}$ solution $\phi$ to \eqref{eq:nlw} on $\PD_{[u_{0}, \infty)}$, which obeys the following bounds.
\begin{align}
\label{eq:nlw:est4dvrphi}
	\abs{\rd_{v} (r \phi)(u,v) - \Phi(v)} \aleq & \eps^{7} r_{+}^{-\frac{1}{3}-7 \dc}, \\
\label{eq:nlw:est4phi}
		\abs{\phi(u,v)} \aleq & \eps r_{+}^{-\frac{1}{3}-\dc} \\
\label{eq:nlw:est4durphi}
	\abs{\rd_{u} (r \phi) (u, v) } \aleq & \eps.
\end{align}

Suppose furthermore that $\Phi \in C^{1}$. Then for every $v_{0} > u_{0}$, we have
\begin{align}
\label{eq:nlw:est4dvdvrphi}
	\sup_{\calD(u_{0}, v_{0})} \abs{\rd_{v}^{2}(r \phi)} \aleq & \sup_{v \in [u_{0}, v_{0}]} \abs{\Phi'} + \eps^{7}, \\
\label{eq:nlw:est4dudurphi}
	\sup_{\calD(u_{0}, v_{0})} \abs{\rd_{u}^{2}(r \phi)} \aleq & \sup_{v \in [u_{0}, v_{0}]} \abs{\Phi'} + \eps^{7}.
\end{align}
\end{proposition}

\begin{proof}
As in the proof of Theorem~\ref{thm:main.finite}, we begin by performing a bootstrap argument with the bound
 \begin{equation} \label{eq:nlw:btstrp}
	\int_{u_{0}}^{u} \abs{r \phi^{7}(u', v)} \, \ud u'
	\leq 2 \eps r_{+}^{-\frac{1}{3} -7\dc}.
\end{equation}
Indeed, assume that \eqref{eq:nlw:btstrp} holds on $\PD_{[u_{0}, u_{1}]}$ for some $u_{0} < u_{1}$. Then by \eqref{eq:nlw:null} and \eqref{eq:nlw:avg:phi}, we obtain
\begin{equation*}
	\abs{\rd_{v}(r \phi)(u,v) - \Phi(v)} \leq 2 \eps r_{+}^{-\frac{1}{3}-7 \dc}, \quad
	\abs{\phi} \aleq \eps (r_{+}^{-\frac{1}{3} -\dc} + r_{+}^{-\frac{1}{3} - 7 \dc}) \aleq \eps r_{+}^{-\frac{1}{3} - \dc},
\end{equation*}
on the same domain. Hence, the following pointwise bound for the nonlinearity holds:
\begin{equation} \label{eq:nlw:nonlin}
	\abs{r \phi^{7}}
	\aleq \eps^{7} r_{+}^{-\frac{4}{3} - 7 \dc}
\end{equation}
Integrating\footnote{Here, we use an analogue of \eqref{eq:est4noverr:barCv}, which in the semilinear setting here, is very easy to obtain.} \eqref{eq:nlw:nonlin} in the incoming direction from $u_{0}$ to $u$, we obtain an improvement of the bootstrap assumption \eqref{eq:nlw:btstrp} for $\eps$ sufficiently small. Then by a routine continuity argument using $C^{0}$ local well-posedness\footnote{We recall again our convention that a $C^0$ solutions means that all of $\phi$, $\rd_v(r\phi)$ and $\rd_u(r\phi)$ are in $C^0$.} of \eqref{eq:nlw:null}, global existence of $\phi$  follows. Moreover, by \eqref{eq:nlw:null} and \eqref{eq:nlw:nonlin}, the bounds \eqref{eq:nlw:est4dvrphi}--\eqref{eq:nlw:est4durphi} follow.

Now assume that $\Phi \in C^{1}$, and let $v_{0} > u_{0}$. To prove \eqref{eq:nlw:est4dvdvrphi}, it suffices to show that
\begin{equation*}
	\abs{\rd_{v} \int_{u_{0}}^{u} r \phi^7(u', v) \, \ud u'} \aleq \eps^{7}
\end{equation*}
for $(u, v) \in \calD(u_{0}, v_{0}):=\{(u,v)\in \mathbb R^2: u\in [u_0,v_0],\, v\in [u,v_0]\}$. This estimate follows from \eqref{eq:nlw:est4dvrphi}, \eqref{eq:nlw:est4phi} and the simple identity
\begin{equation*}
	r \rd_{v} \phi = \rd_{v}(r \phi) - \phi.
\end{equation*}
To show \eqref{eq:nlw:est4dudurphi}, note that at the axis of symmetry $\set{u = v}$, we have
\begin{equation*}
	(\rd_{v} + \rd_{u})^{2} (r \phi) (u, u) = 0.
\end{equation*}
By \eqref{eq:nlw:null} the mixed derivative $\rd_{u} \rd_{v} (r \phi)$ vanishes on the axis; hence we have
\begin{equation*}
	\abs{\rd_{u}^{2} (r \phi)(u, u)} = \abs{\rd_{v}^{2} (r \phi)(u, u)} \aleq \sup_{v \in [u_{0}, v_{0}]} \abs{\Phi'} + \eps^{7}.
\end{equation*}
Then \eqref{eq:nlw:est4dudurphi} follows from the estimate
\begin{equation*}
	\abs{\rd_{u} \int_{u}^{v} r \phi^{7}(u, v') \, \ud v'} \aleq \eps^{7}
\end{equation*}
for $(u, v) \in \calD(u_{0}, v_{0})$, which is proved using \eqref{eq:nlw:est4phi} and \eqref{eq:nlw:est4durphi} as before. \qedhere
\end{proof}

Taking the limit as $u_{0} \to - \infty$, we obtain the following analogue of Theorem~\ref{thm:main}.
\begin{proposition} \label{prop:nlw-main}
Let $\Phi : (-\infty, \infty) \to \bbR$ be a $C^{2}$ function satisfying \eqref{eq:IDcond:nlw}, as well as
\begin{equation*}
	A = \sup_{v \in (-\infty, \infty)} \abs{\Phi'(v)} < \infty.
\end{equation*}
Let $\eps_{1} > 0$ be the constant introduced in Proposition~\ref{prop:nlw-main.finite}. Then if $\eps \leq \eps_{1}$, there exists a unique global $C^{0}$ solution $\phi$ to \eqref{eq:nlw} on $\PD$, whose data at the past null infinity coincide with $\Phi$, i.e.,
\begin{equation}\label{nlw:data.past.null.inf}
	\lim_{u \to -\infty} \rd_{v}(r \phi)(u, v) = \Phi(v) \quad \hbox{ for every } v \in (-\infty, \infty).
\end{equation}
The solution $\phi$ obeys the bounds \eqref{eq:nlw:est4dvrphi}--\eqref{eq:nlw:est4durphi}. Furthermore, $\rd_{v}(r \phi)$ and $\rd_{u}(r \phi)$ are Lipschitz continuous on $\PD$, and the weak derivatives $\rd_{v}^{2}(r \phi)$, $\rd_{u}^{2}(r \phi)$ obey the bounds
\begin{equation}
\label{eq:nlw:lipschitz}
\esssup_{\PD} \abs{\rd_{v}^{2}(r \phi)}
+ \esssup_{\PD} \abs{\rd_{u}^{2}(r \phi)} \aleq A + \eps^{7}.
\end{equation}
\end{proposition}
As in the proof of Theorem~\ref{thm:main}, this proposition is a simple consequence of the uniform $C^{1}$ bounds \eqref{eq:nlw:est4dvdvrphi}, \eqref{eq:nlw:est4dudurphi}, the Arzela-Ascoli theorem, as well as the bound \eqref{eq:nlw:est4phi} to justify that data are as prescribed at the past null infinity. We omit the details.

\subsection{Initial data construction} \label{subsec:nlw-id}
The goal of this subsection is to construct an initial data set $\Phi(v)$ on the past null infinity, so that the free wave development of $\Phi(v)$ has infinite $\dot{H}^{\frac{7}{6}} \times \dot{H}^{\frac{1}{6}}$ norm on the slice $\Sgm_{0} = \set{t = 0}$ and $\Phi(v)$ obeys \eqref{eq:IDcond:nlw}. In Section~\ref{subsec:nlw-pf}, we will show that this initial data set leads to a global solution with properties stated in Theorem~\ref{thm:nlw}. (One can compare this construction with that in the proof of Corollary \ref{cor.infinite.BV.mass}.)

We begin with a few preliminary facts about fractional Sobolev spaces on $\bbR^{d}$.
Let $\calS(\bbR^{d})$ and $\calS'(\bbR^{d})$ be the spaces of Schwartz test functions and tempered distributions on $\bbR^{d}$, respectively.
For $0 < s < \frac{d}{2}$, we define $\dot{H}^{s}(\bbR^{d}) \subseteq \calS'(\bbR^{d})$ to be the closure of the space $\calS(\bbR^{d})$ with respect to the norm
\begin{equation} \label{eq:frac-Hs}
	\nrm{f}_{\dot{H}^{s}}
	= \nrm{\abs{\nb}^{s} f}_{L^{2}},
\end{equation}
where $\abs{\nb}^{s} = (- \lap)^{\frac{s}{2}}$ is the fractional Laplacian. For $s \in (0, 2)$, this operator admits the integral formula
\begin{equation} \label{eq:frac-lap}
	\abs{\nb}^{s} f (x)
	= c_{d, s} \int_{\bbR^{d}} \frac{f(x) - f(y)}{\abs{x-y}^{d + s}} \, \ud y,
\end{equation}
for an appropriate constant $c_{d, s} \neq 0$. For $s \in (1, 2)$, we have the equivalence
\begin{equation*}
	\nrm{f}_{\dot{H}^{s}} \aleq \nrm{\nb f}_{\dot{H}^{s-1}} \aleq \nrm{f}_{\dot{H}^{s}}.
\end{equation*}
If $f \in \dot{H}^{s}(\bbR^{d})$, then it follows that
\begin{equation*}
	\chi_{R} f \to f \quad \hbox{ in } \dot{H}^{s} \hbox{ as } R \to \infty,
\end{equation*}
where $\chi_{R} (\cdot) = \chi(\cdot / R)$ for any $\chi \in C^{\infty}_{0}(\bbR^{d})$ with $\chi(0) = 1$. Hence in order to show that a tempered distribution $f$ does \emph{not} belong to $\dot{H}^{s}$, it suffices to show that $\nrm{\chi_{R} f}_{\dot{H}^{s}}$ diverges as $R \to \infty$.

We now begin the construction of $\Phi$ in earnest. Our idea is to start with a function with the desired property on $\set{t = 0}$, and then find a compatible $\Phi$. Let $\eta$ be a smooth bump function on $(-\infty, \infty)$, which is non-negative, vanishes outside $(-2, 0)$, equals $1$ on $(-\frac{3}{2}, -\frac{1}{2})$. 
For every $R \geq 4$, we define a radial function $\eta_{R}$ on $\bbR^{3}$ by the formula
\begin{equation*}
	\eta_{R}(r) = \frac{1}{(4 \pi)^{\frac{1}{2} }r} \eta(r - R).
\end{equation*}
Note that $\eta_{R}$ is supported on the annulus $\set{R - 2 < r < R}$ with $\nrm{\eta_{R}}_{L^{2}(\bbR^{3})}$ equal to a nonzero constant independent of $R$. Furthermore, for every $k = 0, 1, 2, \ldots$, there exist $0 < b_{k} < B_{k}$ independent of $R$ such that
\begin{equation} \label{eq:nlw:eta-bnd}
	b_{k} \leq \nrm{\rd_{r}^{k} \eta_{R}}_{L^{2}(\bbR^{3})} \leq B_{k}.
\end{equation}
In particular, by interpolation, there exist constants $0 < b < B$ independent of $R$ such that
\begin{equation} \label{eq:nlw:eta-7/6}
	b \leq \nrm{\rd_{r} \eta_{R}}_{\dot{H}^{\frac{1}{6}}(\bbR^{3})} \leq B.
\end{equation}

Given $\eps > 0$, we define a radial function $f$ on $\bbR^{3}$ by
\begin{equation} \label{eq:nlw:id-t=0}
	f(r) = \eps \sum_{k=1}^{\infty} \eta_{4^{k}}(r).
\end{equation}

\begin{lemma} \label{lem:nlw:id-infinite-Hs}
For $f$ defined as above, we have
\begin{equation} \label{eq:nlw:id-infinite-Hs}
	\nrm{\rd_{r} f}_{\dot{H}^{\frac{1}{6}}} = \infty,
\end{equation}
or more precisely, $\nrm{\chi_{R} \rd_{r} f}_{\dot{H}^{\frac{1}{6}}} \to \infty$ as $R \to \infty$.
\end{lemma}

A key ingredient for the proof is the following localization lemma for the fractional Laplacian.
\begin{lemma} \label{lem:frac-lap-loc}
Let $d$ be a positive integer and $0 < s < \min \set{\frac{d}{2}, 2}$. Let $\psi$ be a smooth function supported on a dyadic annulus $\set{x \in \bbR^{d} :  2^{k-1} < \abs{x} < 2^{k}}$ for some $k \in \bbZ$. Then for any integer $\ell \not \in [k-1, k+1]$, we have
\begin{equation*}
	\nrm{\abs{\nb}^{s} \psi}_{L^{2}(\set{2^{\ell-1} < \abs{x} < 2^{\ell}})}
	\aleq 2^{- s (\max \set{k, \ell})} \nrm{\psi}_{L^{2}}.
\end{equation*}
where the implicit constant depends only on $d$ and $s$.
\end{lemma}

\begin{proof}
 For concreteness, we only consider the case $\ell \geq k+1$; the case $\ell \leq k-1$ can be handled analogously. Let $x \in \set{2^{\ell -1} < \abs{x} < 2^{\ell}}$. Recall the integral formula \eqref{eq:frac-lap}; since $\psi(x) = 0$ and $\supp \psi \subseteq \set{2^{k-1} < \abs{x} < 2^{k}}$, we have
 \begin{align*}
	\abs{\abs{\nb}^{s} \psi(x)}
	= \abs{c_{d, s} \int_{\set{2^{k-1} < \abs{y} < 2^{k}}} \frac{- \psi(y)}{\abs{x-y}^{d+s}} \, \ud x}
	\aleq 2^{-\ell(d+s)} \int_{\set{2^{k-1} < \abs{y} < 2^{k}}} \abs{\psi(y)} \, \ud y.
\end{align*}
By Cauchy-Schwarz, the right-hand side is bounded by $\aleq 2^{-\ell(\frac{d}{2}+s)} \nrm{\psi}_{L^{2}}$. Then taking the $L^{2}$ norm over $\set{2^{k-1} < \abs{x} < 2^{k}}$, the lemma follows. \qedhere
\end{proof}

\begin{proof} [Proof of Lemma~\ref{lem:nlw:id-infinite-Hs}]
For $K \geq1$, let
\begin{equation*}
	f_{K} (r) = \eps \sum_{k=1}^{K} \eta_{4^{k}}(r).
\end{equation*}
By the support property of $\eta_{4^{k}}$, the desired statement \eqref{eq:nlw:id-infinite-Hs} would follow once we etablish
\begin{equation} \label{eq:nlw:id-infinite-Hs:key}
	\nrm{\rd_{r} f_{K}}_{\dot{H}^{\frac{1}{6}}} \ageq \eps K^{1/2},
\end{equation}
where the implicit constant is independent of $K$. Observe furthermore that it is enough to prove \eqref{eq:nlw:id-infinite-Hs:key} for only sufficiently large $K$.

Expanding $\nrm{\rd_{r} f_{K}}_{\dot{H}^{\frac{1}{6}}}^{2}$, we have
\begin{align*}
	\nrm{\rd_{r} f_{K}}_{\dot{H}^{\frac{1}{6}}}^{2}
	= \eps^{2} \sum_{k=1}^{K} \nrm{\rd_{r} \eta_{4^{k}}}_{\dot{H}^{\frac{1}{6}}}^{2}
	+ 2 \eps^{2} \sum_{1 \leq k < \ell \leq K} \brk{\abs{\nb}^{\frac{1}{6}} \rd_{r} \eta_{4^{k}} , \abs{\nb}^{\frac{1}{6}} \rd_{r} \eta_{4^{\ell}}}_{L^{2}}.
\end{align*}
For the diagonal terms,  we have a lower bound
\begin{equation} \label{eq:nlw:id-infinite-Hs:diag}
\eps^{2} \sum_{k=1}^{K} \nrm{\rd_{r} \eta_{4^{k}}}_{\dot{H}^{\frac{1}{6}}}^{2}
\geq b^{2} \eps^{2} K,
\end{equation}
by \eqref{eq:nlw:eta-7/6}. For the cross terms, we first estimate
\begin{align*}
	\abs{\brk{\abs{\nb}^{\frac{1}{6}} \rd_{r} \eta_{4^{k}} , \abs{\nb}^{\frac{1}{6}} \rd_{r} \eta_{4^{\ell}}}_{L^{2}}}
	\leq \sum_{j} \int_{\set{2^{j-1} < \abs{y} < 2^{j}}} \abs{\abs{\nb}^{\frac{1}{6}} \rd_{r} \eta_{4^{k}} \abs{\nb}^{\frac{1}{6}} \rd_{r}\eta_{4^{\ell}}} \, \ud y.
\end{align*}
Note that
\begin{equation*}
\supp \eta_{4^{k}} \subseteq \set{2^{2k-1} < \abs{x} < 2^{2k}}, \quad
\supp \eta_{4^{\ell}} \subseteq \set{2^{2\ell-1} < \abs{x} < 2^{2\ell}},
\end{equation*}
where $2^{2k} < 2^{2\ell-1} < 2^{2\ell}$. Hence we may apply Lemma~\ref{lem:frac-lap-loc} to $\abs{\nb}^{\frac{1}{6}} \rd_{r} \eta_{4^{k}}$ when $2^{j} \geq 2^{2\ell-1}$, and to $\abs{\nb}^{\frac{1}{6}} \rd_{r} \eta_{4^{\ell}}$ when $2^{j} < 2^{2\ell-1}$. Using the upper bounds in \eqref{eq:nlw:eta-bnd} and \eqref{eq:nlw:eta-7/6}, we obtain
\begin{equation*}
	\abs{\brk{\abs{\nb}^{\frac{1}{6}} \rd_{r} \eta_{4^{k}} , \abs{\nb}^{\frac{1}{6}} \rd_{r} \eta_{4^{\ell}}}_{L^{2}}}
	\aleq 2^{-\frac{1}{3} \ell} \eps^{2} B B_{1}.
\end{equation*}
Summing up this bound, we obtain
\begin{align*}
2 \eps^{2} \abs{\sum_{1 \leq k < \ell \leq K} \brk{\abs{\nb}^{\frac{1}{6}} \rd_{r} \eta_{4^{k}} , \abs{\nb}^{\frac{1}{6}} \rd_{r} \eta_{4^{\ell}}}_{L^{2}}}
\aleq \eps^{2},
\end{align*}
where the implicit constant depends only on $B$ and $B_{1}$. Comparing this upper bound with the uniform lower bound \eqref{eq:nlw:id-infinite-Hs:diag}, the desired bound \eqref{eq:nlw:id-infinite-Hs:key} follows for sufficiently large $K$. \qedhere
\end{proof}

We now seek an initial data set $\Phi$ on the past null infinity whose free wave development restricts to $f$ on the slice $\Sgm_{0} = \set{t = 0} = \set{u + v = 0}$.
This condition is equivalent to
\begin{equation} \label{eq:phi-f}
	\frac{1}{r} \int_{-r/2}^{r/2} \Phi(v') \, \ud v' = f(r).
\end{equation}
We will furthermore require $\Phi(v)$ to be even with respect to $v = 0$, i.e.,
\begin{equation*}
	\Phi(-v) = \Phi(v).
\end{equation*}
This implies that the time derivative of the free wave development restricts to zero on $\Sgm_{0}$.
By the evenness condition, \eqref{eq:phi-f} can be achieved by defining
\begin{equation*}
	\Phi(r) = \frac{\ud}{\ud r} (r f(r)) \quad \hbox{ for } r \geq 0.
\end{equation*}
Recalling the definition of $f$, we see that $\Phi$ is given by
\begin{equation} \label{eq:nlw:Phi}
\Phi(v) = \frac{\eps}{(4 \pi)^{1/2}} \sum_{k = 1}^{\infty} \bb( \eta'(v - 4^{k}) + \eta'(-v - 4^{k}) \bb).
\end{equation}
With the preceding formulae, it can be readily checked that $\Phi$ obeys the hypothesis \eqref{eq:IDcond:nlw} of Proposition~\ref{prop:nlw-main}.
\begin{lemma} \label{lem:nlw:Phi}
Let $\Phi$ be defined as above. Then we have
\begin{equation*}
	\int_{u}^{v} \abs{\Phi(v')} \, \ud v' \aleq \eps \min\{\log (2 + (v - u)), (v-u)\},  \quad
	\sup \abs{\Phi} \aleq \eps,  \quad
	\sup \abs{\Phi'} \aleq \eps.
\end{equation*}
\end{lemma}
\begin{proof}
 The latter two bounds are immediate from the formula \eqref{eq:nlw:Phi} for $\Phi$.
 For the first bound, due to the exponential separation of the bumps in \eqref{eq:nlw:Phi}, observe that the number of bumps $\eta'(\cdot - 4^{k})$ whose support intersects the interval $[u, v]$ is bounded by $\aleq \log (v - u)$ if $v - u \geq 2$. The estimate is obvious when $v-u<2$. \qedhere
\end{proof}

\subsection{Completion of the proof of Theorem~\ref{thm:nlw}} \label{subsec:nlw-pf}
By Lemma~\ref{lem:nlw:Phi}, $\Phi$ satisfies the hypothesis of Proposition~\ref{prop:nlw-main} with any $0 \leq \dc < \frac{2}{3}$. Hence, taking $\eps$ sufficiently small, we may apply Proposition~\ref{prop:nlw-main} to construct a global solution $\phi$ with $\Phi$ as data at the past null infinity (in the sense of \eqref{nlw:data.past.null.inf}), which moreover obeys the bounds \eqref{eq:nlw:est4dvrphi}--\eqref{eq:nlw:est4durphi} and \eqref{eq:nlw:lipschitz}. Our goal is to show that this solution possesses the properties listed in Theorem~\ref{thm:nlw}.

We decompose $\phi = \phi^{hom} + \phi^{inhom}$, where
\begin{align*}
	\phi^{hom} (u, v) =& \frac{1}{r} \int_{u}^{v} \Phi(v') \, \ud v', \\
	\phi^{inhom} (u, v) =& \frac{1}{r} \int_{u}^{v} \int_{-\infty}^{u} r \phi^{7} (u', v')  \, \ud u' \, \ud v',
\end{align*}
Since $\phi^{hom} \restriction_{\Sgm_{0}} = f$ and $\Phi$ is even with respect to $v = 0$, we have
\begin{equation} \label{eq:nlw:dv-phi-hom:t=0}
\rd_{v} \phi^{hom} ( -\tfrac{r}{2}, \tfrac{r}{2})
= f'(r).
\end{equation}
Hence Lemma~\ref{lem:nlw:id-infinite-Hs} applies to $\rd_{v} \phi^{hom}$. On the other hand, the following improved estimates hold for $\rd_{v} \phi^{inhom}$ on $\Sgm_{0}$.
\begin{lemma} \label{lem:nlw:phi-inhom}
Let $\phi^{inhom}$ be defined as above. Then we have
\begin{equation} \label{eq:nlw:est4phi-inhom}
\nrm{\rd_{v} \phi^{inhom}}_{L^{2}(\Sgm_{0})} + \nrm{\rd_{u} \rd_{v} \phi^{inhom}}_{L^{2}(\Sgm_{0})} + \nrm{\rd_{v}^{2} \phi^{inhom}}_{L^{2}(\Sgm_{0})} \aleq \eps^{7}.
\end{equation}
\end{lemma}


\begin{proof}
We need to show that
\begin{align}
\label{eq:nlw:est4phi-inhom:1}
\int_{0}^{\infty} r^{2} (\rd_{v} \phi^{inhom})^{2} (- \tfrac{r}{2}, \tfrac{r}{2} ) \, \ud r \aleq & \eps^{7}, \\
\label{eq:nlw:est4phi-inhom:2}
\int_{0}^{\infty} r^{2} (\rd_{u} \rd_{v} \phi^{inhom})^{2} (- \tfrac{r}{2}, \tfrac{r}{2} ) \, \ud r \aleq & \eps^{7}, \\
\label{eq:nlw:est4phi-inhom:3}
\int_{0}^{\infty} r^{2} (\rd_{v}^{2} \phi^{inhom})^{2} (- \tfrac{r}{2}, \tfrac{r}{2} ) \, \ud r \aleq & \eps^{7}.
\end{align}

The plan is to establish pointwise estimates for successively higher derivatives of $\phi^{inhom}$, from which \eqref{eq:nlw:est4phi-inhom:1}--\eqref{eq:nlw:est4phi-inhom:3} follow. Fix a number $0 < \dc < \frac{2}{3}$ sufficiently close to $\frac{2}{3}$; in the remainder of the proof, all implicit constants may depend on $\dc$.

By the definition of $\phi^{inhom}$ and \eqref{eq:nlw:est4phi}, we have
\begin{equation} \label{eq:nlw:est4phi-inhom:basic}
	\abs{\rd_{v}(r \phi^{inhom})} \aleq \eps^{7} r_{+}^{-\frac{1}{3} - 7 \dc}, \quad
	\abs{\rd_{u}(r \phi^{inhom})} \aleq \eps^{7}, \quad
	\abs{\phi^{inhom}} \aleq \eps^{7} r_{+}^{-1}.
\end{equation}
Note in particular the rapid decay of $\rd_{v}(r \phi^{inhom})$ in $r$ compared to $\rd_{v}(r \phi^{hom}) = \Phi$.
This feature will be key to the improved bounds below.

By differentiating the averaging formula using Lemma~\ref{lem:avg-d}, we obtain
\begin{align*}
	\rd_{v} \phi^{inhom}
	= \frac{1}{r^{2}} \int_{u}^{v} \int_{-\infty}^{u} \rd_{v} (r \phi^{7})(u', v') \, \ud u' \, r (u, v') \,\ud v',
\end{align*}
whereas simply commuting $r$ with $\rd_{v}$ gives
\begin{equation*}
	r \rd_{v} \phi^{inhom}
	= \rd_{v} (r \phi^{inhom}) - \phi^{inhom}.
\end{equation*}
As a consequence, we obtain the estimate
\begin{equation} \label{eq:nlw:est4dvphi-inhom}
	\abs{\rd_{v} \phi^{inhom}}
	\aleq \eps^{7} \bb( r_{+}^{-\frac{4}{3} - 7 \dc} + r_{+}^{-2} \bb),
\end{equation}
which immediately implies \eqref{eq:nlw:est4phi-inhom:1}, provided that $\gmm$ is sufficiently large (i.e., $\gmm >  1/42$).

In order to estimate $\rd_{u} \phi^{inhom}$, note that
\begin{equation*}
	\rd_{v} (r \rd_{u} \phi^{inhom}) = \rd_{u} \rd_{v}(r \phi^{inhom}) + \rd_{v} \phi^{inhom} = r \phi^{7} + \rd_{v} \phi^{inhom}.
\end{equation*}
Integrating over $v' \in [u, v]$ and using \eqref{eq:nlw:est4dvrphi} and \eqref{eq:nlw:est4dvphi-inhom}, we obtain
\begin{equation} \label{eq:nlw:est4duphi-inhom}
	\abs{\rd_{u} \phi^{inhom}}
	\aleq \eps^{7} r_{+}^{-1}.
\end{equation}

Next, by the identity
\begin{equation*}
r \rd_{u} \rd_{v} \phi^{inhom}
= \rd_{u} \rd_{v} (r \phi^{inhom}) + (\rd_{v} - \rd_{u}) \phi^{inhom}
= r \phi^{7} +(\rd_{v} - \rd_{u}) \phi^{inhom},
\end{equation*}
and estimates \eqref{eq:nlw:est4phi}, \eqref{eq:nlw:est4dvphi-inhom} and \eqref{eq:nlw:est4duphi-inhom}, we arrive at the bound
\begin{equation}
\abs{r \rd_{u} \rd_{v} \phi^{inhom}}
\aleq \eps^{7} r_{+}^{-1},
\end{equation}
from which \eqref{eq:nlw:est4phi-inhom:2} follows.

Finally, note that
\begin{align*}
r \rd_{v}^{2} \phi^{inhom} (u, v)
=& \rd_{v}^{2} (r \phi^{inhom}) (u, v) - 2 \rd_{v} \phi^{inhom}(u, v) \\
= & \int_{-\infty}^{u} \rd_{v} (r \phi^{7})(u', v) \, \ud u' - 2 \rd_{v} \phi^{inhom}(u, v).
\end{align*}
Therefore, using \eqref{eq:nlw:est4dvrphi}, \eqref{eq:nlw:est4phi} and \eqref{eq:nlw:est4dvphi-inhom}, we have
\begin{equation}
	\abs{r \rd_{v}^{2} \phi^{inhom}}
	\aleq \eps^{7} (r_{+}^{-1-7 \dc} + r_{+}^{-2}).
\end{equation}
This bound is sufficient to establish \eqref{eq:nlw:est4phi-inhom:3}, which completes the proof of the lemma. \qedhere
\end{proof}

We are now ready to establish Theorem~\ref{thm:nlw}.
\begin{proof} [Proof of Theorem~\ref{thm:nlw}]
In this proof, we employ the polar coordinates $(t, r)$ instead of the double null coordinates $(u, v)$.

As a first step, we claim that
\begin{equation} \label{eq:nlw:pf:infinite-Hs}
	\nrm{(\phi, \rd_{t} \phi)}_{\dot{H}^{\frac{7}{6}} \times \dot{H}^{\frac{1}{6}}(\Sgm_{0})} = \infty.
\end{equation}
Let $L = \rd_{t} + \rd_{r}$, which coincides with $\rd_{v}$ in the double null coordinates. By \eqref{eq:nlw:dv-phi-hom:t=0}, we have $\nrm{L \phi^{hom}}_{\dot{H}^{\frac{1}{6}}} = \infty$, whereas by Lemma~\ref{lem:nlw:phi-inhom} and interpolation it follows that
\begin{equation*}
	\nrm{L \phi^{inhom}}_{\dot{H}^{\frac{1}{6}}(\Sgm_{0})}
	\aleq \nrm{L \phi^{inhom}}_{H^{1}(\Sgm_{0})} \aleq \eps^{7} < \infty.
\end{equation*}
Hence we have proved
\begin{equation*}
	\nrm{L \phi}_{\dot{H}^{\frac{1}{6}}(\Sgm_{0})} = \infty,
\end{equation*}
which implies \eqref{eq:nlw:pf:infinite-Hs}.

Next, observe that the $L^{12}_{x}$ norm of $\phi$ on every time slice $\Sgm_{t} = \set{t = const}$ is finite by \eqref{eq:nlw:est4phi}.
Therefore, the critical Strichartz norm $\nrm{\phi}_{L^{12}((-T, T) \times \bbR^{3})}$ is finite for every $0 < T < \infty$.
It follows that
\begin{equation} \label{eq:nlw:pf:infinite-Hs:all-t}
	\nrm{(\phi, \rd_{t} \phi)}_{\dot{H}^{\frac{7}{6}} \times \dot{H}^{\frac{1}{6}}(\Sgm_{t})} = \infty \quad \hbox{ for all } t \in (-\infty, \infty).
\end{equation}
Indeed, if \eqref{eq:nlw:pf:infinite-Hs:all-t} failed for any $t$, then by finiteness of the $L^{12}_{t,x}$ norm on finite time intervals and the standard well-posedness theory for \eqref{eq:NLW}, we would contradict \eqref{eq:nlw:pf:infinite-Hs}.

To complete the proof of Theorem~\ref{thm:nlw}, it only remains to show that the critical Strichartz norm diverges towards both the future and the past.
By the explicit formula \eqref{eq:nlw:Phi} and the mean value theorem, we may find a doubly infinite increasing sequence $\set{t_{n}}_{n = -\infty}^{\infty}$ such that $\lim_{n \to -\infty} t_{n} = -\infty$, $\lim_{n \to \infty} t_{n} = \infty$ and
\begin{equation*}
	\phi^{hom} (t_{n}, 0) = \Phi \bb( \frac{t_{n}}{2} \bb) = \frac{\eps}{(4 \pi)^{1/2}}.
\end{equation*}
By uniform $C^{1}$ regularity of $\Phi$, there exists universal constants $\dlt > 0$ and $c > 0$ such that
\begin{equation*}
	\phi^{hom} (t, r) \geq c \eps \quad \hbox{ for all } (t, r) \in [t_{n}-\dlt, t_{n}+\dlt] \times [0, \dlt].
\end{equation*}
On the other hand, we have a uniform bound $\abs{\phi^{inhom}} \aleq \eps^{7}$ for $\phi^{inhom}$. Hence for $\eps > 0$ sufficiently small, we have
\begin{equation*}
	\phi (t, r) \geq \frac{c}{2} \eps \quad \hbox{ for all } (t, r) \in [t_{n}-\dlt, t_{n}+\dlt] \times [0, \dlt].
\end{equation*}
Since $t_{n}$ is doubly infinite with $t_{n} \to \pm \infty$ as $n \to \pm \infty$, it follows that the $L^{12}_{t,x}$ norm of $\phi$ (as a matter of fact, any space-time norm $L^{p}_{t} L^{r}_{x}$ with $1 \leq p < \infty$) diverges to infinity towards both the future and the past, as desired. \qedhere
\end{proof}
%




\begin{thebibliography}{1}

\bibitem{BeSo}
Marius Beceanu and Avy Soffer, \emph{{Large Outgoing Solutions to Supercritical Wave Equations}}, preprint (2016), arXiV:1601.06335

\bibitem{Christodoulou:1986}
Demetrios Christodoulou, \emph{{The problem of a self-gravitating scalar field}},
Comm. Math. Phys. 105 (1986), no. 3, 337-361.

\bibitem{Christodoulou:1991}
\bysame, \emph{{The formation of black holes and singularities in spherically symmetric gravitational collapse}},
Comm. Pure Appl. Math. 44 (1991), no. 3, 339-373.

\bibitem{Christodoulou:1993bt}
\bysame, \emph{{Bounded variation solutions of the spherically symmetric Einstein-scalar field equations}}, Comm. Pure Appl. Math. 46 (1993), no. 8, 1131-1220.

\bibitem{Chr}
\bysame, \emph{The formation of black holes in general relativity}, Monographs in Mathematics, European Mathematical Soc. (2009).

\bibitem{CK}
Demetrios Christodoulou and Sergiu Klainerman, \emph{{The global nonlinear stability of the Minkowski space}},
Princeton Mathematical Series 41 (1993).

\bibitem{DafRen}
Mihalis Dafermos and Alan D.~Rendall, \emph{{Strong cosmic censorship for surface-symmetric cosmological spacetimes with collisionless matter}}, preprint (2007), arXiV:gr-qc/0701034
\bibitem{KrSc}
Joachim Krieger and Wilhelm Schlag, \emph{{Large global solutions for energy supercritical nonlinear wave equations on $\mathbb R^{3+1}$}},
 preprint (2014), arXiV:1403.2913.

\bibitem{LR}
Hans Lindblad and Igor Rodnianski, \emph{{The global stability of Minkowski space-time in harmonic gauge}},
Annals of Math (2) 171 (2010), no. 3, 1401-1477.

\bibitem{LO1}
Jonathan Luk and Sung-Jin Oh, \emph{Quantitative decay rates for dispersive solutions to the Einstein-scalar field system in spherical symmetry}, Analysis and PDE 8 (2014), no.7, 1603-1674, arXiv:1402.2984.

\bibitem{LO2}
\bysame, \emph{Global nonlinear stability of large dispersive solutions to the Einstein equations}, in preparation (2016).

\bibitem{WaYu}
Jinhua Wang and Pin Yu, \emph{A large data regime for non-linear wave equations}, preprint (2012), arXiv:1210.2056.

\bibitem{Yang}
Shiwu Yang, \emph{{Global solutions of nonlinear wave equations with large data}}, Selecta Mathematica, online first (2014).




\end{thebibliography}
\end{document}